\documentclass[11pt]{article}
\usepackage{amsmath,amssymb}
\usepackage{amsthm}
\usepackage{bbm}
\usepackage{bm}
\usepackage{centernot}
\usepackage{changepage}
\usepackage{color}
\usepackage{enumitem}
\usepackage{framed}
\usepackage{graphicx}
\usepackage{interval}
\usepackage{mathtools}
\usepackage{multibib}
\usepackage{multirow}
\usepackage[round,longnamesfirst]{natbib}
\usepackage{rotating}
\usepackage{setspace}
\usepackage{subcaption}
\usepackage{threeparttable}
\usepackage{tikz}
\usepackage{titlesec}
\usepackage{titling}
\usepackage{authblk}
\usepackage{url}

\DeclareMathOperator*{\argmin}{\arg\!\min}

\DeclareMathOperator*{\plim}{\mathrm{p}\!\lim}
\mathtoolsset{showonlyrefs}
\newtheorem{thm}{Theorem}
\newtheorem{prop}{Proposition}

\newtheorem{lem}{Lemma}
\newtheorem{asm}{Assumption}

\newtheorem{rem}{Remark}
\theoremstyle{definition}
\newtheorem{definition}{Definition}

\newtheorem{algo}{Algorithm}

\renewenvironment{proof}[1][\proofname]{{\bfseries #1. }}{\qed}

\newcites{app}{References to Appendices}
\DeclareRobustCommand\full  {\tikz[baseline=-0.6ex]\draw[thick] (0,0)--(0.5,0);}

\DeclareRobustCommand\dashed{\tikz[baseline=-0.6ex]\draw[thick,dashed] (0,0)--(0.54,0);}

\title{K-Means Panel Data Clustering in the Presence of Small Groups}

\author[]{Mikihito Nishi\footnote{This work is supported by JSPS KAKENHI Grant number 25KJ0041. A part of this research was conducted while the author was at Hitotsubashi University. The author would like to thank participants at the 32nd Kansai Econometrics Meeting, the 19th International Symposium on Econometric Theory and Applications, and the Summer Workshop on Economic Theory 2025 for their valuable comments and suggestions. All errors are mine. Address correspondence to: Graduate School of Economics, University of Tokyo, 7-3-1 Hongo, Bunkyou-ku, Tokyo, 113-0033, Japan; e-mail: mnishi@g.ecc.u-tokyo.ac.jp}}

\affil[]{Graduate School of Economics, University of Tokyo}

\date{\today}

\begin{document}
\onehalfspacing

    \begin{titlingpage}

        \maketitle

        \begin{abstract}
		  We consider panel data models with group structure. We study the asymptotic behavior of least-squares estimators and information criterion for the number of groups, allowing for the presence of small groups that have an asymptotically negligible relative size. Our contributions are threefold. First, we derive sufficient conditions under which the least-squares estimators are consistent and asymptotically normal. One of the conditions implies that a longer sample period is required as there are smaller groups. Second, we show that information criteria for the number of groups proposed in earlier works can be inconsistent or perform poorly in the presence of small groups. Third, we propose modified information criteria (MIC) designed to perform well in the presence of small groups. A Monte Carlo simulation confirms their good performance in finite samples. An empirical application illustrates that K-means clustering paired with the proposed MIC allows one to discover small groups without producing too many groups. This enables characterizing small groups and differentiating them from the other large groups in a parsimonious group structure.
        \end{abstract}

\medskip

\noindent

\emph{Keywords}: Clustering, grouped fixed-effects, information criterion, K-means, panel data

\medskip

\noindent

\emph{JEL Codes}: C23, C38

\end{titlingpage}

\section{Introduction}

Modeling cross-sectional heterogeneity is an important task when dealing with panel data, in which researchers face a trade-off between flexibility of models with complete heterogeneity and efficiency obtained from models with homogeneity across units. Recently, assuming a group structure among individuals has been viewed as a parsimonious yet sufficiently flexible way to model heterogeneity. In the group structure model, units within a group share the same parameter, while units belonging to different groups have different characteristics. Many researchers have developed estimation theory for grouped heterogeneity and unknown group memberships. \citet{linEstimationPanelData2012} and  \citet{bonhommeGroupedPatternsHeterogeneity2015} base their estimation strategies on the least-squares (LS) principle and the K-means algorithm. In particular, \citet{bonhommeGroupedPatternsHeterogeneity2015} establish consistency and asymptotic normality of the LS estimator in models with group-specific time-varying fixed-effects, or grouped fixed-effects (GFE). In a similar spirit, \citet{andoClusteringHugeNumber2017} consider models with group-specific interactive fixed-effects and develop asymptotic theory based on a penalized objective function. \citet{suIdentifyingLatentStructures2016}, \citet{mehrabaniEstimationIdentificationLatent2023}, and \citet{wangHomogeneitySparsityAnalysis2024} propose using penalized objective functions to estimate group-specific slope coefficients. \citet{okuiHeterogeneousStructuralBreaks2021} and \citet{lumsdaineEstimationPanelGroup2023} consider models in which group-specific slope parameters and group structure possibly experience structural breaks. The vast literature on the topic also includes \citet{hahnPanelDataModels2010}, \citet{luDeterminingNumberGroups2017}, \citet{liuIdentificationEstimationPanel2020}, \citet{wangIdentifyingLatentGroup2021}, \citet{chetverikovSpectralPostspectralEstimators2022}, \citet{mugnierSimpleComputationallyTrivial2023}, and \citet{yuSpectralClusteringVariance2024}. 

A commonly employed assumption in the above literature is that all latent groups have a size proportional to the number of cross-sections $N$, ruling out cases where one or more groups have a disproportionately small size \citep{bonhommeGroupedPatternsHeterogeneity2015,suIdentifyingLatentStructures2016,andoClusteringHugeNumber2017,luDeterminingNumberGroups2017,wangIdentifyingLatentGroup2021,okuiHeterogeneousStructuralBreaks2021,chetverikovSpectralPostspectralEstimators2022,lumsdaineEstimationPanelGroup2023}. This assumption, however, may be violated in real data analysis, as evidenced in some studies. For example, \citet{andoClusteringHugeNumber2017} apply a clustering algorithm to stock returns data and obtain a group of 2616 cross-sections, and, at the same time, also obtain a group consisting of only 19 units. Another example is \citet{aparicio-perezDisentanglingHeterogeneousEffect2025}, who investigate the nexus between a country's natural resources endowment and economic growth. In their analysis, they obtain 6 groups, one of which consists of 47 countries, while there is a group of only 2 countries. \citet{grunewaldTradeoffIncomeInequality2017} and \citet{martinez-zarzosoSearchingGroupedPatterns2020} also document disproportionality in group size. Although some recent studies establish inferential theory allowing for the presence of small groups \citep{mehrabaniEstimationIdentificationLatent2023,mugnierSimpleComputationallyTrivial2023,wangHomogeneitySparsityAnalysis2024}, the properties of clustering methods are still underexplored in such a situation.

One might think that the assumption of all groups having a size proportional to $N$ is imposed just to simplify the theoretical derivation and is innocuous. But we demonstrate that this is not true, at least as far as the LS estimator and the K-means algorithm are concerned. Specifically, we consider models with group-specific slope coefficients and GFE as considered in \citet{bonhommeGroupedPatternsHeterogeneity2015}, and study the asymptotic behavior of the LS estimators for slope coefficients, GFE, and group memberships and information criteria for the unknown number of groups.

In terms of the model specification and estimators employed, \citet{bonhommeGroupedPatternsHeterogeneity2015} is obviously a precursor to our study, but we extend their analysis in three nontrivial ways. First, we derive sufficient conditions under which the LS estimators are consistent and asymptotically normal, in an environment where at least one group has a size of smaller order than $N$. One of the sufficient conditions implies that a longer sample period $T$ is required as there are smaller groups. This points out a potential limitation of the LS clustering that may arise in short panels with small groups. Second, we study the asymptotic behavior of information criteria for the number of groups. Because \citet{bonhommeGroupedPatternsHeterogeneity2015} propose an information criterion without providing a theoretical analysis, the results drawn in the present work are entirely new. We derive a sufficient condition under which the information criterion is \textit{inconsistent} and \textit{underestimates} the number of groups. We show that an information criterion proposed in \citet{baiDeterminingNumberFactors2002} can satisfy the condition and hence underestimate the number of groups, particularly when there are small groups or the model includes GFE. Although we find that the BIC proposed in \citet{bonhommeGroupedPatternsHeterogeneity2015} does not satisfy the condition for underestimation, our simulation shows that it overestimates the number of groups in finite samples when the model does not include GFE. Finally, we propose modified information criteria (MIC) designed not to satisfy the conditions for inconsistency and to improve the information criteria proposed in \citet{baiDeterminingNumberFactors2002} and \citet{bonhommeGroupedPatternsHeterogeneity2015}. A simulation study shows their good performance in the presence of small groups. 

The scope of the present work is wider than that of earlier studies which consider the panel model with small groups.\footnote{Note that the ``small group" problem we discuss has nothing to do with the ``small cluster" problem considered in \citet{sunHomogeneityPursuitClustered2025}. They consider partitioning $m$ clusters into groups, where $m$ corresponds to $N$ in our notation, and assume that each group has a size proportional to $m$. What they mean by small \textit{clusters} is that each cluster has a small number of observations, which is analogous to a small $T$ (sample period) in our case, and thus does not refer to small \textit{groups} in our sense. See Remark \ref{rem:clarification} below for details.} \citet{wangHomogeneityPursuitPanel2018}, \citet{mehrabaniEstimationIdentificationLatent2023} and \citet{wangHomogeneitySparsityAnalysis2024} propose estimation frameworks based on penalized objective functions and show that their estimators are consistent. However, they study models where heterogeneity manifests itself only through slope coefficients and do not consider models with GFE. Furthermore, we show that a condition assumed in \citet{wangHomogeneityPursuitPanel2018} and \citet{mehrabaniEstimationIdentificationLatent2023} which ensures the consistency of their estimators for the number of groups does not hold in the presence of small groups. \citet{mugnierSimpleComputationallyTrivial2023}, who proposes a clustering method for the GFE model, shows that his estimators are consistent as long as each group contains at least two units. However, he mainly assumes homogeneous slope coefficients across cross-sections and only briefly discusses how his results may be extended to the cases where slope parameters are group-specific. Unlike these works, we study the LS estimator and the estimator for the number of groups, allowing for small groups, group-specific slope parameters, and GFE.

In our empirical application, we apply the K-means algorithm coupled with the proposed MIC to examine the determinants of sales growth of Japanese firms. Using the MIC, we discover a small group in our dataset, which is overlooked when we use \citeauthor{baiDeterminingNumberFactors2002}'s (\citeyear{baiDeterminingNumberFactors2002}) information criterion. This empirical exercise also illustrates that the MIC enables detecting small groups without producing too many groups, unlike the BIC considered in \citet{bonhommeGroupedPatternsHeterogeneity2015}. This allows us to characterize the small group and differentiate it from the other large groups while keeping a reasonably parsimonious structure. 

The remainder of this paper is organized as follows. In Section \ref{sec:model_asm}, we specify the model and spell out conditions under which we work. Section \ref{sec:without_gfe} gives asymptotic results for the LS estimator and the information criterion for the number of groups in models without GFE. In Section \ref{sec:with_gfe}, we extend the analysis given in Section \ref{sec:without_gfe} to models with both group-specific slope parameters and GFE. Section \ref{sec:mc} conducts a simulation study. In Section \ref{sec:empirics}, an empirical application is presented. Section \ref{sec:conclusion} concludes the paper. All mathematical proofs are relegated to Appendix.

\textbf{Notation}: For any $m\times n$ matrix $A=(a_{ij})$, $\|A\|\coloneqq(\sum_{i=1}^m\sum_{j=1}^na_{ij}^2)^{1/2}$ denotes the Frobenius norm of $A$. For any square matrix $A$, $\lambda_{\min}(A)$ and $\lambda_{\max}(A)$ denote the minimum and maximum eigenvalues of $A$, respectively. For any positive integer $n$, $[n]\coloneqq\{1,2,\ldots,n\}$ is the set of positive integers up to $n$. $\stackrel{p}{\to}$ and $\stackrel{d}{\to}$ signify convergence in probability and convergence in distribution as $N,T\to\infty$, respectively, where $N$ is the number of cross-sections and $T$ is the length of the sample period.

\section{Model and Assumptions}\label{sec:model_asm}

\subsection{Model and the LS estimator}

Consider the following panel data model:
\begin{align}
    y_{it}=x_{it}'\theta_{k_i} + \varepsilon_{it}, \ i=1,\ldots,N, \ t=1,\ldots,T,
    \label{model:without_gfe}
\end{align}
where $x_{it}$ is a $p\times 1$ vector of covariates, $k_i\in[K]$ denotes the nonrandom membership of unit $i$, $\theta_k\in\Theta\subset \mathbb{R}^p$ for $k\in[K]$, and $K$ is the number of groups.\footnote{We can also include individual fixed-effects in model \eqref{model:without_gfe}. Including individual fixed-effects does not change the conclusion given later as long as assumptions given below are modified suitably and regressor $x_{it}$ does not include a lagged dependent variable. We work with model \eqref{model:without_gfe} for notational simplicity.} Let $\boldsymbol{\gamma}_N \coloneqq (k_1,\ldots,k_N)\in [K]^N$ be the set of memberships of $N$ cross-sections and $\boldsymbol{\theta}$ denote the set of all $\theta_k$'s. We use $K^0$, $\theta_k^0$ and $k_i^0$ to denote the true values of $K$, $\theta_k$, and $k_i$, respectively. Similarly, $\boldsymbol{\theta}^0$ and $\boldsymbol{\gamma}_N^0$ are the sets of $\theta_k^0$'s and $k_i^0$'s, respectively. Also let $G_k^0\coloneqq\{i\in[N]:k_i^0=k\}$ be the set of cross-sections belonging to the true $k$-th group and $N_k^0\coloneqq|G_k^0|=\sum_{i=1}^N1\{k_i^0=k\}$ be the size (cardinality) of $G_k^0$.

We use the least-squares approach, or K-means, to estimate parameters $\boldsymbol{\theta}^0$ and $\boldsymbol{\gamma}_N^0$. Given a prespecified value $K$ for the number of groups, the least-squares estimator is the solution to the following minimization problem:
\begin{align}
    (\widehat{\boldsymbol{\theta}}^{(K)}, \widehat{\boldsymbol{\gamma}}_N^{(K)}) = \argmin_{(\boldsymbol{\theta},\boldsymbol{\gamma}_N)\in\Theta^K\times[K]^N}\sum_{i=1}^N\sum_{t=1}^T\left(y_{it} - x_{it}'\theta_{k_i}\right)^2,
    \label{def:LS_estimator}
\end{align}
where $\widehat{\boldsymbol{\gamma}}_N^{(K)}=(\widehat{k}_1^{(K)},\ldots,\widehat{k}_N^{(K)})$ is the set of estimated memberships. Let $\widehat{\sigma}^2(K, \widehat{\boldsymbol{\gamma}}_N^{(K)})\coloneqq (NT)^{-1}\sum_{i=1}^N\sum_{t=1}^T\bigl(y_{it}-x_{it}'\widehat{\theta}_{\widehat{k}_i^{(K)}}^{(K)}\bigr)^2$ be the averaged least-squares SSR obtained under $K$ groups.

\subsection{Assumptions}

Suppose that the true group sizes, $N_k^0$, can be expressed as $N_k^0=\tau_kN^{\alpha_k}$ with $\tau_k>0$ and $\alpha_k\in[0,1]$ such that $N_k^0\geq1$ for all $k\in[K^0]$.\footnote{$N=\sum_{k=1}^{K^0}N_k^0$ is implicitly assumed. This condition may require $\tau_k$ to depend on $N$ (i.e., $\tau_{N,k}$), but we omit this dependence on $N$ for notational simplicity. The results and conclusion below do not change as long as there exist some constants $c$ and $C$ such that $c,C\in(0,\infty)$ and $c<\tau_{N,k}<C$ for all $N$ and $k$.} Without loss of generality, we assume $\alpha_1\geq\alpha_2\geq\cdots\geq\alpha_{K^0}$.

\begin{asm}
 (a) $K^0$ is fixed as $N,T\to\infty$. (b) $1=\alpha_1=\alpha_2=\cdots=\alpha_m>\alpha_{m+1}\geq\cdots\geq\alpha_{K^0}\geq0$ for some $m\in[K^0-1]$.
 \label{asm:group_size}
\end{asm}
$m$ is the number of groups with asymptotically nonnegligible relative size (proportional to $N$). We require that at least one group have a size proportional to $N$ ($m\geq1$). The conventional situation where all groups have a nonnegligible size ($m=K^0$) is excluded.

\begin{rem}
    Let us clarify that the ``small cluster" problem considered in \citet{sunHomogeneityPursuitClustered2025} has nothing to do with the ``small group" problem discussed in the present work. Specifically, \citet{sunHomogeneityPursuitClustered2025} consider partitioning $m$ clusters into $K$ groups in the environment where $m\to\infty$ and $K$ is bounded. They assume that the size of each of the $K$ groups, $m_k$ (say), satisfies $m_k/m\to c_k>0$ (Assumption C5), which corresponds to $\alpha_k=1$ for all $k\in[K]$ in our notation. Hence, they consider the situation where all groups have a size proportional to $m$. What they mean by ``small clusters" is that each cluster has a bounded observation $n$. In other words, they deal with a $(m,n)$ clustered data and consider large-$m$, small-$n$ asymptotics, this being analogous to a $(N,T)$ panel with large $N$ and small $T$.
    \label{rem:clarification}
\end{rem}

\begin{asm}
    There exist constants $M>0$ and $c>0$ such that:
    
        \noindent(a) $\Theta$ is a compact subset of $\mathbb{R}^p$.
        
        \noindent(b) $E\left[\left\|x_{it}\right\|^4\right]\leq M$ for all $i,t$.
        
        \noindent(c) $E[\varepsilon_{it}]=0$ and $E[\varepsilon_{it}^4]\leq M$.
        
        \noindent(d) $(NT)^{-1/2}\sum_{i=1}^N\sum_{t=1}^T x_{it}\varepsilon_{it}=O_p(1)$.
        
        \noindent(e) Let $\rho_{NT}(\boldsymbol{\gamma}_N,k,\widetilde{k})$ be the minimum eigenvalue of $(N_k^0)^{-1}\sum_{i=1}^N1\{k_i^0=k\}1\{k_i=\widetilde{k}\}T^{-1}\sum_{t=1}^Tx_{it}x_{it}'$, where $\boldsymbol{\gamma}_N=(k_1,\ldots,k_N)\in[K^0]^N$. There exists a $\widehat{\rho}_{NT}\stackrel{p}{\to}\rho>0$ such that for all $k\in[K^0]$ and sufficiently large $T$, $\min_{\boldsymbol{\gamma}_N\in[K^0]^N}\max_{\widetilde{k}\in[K^0]}\rho_{NT}(\boldsymbol{\gamma}_N,k,\widetilde{k})\geq \widehat{\rho}_{NT}$.
        
        \noindent(f) There exists some $\widehat{c}_T$ such that for any $i$ and sufficiently large $T$, $\lambda_{\min}\left(T^{-1}\sum_{t=1}^Tx_{it}x_{it}'\right)\geq\widehat{c}_T$, and $\lim_{T\to\infty}\widehat{c}_T>c$.
        
        \noindent(g) If $\alpha_{K^0}\geq1/2$, then $N/T^{\nu}\to0$ for some $\nu>0$. If $\alpha_{K^0}<1/2$, then $N^{1-2\alpha_{K^0}}/T\to0$.
        
        \noindent(h) For all $k\neq\widetilde{k}$, $\left\|\theta_k^0 - \theta_{\widetilde{k}}^0\right\|>c$.
        
        \noindent(i) $\sup_{i\in[N]}P\left(T^{-1}\sum_{t=1}^T\left\|x_{it}\right\|^2 \geq M\right) = o(T^{-\delta})$ for all $\delta>0$.
        
        \noindent(j) For any constant $\epsilon>0$, $\sup_{i\in[N]}P\left(\left\|T^{-1}\sum_{t=1}^Tx_{it}\varepsilon_{it}\right\|\geq \epsilon\right)=o(T^{-\delta})$ for all $\delta>0$.
        
        \noindent(k) $(NT)^{-1}\sum_{t=1}^N\sum_{t=1}^T\varepsilon_{it}^2 \stackrel{p}{\to}\sigma^2>0$.
        \label{asm:basic}
\end{asm}

Assumption \ref{asm:basic}(a) requires that the parameter space be compact, a standard assumption in the literature. Assumptions \ref{asm:basic}(b) and (c) are also standard moment conditions. Assumption \ref{asm:basic}(d) essentially requires that $x_{it}\varepsilon_{it}$ be weakly correlated cross-sectionally and temporarily. Assumption \ref{asm:basic}(e) is similar to the usual rank condition in OLS estimation. In particular, this condition implies that $(N_k^0T)^{-1}\sum_{i\in G_k^0}\sum_{t=1}^Tx_{it}x_{it}'$ is uniformly positive definite for sufficiently large $N$ and $T$. Assumption \ref{asm:basic}(f) is similar.\footnote{In fact, (f) implies (e) under Assumption \ref{asm:group_size}(a), but we keep (e) as an assumption for convenience.} Assumption (f) is essentially the same as Assumption 1(iv) of \citet{lumsdaineEstimationPanelGroup2023}.

The first part of Assumption \ref{asm:basic}(g) with $\alpha_{K^0}=1$ is assumed in Corollary 1 of \citet{bonhommeGroupedPatternsHeterogeneity2015}. In the second part, $T$ must diverge faster as $\alpha_{K^0}$ is closer to 0, which implies that a smaller $\alpha_{K^0}$ needs to be compensated by a larger $T$. In particular, $T$ must diverge faster than $N$ if $\alpha_{K^0}=0$. Assumption \ref{asm:basic}(h) guarantees the identification of groups and is called the group separation condition in the literature. Assumption \ref{asm:basic}(i) requires that $T^{-1}\sum_{t=1}^T\|x_{it}\|^2$ be bounded by some (large) $M$ with arbitrary probability for sufficiently large $T$. In Assumption \ref{asm:basic}(j), $T^{-1}\sum_{t=1}^Tx_{it}\varepsilon_{it}$ has an arbitrarily thin tail for sufficiently large $T$. Assumptions \ref{asm:basic}(i) and (j) hold if $x_{it}$ and $x_{it}\varepsilon_{it}$ satisfy certain tail and mixing conditions; see \citet{bonhommeGroupedPatternsHeterogeneity2015}.

The following assumption is exploited to derive the asymptotic distribution of $\widehat{\boldsymbol{\theta}}^{(K^0)}$.
\begin{asm}
    \noindent(a) $E[x_{it}\varepsilon_{it}]=0$ for all $i,t$.
    
    \noindent(b) For all $k\in[K^0]$, there exist positive definite matrices $\Sigma_k$ and $\Omega_k$ such that
        \begin{align}
            \Sigma_k &= \plim_{N,T\to\infty} \frac{1}{N_{k}^0T}\sum_{i\in G_k^0}\sum_{t=1}^Tx_{it}x_{it}', \\
            \Omega_k &=\lim_{N,T\to\infty}\frac{1}{N_{k}^0T}\sum_{i\in G_k^0}\sum_{j\in G_k^0}\sum_{t=1}^T\sum_{s=1}^TE\left[x_{it}x_{js}'\varepsilon_{it}\varepsilon_{js}\right].
        \end{align}
        
    \noindent(c) $(N_{k}^0T)^{-1/2}\sum_{i\in G_k^0}\sum_{t=1}^Tx_{it}\varepsilon_{it}\stackrel{d}{\to} N(0,\Omega_k)$ for all $k\in[K^0]$.
    \label{asm:dist}
\end{asm}

\section{Asymptotic Results in Models without GFE}\label{sec:without_gfe}

\subsection{Consistency and asymptotic normality of the LS estimator when $K^0$ is known}
The following theorem, which is understood up to group relabeling, is obtained assuming that $K=K^0$ is known.

\begin{thm}
    Consider model \eqref{model:without_gfe}. Suppose that Assumptions \ref{asm:group_size}-\ref{asm:dist} hold. Then as $N,T\to\infty$, we have
    
    \noindent(i) $P\left(\bigcup_{i=1}^N\left\{\widehat{k}_i^{(K^0)}\neq k_i^0\right\}\right)=o(1)$,
    
    \noindent(ii) $\sqrt{N_k^0T}\left(\widehat{\theta}_k^{(K^0)} - \theta_k^0\right) \stackrel{d}{\to}N\left(0,\Sigma_k^{-1}\Omega_k\Sigma_k^{-1}\right)$ for each $k\in[K^0]$,
    
    \noindent(iii) $\widehat{\sigma}^2(K^0, \widehat{\boldsymbol{\gamma}}_N^{(K^0)})=(NT)^{-1}\sum_{i=1}^N\sum_{t=1}^T\varepsilon_{it}^2 + O_p\left(\left(NT\right)^{-1}\right) \stackrel{p}{\to}\sigma^2$.
    \label{thm:consistency_normality}
\end{thm}

\citet{bonhommeGroupedPatternsHeterogeneity2015} and \citet{lumsdaineEstimationPanelGroup2023} derive the same results as Theorem \ref{thm:consistency_normality}(i) and (ii), but they assume that all groups have a size proportional to $N$ (i.e., $\alpha_k=1$ for all $k$). Theorem 1 extends their results to models with small groups. A key condition is Assumption \ref{asm:basic}(g), which requires that $T$ be sufficiently large adaptively to the size of the smallest $K^0$-th group. We cannot ensure consistency and asymptotic normality without this assumption. This warns us about a potential limitation of K-means that may arise in panel data with small $T$ and small groups. Such a situation does not seem unusual, as researchers often apply clustering methods to panels of short to moderate length \citep{bonhommeGroupedPatternsHeterogeneity2015,suIdentifyingLatentStructures2016,martinez-zarzosoSearchingGroupedPatterns2020}.

\subsection{Asymptotics for the information criterion for the number of groups}

Now, we turn to the problem of estimating $K^0$ by using the information criterion. Referring to \citet{baiDeterminingNumberFactors2002} and \citet{baiPanelDataModels2009}, \citet{bonhommeGroupedPatternsHeterogeneity2015} propose the following criterion:
\begin{align}
    \mathrm{IC}(K,h_{NT}) \coloneqq \widehat{\sigma}^2(K,\widehat{\boldsymbol{\gamma}}_N^{(K)}) + n(K)\widetilde{\sigma}^2h_{NT},
    \label{def:ic}
\end{align}
where $n(K)$ is the number of parameters under $K$ groups, $\widetilde{\sigma}^2$ is a consistent estimate of $\sigma^2=\plim (NT)^{-1}\sum_{i=1}^N\sum_{t=1}^T\varepsilon_{it}^2$, and $h_{NT}>0$ is a penalty term satisfying $h_{NT}\to0$ as $N,T\to\infty$.\footnote{$\widetilde{\sigma}^2$ is typically constructed as $\widetilde{\sigma}^2=\{NT-n(K_{\max})\}^{-1}\times NT\widehat{\sigma}^2(K_{\max},\widehat{\boldsymbol{\gamma}}_N^{(K_{\max})})$ \citep{baiDeterminingNumberFactors2002,bonhommeGroupedPatternsHeterogeneity2015}.} $n(K)=N+pK$ in model \eqref{model:without_gfe}. The estimated number of groups is
\begin{align}
    \widehat{K}(h_{NT}) \coloneqq \argmin_{K\in[K_{\max}]}\mathrm{IC}(K,h_{NT})
\end{align}
for some $K_{\max}\geq K^0$.

\begin{prop}
    Consider model \eqref{model:without_gfe}. Suppose that Assumptions \ref{asm:group_size}-\ref{asm:dist} hold.
    
    \noindent(i) If 
    \begin{align}
        NTh_{NT} \to\infty, \label{cond:not_overestimate}
    \end{align}
    then we have $P(\widehat{K}(h_{NT}) \leq K^0) \to 1$ as $N,T\to\infty$.

    \noindent(ii) If
    \begin{align}
        N^{1-\alpha_{m+1}}h_{NT} \to \infty, \label{cond:underestimate}
    \end{align}
    then we have $P(\widehat{K}(h_{NT}) < K^0) \to 1$ as $N,T\to\infty$.
    \label{prop:ic_inconsistency}
\end{prop}

Proposition \ref{prop:ic_inconsistency} has two parts. Part (i) gives a sufficient condition on $h_{NT}$ for $\widehat{K}(h_{NT})$ not to overestimate. The second part is an inconsistency result. Condition \eqref{cond:underestimate} tells us what kind of $h_{NT}$ we should \textit{not} use. In the next subsection, we will check whether popular penalty choices satisfy \eqref{cond:not_overestimate} and \eqref{cond:underestimate}.

\begin{rem}
In the proof of Proposition \ref{prop:ic_inconsistency}, we show that for any $\underline{K}\in\{m,m+1,\ldots,K^0-1\}$, there exists a sequence of a set of groupings $\Gamma_N(\underline{K},m)$ such that for all $\boldsymbol{\gamma}_N(\underline{K},m)=(k_1,\ldots,k_N)\in\Gamma_N(\underline{K},m)$, 
\begin{align}
    \widehat{\sigma}^2(\underline{K},\widehat{\boldsymbol{\gamma}}_N^{(\underline{K})}) \leq \widehat{\sigma}^2(\underline{K},\boldsymbol{\gamma}_N(\underline{K},m)) = \widehat{\sigma}^2(K^0,\widehat{\boldsymbol{\gamma}}_N^{(K^0)}) + O_p(N^{\alpha_{m+1}-1}),
    \label{ineq:first_term_IC}
\end{align}
where $\widehat{\sigma}^2(K,\boldsymbol{\gamma}_N)$ is the averaged SSR obtained from the OLS estimation based on a given partition $\boldsymbol{\gamma}_N\in[K]^N$. This implies that the first terms of $\mathrm{IC}(\underline{K},h_{NT})$ and $\mathrm{IC}(K^0,h_{NT})$ differ only marginally, the difference being of order $O_p(N^{\alpha_{m+1}-1})$. Because $n(\underline{K})h_{NT} < n(K^0)h_{NT}$, $\mathrm{IC}(\underline{K},h_{NT})$ is asymptotically smaller than $\mathrm{IC}(K^0,h_{NT})$ if the difference between the second terms dominates the difference between the first terms, which is ensured under condition \eqref{cond:underestimate}.

A typical argument to prove the consistency of information criteria involves assuming or arguing that there exists a fixed constant $c>0$ such that
\begin{align}
    \min_{1\leq \underline{K}<K^0}\min_{\gamma_N\in[\underline{K}]^N} \hat{\sigma}^2(\underline{K},\gamma_N) - \frac{1}{NT}\sum_{i=1}^N\sum_{t=1}^T\varepsilon_{it}^2 > c,
    \label{cond:worse_fit}
\end{align}
w.p.a.1 \citep{suIdentifyingLatentStructures2016,wangHomogeneityPursuitPanel2018,wangIdentifyingLatentGroup2021,mehrabaniEstimationIdentificationLatent2023,sunHomogeneityPursuitClustered2025}. Looking at \eqref{ineq:first_term_IC} and Theorem \ref{thm:consistency_normality}(iii) reveals that \eqref{cond:worse_fit} does not hold in the presence of small groups.
\label{rem:inconsistency}
\end{rem}

\subsection{Examples of penalty $h_{NT}$}

We consider two examples. The first one is the penalty proposed in \citet{baiDeterminingNumberFactors2002}, of the form
\begin{align}
    h_{NT}^{\mathrm{BN}} \coloneqq \frac{\ln(\min\{N,T\})}{\min\{N,T\}}.
    \label{def:bn}
\end{align}
Clearly, $h_{NT}^{\mathrm{BN}}$ satisfies \eqref{cond:not_overestimate}, and hence $\widehat{K}(h_{NT}^{\mathrm{BN}})$ asymptotically does not overestimate. If $N/T\to c\in[0,\infty)$, condition \eqref{cond:underestimate} is not satisfied under $h_{NT}=h_{NT}^{\mathrm{BN}}$, because $N^{1-\alpha_{m+1}}h_{NT}^{\mathrm{BN}}=N^{1-\alpha_{m+1}}(\ln N)/N\to0$ (unless $\alpha_{m+1}=0$). If $N/T\to\infty$, then $N^{1-\alpha_{m+1}}h_{NT}^{\mathrm{BN}}=N^{1-\alpha_{m+1}}(\ln T)/T$. Recalling Assumption \ref{asm:basic}(g), we can easily verify that, when $\alpha_{K^0}\geq1/2$, condition \eqref{cond:underestimate} holds if $\nu>(1-\alpha_{m+1})^{-1}$ and $N/T^{\nu'}\to c \in (0,\infty]$ for some $\nu'\in[(1-\alpha_{m+1})^{-1},\nu)$, where $\nu$ is the number appearing in Assumption \ref{asm:basic}(g). When $\alpha_{K^0}<1/2$ and $N/T\to\infty$, and under Assumption \ref{asm:basic}(g), condition \eqref{cond:underestimate} holds if, for example, $T=cN^{\delta}$ for some $c>0$ and some $\delta\in(1-2\alpha_{K^0},1-\alpha_{m+1}]$. Under these conditions, $\widehat{K}(h_{NT}^{\mathrm{BN}})$ asymptotically underestimates. Intuitively speaking, $\widehat{K}(h_{NT}^{\mathrm{BN}})$ cannot capture a small amount of information coming from cross-sections that belong to small groups, so that the small groups can be overlooked.

In sum, $\widehat{K}(h_{NT}^{\mathrm{BN}})$ can underestimate the number of groups if $N/T\to\infty$, as long as the regularity conditions hold. This is another result that warns us about a pitfall we may face when dealing with short panels. Furthermore, the above analysis reveals that $\widehat{K}(h_{NT}^{\mathrm{BN}})$ can be inconsistent if $N$ is disproportionately large relative to $T$, even if $T$ itself is ``large" in the absolute sense. We emphasize that not only the value of $(N,T)$ but also the ratio $N/T$ matters.

Another example is the BIC-type penalty proposed in the supplementary materials to \citet{baiPanelDataModels2009} and \citet{bonhommeGroupedPatternsHeterogeneity2015} (in models with interactive fixed effects):
\begin{align}
    h_{NT}^{\mathrm{BIC}} \coloneqq \frac{\ln(NT)}{NT}.
\end{align}

It can be easily verified that \eqref{cond:not_overestimate} holds under $h_{NT}=h_{NT}^{\mathrm{BIC}}$. Moreover, $h_{NT}=h_{NT}^{\mathrm{BIC}}$ does not satisfy \eqref{cond:underestimate} unless $\alpha_{m+1}=0$ and $N=\exp(T^{\beta})$ for some $\beta>1$, which is a rare situation in typical applications. However, as pointed out by \citet{bonhommeGroupedPatternsHeterogeneity2015}, $\widehat{K}(h_{NT}^{\mathrm{BIC}})$ may overestimate the number of groups in finite samples. Indeed, our simulation (reported in Section \ref{sec:mc}) documents this tendency, indicating that $h_{NT}=h_{NT}^{\mathrm{BIC}}$ is not suitable unless a large sample is available.\footnote{This is not the case when we consider models with GFE; see Section \ref{sec:with_gfe}.}

\subsection{Modified information criterion}

The above analysis provides us with an insight into what kind of penalty may be useful. On the one hand, $\widehat{K}(h_{NT}^{\mathrm{BN}})$ can be inconsistent when $N/T\to \infty$, because $h_{NT}^{\mathrm{BN}}$ shrinks toward zero so \textit{slowly} that \eqref{cond:underestimate} holds. On the other hand, the shrinking rate of $h_{NT}^{\mathrm{BIC}}$ is fast enough to prevent \eqref{cond:underestimate} from being satisfied, but this rate is so \textit{fast} that $\widehat{K}(h_{NT}^{\mathrm{BIC}})$ overestimates the number of groups in finite samples. Therefore, a penalty that performs well needs to: (i) satisfy \eqref{cond:not_overestimate}, (ii) not satisfy \eqref{cond:underestimate}, (iii) shrink toward zero faster than $h_{NT}^{\mathrm{BN}}$ when $N/T\to\infty$, and (iv) shrink more slowly than $h_{NT}^{\mathrm{BIC}}$. Given this observation, we propose the following penalty:
\begin{align}
    h_{NT}^{\mathrm{MIC}1} \coloneqq \begin{cases}\displaystyle
        \frac{\ln N}{N} & \mathrm{if} \ N\leq T \\[2ex] \displaystyle
        0.5\frac{\ln (NT)}{N} & \mathrm{if} \ N> T
    \end{cases}.
\end{align}

It is straightforward to check that $h_{NT}^{\mathrm{MIC}1}$ satisfies \eqref{cond:not_overestimate} but does not satisfy \eqref{cond:underestimate} if $\alpha_{m+1}>0$.\footnote{If $\alpha_{m+1}=0$, it suffices to replace the denominator of $h_{NT}^{\mathrm{MIC}1}$ with $N^{1+\epsilon}$ for some small $\epsilon>0$.} $h_{NT}^{\mathrm{MIC}1}$ takes the same form as $h_{NT}^{\mathrm{BN}}$ when $N\leq T$, and is asymptotically equivalent to $h_{NT}^{\mathrm{BN}}$ if $N/T\to c\in [0,\infty)$. This is because our simulation shows that $h_{NT}^{\mathrm{BN}}$ performs well in this case. When $N/T\to\infty$, $h_{NT}^{\mathrm{MIC}1}=0.5\ln(NT)/N \leq (\ln N)/N \ll (\ln T)/T=h_{NT}^{\mathrm{BN}}$. Finally, $h_{NT}^{\mathrm{MIC}1} \gg h_{NT}^{\mathrm{BIC}}$. $h_{NT}^{\mathrm{MIC}1}$ satisfies all the required conditions.\footnote{A scaling constant used in the definition of $h_{NT}^{\mathrm{MIC}1}$ is not important asymptotically. We adjust a scaling constant just to preserve continuity of $h_{NT}^{\mathrm{MIC}1}$ in $(N,T)$. Although the choice of a scaling constant may affect the finite-sample performance, we do not explore in this direction.} In our simulation, $\widehat{K}(h_{NT}^{\mathrm{MIC}1})$ shows a good finite-sample performance; see Section \ref{sec:mc}.

\begin{rem}
    Note that we do not prove the consistency of $\widehat{K}(h_{NT}^{\mathrm{MIC}1})$. In the literature of panel data clustering, the consistency of information criteria is typically established assuming \eqref{cond:worse_fit} \citep{suIdentifyingLatentStructures2016,wangHomogeneityPursuitPanel2018,wangIdentifyingLatentGroup2021,mehrabaniEstimationIdentificationLatent2023}. As we discussed in Remark \ref{rem:inconsistency}, however, \eqref{cond:worse_fit} does not hold in the presence of small groups. It seems difficult to establish the consistency of $\widehat{K}(h_{NT}^{\mathrm{MIC}1})$ without \eqref{cond:worse_fit}.
\end{rem}

\section{The Case of Models with GFE}\label{sec:with_gfe}

Since the introduction by \citet{bonhommeGroupedPatternsHeterogeneity2015}, panel data models with GFE have drawn researchers' attention \citep[e.g.,][]{chetverikovSpectralPostspectralEstimators2022,mugnierSimpleComputationallyTrivial2023}. In this section, we consider the following model:
\begin{align}
    y_{it}=x_{it}'\theta_{k_i} + \mu_{k_it} + \varepsilon_{it}, \ i=1,\ldots,N, \ t=1,\ldots,T,
    \label{model:with_gfe}
\end{align}
where $\mu_{kt}\in\mathcal{M}\subset\mathbb{R}$ for $k\in[K]$ and $t\in[T]$. Here, $\mu_{k_it}$ is introduced as GFE into model \eqref{model:without_gfe}. To simplify the discussion, we assume $\mu_{kt}$'s are nonrandom. Define $\mu_k\coloneqq(\mu_{k1},\ldots,\mu_{kT})'$ for all $k$, and let $\boldsymbol{\mu}$ denote the set of all $\mu_k$'s. As earlier, we use the superscripts $0$ to refer to the true parameter values.

The least-squares estimator under $K$ groups is
\begin{align}
    (\widehat{\boldsymbol{\theta}}^{(K)}, \widehat{\boldsymbol{\mu}}^{(K)}, \widehat{\boldsymbol{\gamma}}_N^{(K)}) = \argmin_{(\boldsymbol{\theta},\boldsymbol{\mu},\boldsymbol{\gamma}_N)\in\Theta^K\times\mathcal{M}^{KT}\times[K]^N}\sum_{i=1}^N\sum_{t=1}^T\left(y_{it} - x_{it}'\theta_{k_i} - \mu_{k_it}\right)^2.
    \label{def:LS_estimator_gfe}
\end{align}
With some abuse of notation, we let $\widehat{\sigma}^2(K, \widehat{\boldsymbol{\gamma}}_N^{(K)})= (NT)^{-1}\sum_{i=1}^N\sum_{t=1}^T\bigl(y_{it}-x_{it}'\widehat{\theta}_{\widehat{k}_i^{(K)}}^{(K)} - \widehat{\mu}_{\widehat{k}_i^{(K)}t}^{(K)}\bigr)^2$ denote the averaged least-squares SSR obtained from \eqref{model:with_gfe} with $K$ groups.

\subsection{Assumptions and results for the LS estimator}

Introducing GFE requires modifying Assumptions \ref{asm:group_size}-\ref{asm:dist}.

\setcounter{asm}{0}
\renewcommand{\theasm}{\arabic{asm}'}
\begin{asm}
 (a) $K^0$ is fixed as $N,T\to\infty$. (b') $1=\alpha_1=\alpha_2=\cdots=\alpha_m>\alpha_{m+1}\geq\cdots\geq\alpha_{K^0}>0$ for some $m\in[K^0-1]$.
 \label{asm:group_size_gfe}
\end{asm}

Assumption \ref{asm:group_size_gfe} requires that all groups be of divergent size as $N\to\infty$. This is because, for each $k\in[K]$, we can only use cross-sectional information to estimate $\mu_{kt}$ for each $t=1,\ldots,T$; that is, a consistent estimation of $\mu_{kt}$ is achieved only when $N_k^0\to \infty$.

\begin{asm}
    Assumptions \ref{asm:basic}(a), (d), (e), (f), (h), and (j) are replaced by the following conditions:

    \noindent(a') $\Theta$ is a compact subset of $\mathbb{R}^p$. $\mathcal{M}$ is a compact subset of $\mathbb{R}$.

    \noindent(d') $(NT)^{-1/2}\sum_{i=1}^N\sum_{t=1}^T x_{it}\varepsilon_{it}=O_p(1)$, and $(NT)^{-1/2}\sum_{i=1}^N\sum_{t=1}^T \varepsilon_{it}=O_p(1)$.
    
    \noindent(e') Let $\rho_{NT}^*(\boldsymbol{\gamma}_N,k,\widetilde{k})$ be the minimum eigenvalue of $(N_k^0)^{-1}\sum_{i=1}^N1\{k_i^0=k\}1\{k_i=\widetilde{k}\}D_{T,i}$, where
    \begin{align}
        D_{T,i} \coloneqq \begin{pmatrix}
            T^{-1}\sum_{t=1}^Tx_{it}x_{it}' & x_{i1}/\sqrt{T} & x_{i2}/\sqrt{T} & \cdots & x_{iT}/\sqrt{T} \\
            x_{i1}'/\sqrt{T} & 1 & 0 & \cdots & 0 \\
            x_{i2}'/\sqrt{T} & 0 & 1 & \cdots & 0 \\
            \vdots & \vdots & \vdots & \cdots & \vdots \\
            x_{iT}'/\sqrt{T} & 0 & 0 & \cdots & 1
        \end{pmatrix},
    \end{align}
    and $\boldsymbol{\gamma}_N=(k_1,\ldots,k_N)\in[K^0]^N$. There exists a $\widehat{\rho}_{NT}^*\stackrel{p}{\to}\rho^*>0$ such that for all $k\in[K^0]$ and sufficiently large $T$, $\min_{\boldsymbol{\gamma}_N\in[K^0]^N}\max_{\widetilde{k}\in[K^0]}\rho_{NT}^*(\boldsymbol{\gamma}_N,k,\widetilde{k})\geq \widehat{\rho}_{NT}^*$.
    
    \noindent(f') There exists some $\widehat{c}_T^*$ such that for any $i$ and sufficiently large $T$, $\lambda_{\min}\left(D_{T,i}\right)\geq\widehat{c}_T^*$, and $\lim\widehat{c}_T^*>c$ as $T\to\infty$.
    
    \noindent(h') For all $k\neq\widetilde{k}$, $\left\|\theta_k^0 - \theta_{\widetilde{k}}^0\right\|>c$ and $\lim_{T\to\infty}T^{-1}\sum_{t=1}^T(\mu_{kt}^0-\mu_{\widetilde{k}t}^0)^2 >c$.
    
    \noindent(j') For any constant $\epsilon>0$, $\sup_{i\in[N]}P\left(\left\|T^{-1}\sum_{t=1}^Tx_{it}\varepsilon_{it}\right\|\geq \epsilon\right)=o(T^{-\delta})$, $\sup_{i\in[N]}P\left(\left|T^{-1}\sum_{t=1}^T\varepsilon_{it}\right|\geq \epsilon\right)=o(T^{-\delta})$, and $\sup_{i\in[N]}P\left(\left|T^{-1}\sum_{t=1}^T\left(\varepsilon_{it}^2 -E[\varepsilon_{it}^2]\right)\right|\geq \epsilon\right)=o(T^{-\delta})$ for all $\delta>0$.
    \label{asm:basic_gfe}
\end{asm}

Assumption \ref{asm:basic_gfe}(e') is essentially the same as Assumption S.2(a) of \citet{bonhommeGroupedPatternsHeterogeneity2015}.

\begin{asm}
    Assumptions \ref{asm:dist}(b) and (c) are replaced with the following conditions.

    \noindent(b')For all $k\in[K^0]$, there exist positive definite matrices $\Sigma_k^*$ and $\Omega_k^*$ such that
        \begin{align}
            \Sigma_k^* &= \plim_{N,T\to\infty} \frac{1}{N_{k}^0T}\sum_{i\in G_k^0}\sum_{t=1}^T(x_{it} - \overline{x}_{kt})(x_{it} - \overline{x}_{kt})', \\
            \Omega_k^* &=\lim_{N,T\to\infty}\frac{1}{N_{k}^0T}\sum_{i\in G_k^0}\sum_{j\in G_k^0}\sum_{t=1}^T\sum_{s=1}^TE\left[(x_{it} - \overline{x}_{kt})(x_{jt} - \overline{x}_{kt})'\varepsilon_{it}\varepsilon_{js}\right],
        \end{align}
    where $\overline{x}_{kt} \coloneqq (N_k^0)^{-1}\sum_{i\in G_k^0}x_{it}$.
    
    \noindent(c') $(N_{k}^0T)^{-1/2}\sum_{i\in G_k^0}\sum_{t=1}^T(x_{it} - \overline{x}_{kt})\varepsilon_{it}\stackrel{d}{\to} N(0,\Omega_k^*)$ for all $k\in[K^0]$.

    Additionally, the following conditions hold.

    \noindent(d) For all $k\in[K^0]$ and $t\in[T]$, $\lim_{N\to\infty}(N_k^0)^{-1}\sum_{i\in G_k^0}\sum_{j\in G_k^0}E[\varepsilon_{it}\varepsilon_{jt}] = \omega_{kt}>0$.

    \noindent(e) For all $k\in[K^0]$ and $t\in[T]$, $(N_k^0)^{-1/2}\sum_{i\in G_k^0}\varepsilon_{it} \stackrel{d}{\to} N(0,\omega_{kt})$ as $N\to\infty$.
    \label{asm:dist_gfe}
\end{asm}

The following theorem establishes the asymptotic normality of the LS estimator from \eqref{model:with_gfe} with $K=K^0$.

\begin{thm}
    Consider model \eqref{model:with_gfe}. Suppose that Assumptions \ref{asm:group_size_gfe}-\ref{asm:dist_gfe} hold. Then as $N,T\to\infty$, we have
    
    \noindent(i) $P\left(\bigcup_{i=1}^N\left\{\widehat{k}_i^{(K^0)}\neq k_i^0\right\}\right)=o(1)$,
    
    \noindent(ii) $\sqrt{N_k^0T}\left(\widehat{\theta}_k^{(K^0)} - \theta_k^0\right) \stackrel{d}{\to}N\left(0,(\Sigma_k^{*})^{-1}\Omega_k^*(\Sigma_k^{*})^{-1}\right)$ for each $k\in[K^0]$,

    \noindent(iii) $\sqrt{N_k^0}\left(\widehat{\mu}_{kt}^{(K^0)} - \mu_{kt}^0\right) \stackrel{d}{\to}N\left(0,\omega_{kt}\right)$ for each $k\in[K^0]$ and $t\in[T]$,
    
    \noindent(iv) $\widehat{\sigma}^2(K^0,\widehat{\boldsymbol{\gamma}}_N^{(K^0)})=(NT)^{-1}\sum_{i=1}^N\sum_{t=1}^T\varepsilon_{it}^2 + O_p\left(N^{-1}\right) \stackrel{p}{\to}\sigma^2$.
    \label{thm:consistency_normality_gfe}
\end{thm}

Theorem \ref{thm:consistency_normality_gfe} is a generalization of Corollary S.2 of \citet{bonhommeGroupedPatternsHeterogeneity2015} to the case of panels with small groups.

\subsection{(In)Consistency of information criteria}\label{subsec:ic_gfe}

As before, we use $\mathrm{IC}(K,h_{NT})$ defined in \eqref{def:ic} as the criterion to determine the number of groups. Note that the definition of $\mathrm{IC}(K,h_{NT})$ in this section differs from that considered in Section \ref{sec:without_gfe} in two respects. First, the first term is $\widehat{\sigma}^2(K, \widehat{\boldsymbol{\gamma}}_N^{(K)})= (NT)^{-1}\sum_{i=1}^N\sum_{t=1}^T\bigl(y_{it}-x_{it}'\widehat{\theta}_{\widehat{k}_i^{(K)}}^{(K)} - \widehat{\mu}_{\widehat{k}_i^{(K)}t}^{(K)}\bigr)^2$. Second, the number of parameters is $n(K)=N+(p+T)K$.

\begin{prop}
    Consider model \eqref{model:with_gfe}. Suppose that Assumptions \ref{asm:group_size}-\ref{asm:dist} hold.
    
    \noindent(i) If \eqref{cond:not_overestimate} holds, then we have $P(\widehat{K}(h_{NT}) \leq K^0)\to 1$ as $N,T\to\infty$.
    
    \noindent(ii) If 
    \begin{align}
        TN^{1-\alpha_{m+1}}h_{NT} \to \infty, \label{cond:underestimate_gfe}
    \end{align}
    then we have $P(\widehat{K}(h_{NT}) < K^0) \to 1$ as $N,T\to\infty$.
    \label{prop:ic_inconsistency_gfe}
\end{prop}

A key condition is \eqref{cond:underestimate_gfe}, similar to \eqref{cond:underestimate}. Condition \eqref{cond:underestimate_gfe}, however, differs from \eqref{cond:underestimate} by a multiplicative factor of $T$. This factor emerges because the number of parameters is now $n(K)=N+(p+T)K$. The difference between the second terms of $\mathrm{IC}(K_1,h_{NT})$ and $\mathrm{IC}(K_2,h_{NT})$ ($K_1> K_2$), say, is $\widetilde{\sigma}^2(n(K_1) - n(K_2))h_{NT} = \widetilde{\sigma}^2(p+T)(K_1-K_2)h_{NT} \asymp Th_{NT}$, so that $Th_{NT}$ plays the role of a penalty in effect. 

Let us investigate whether penalties considered in Section \ref{sec:without_gfe} satisfy \eqref{cond:underestimate_gfe}. Recall that $h_{NT}^{\mathrm{BN}}$ is defined by \eqref{def:bn}. When $N/T\to\infty$, we have $TN^{1-\alpha_{m+1}}h_{NT}^{\mathrm{BN}} = TN^{1-\alpha_{m+1}}(\ln T)/T \to \infty$. When $N/T\to c\in[0,\infty)$, then $TN^{1-\alpha_{m+1}}h_{NT}^{\mathrm{BN}} \asymp TN^{1-\alpha_{m+1}}(\ln N)/N \to \infty$. Therefore, $h_{NT}^{\mathrm{BN}}$ \textit{necessarily} satisfies \eqref{cond:underestimate_gfe}, which implies that $\widehat{K}(h_{NT}^{\mathrm{BN}})$ asymptotically underestimates in model \eqref{model:with_gfe}. Furthermore, $h_{NT}^{\mathrm{BN}}$ satisfies \eqref{cond:underestimate_gfe} even when $\alpha_{m+1}=\alpha_{K^0}=1$ (i.e., there is no small group), although we exclude this case by Assumption \ref{asm:group_size_gfe}. Indeed, our simulation shows that $\widehat{K}(h_{NT}^{\mathrm{BN}})$ underestimates the number of groups in GFE model \eqref{model:with_gfe} even when all groups have a size proportional to $N$, indicating that underestimation by $\widehat{K}(h_{NT}^{\mathrm{BN}})$ is not a matter confined to the small-group situation.

Next, consider $h_{NT}^{\mathrm{BIC}}$. When $N/T\to c\in(0,\infty]$, we have $TN^{1-\alpha_{m+1}}h_{NT}^{\mathrm{BIC}}=TN^{1-\alpha_{m+1}}\ln(NT)/NT \leq \ln(N^2)/N^{\alpha_{m+1}}\to0$, and hence \eqref{cond:underestimate_gfe} does not hold in this case. On the other hand, $h_{NT}^{\mathrm{BIC}}$ can satisfy \eqref{cond:underestimate_gfe} and hence underestimate the number of groups when $N/T\to0$. For example, it is straightforward to check that $TN^{1-\alpha_{m+1}}h_{NT}^{\mathrm{BIC}} \to \infty$ if $T=\exp(N^{\alpha_{m+1}+\epsilon})$ for some $\epsilon>0$. In our simulation reported in Section \ref{sec:mc}, $\widehat{K}(h_{NT}^{\mathrm{BIC}})$ performs well in a setting where $T$ is not too small or too large relative to $N$.

Finally, we consider constructing a penalty that fits model \eqref{model:with_gfe}. Since $Th_{NT}$ essentially behaves as a penalty, we suggest the following penalty, which is obtained by simply dividing $h_{NT}^{\mathrm{MIC}1}$ by $T$:
\begin{align}
    h_{NT}^{\mathrm{MIC}2} \coloneqq \begin{cases}\displaystyle
        2\frac{\ln N}{NT} & \mathrm{if} \ N\leq T \\[2ex] \displaystyle
        \frac{\ln (NT)}{NT} & \mathrm{if} \ N> T 
    \end{cases}.
\end{align}

When $N>T$, $h_{NT}^{\mathrm{MIC}2}=h_{NT}^{\mathrm{BIC}}$. When $N\leq T$, $h_{NT}^{\mathrm{MIC}2}$ is still asymptotically equivalent to $h_{NT}^{\mathrm{BIC}}$ if $\ln T = O(\ln N)$. However, if $T$ satisfies $\ln T/\ln N \to\infty$, then $h_{NT}^{\mathrm{MIC}2} \ll h_{NT}^{\mathrm{BIC}}$ and thus $\widehat{K}(h_{NT}^{\mathrm{MIC}2})$ tends to be larger than $\widehat{K}(h_{NT}^{\mathrm{BIC}})$. In a typical application where $T$ is not too large relative to $N$, one will often have $\widehat{K}(h_{NT}^{\mathrm{MIC}2})=\widehat{K}(h_{NT}^{\mathrm{BIC}})$.

\section{Monte Carlo Simulation}\label{sec:mc}

In this section, we evaluate the finite-sample performance of the information criteria for the number of groups and the LS estimator calculated under $K=K^0$. We consider the following three data generating processes (DGPs) in which the number of groups is $K^0=3$:

DGP 1 (Static panel). $y_{it}=\theta_{k_i,1}x_{it,1} + \theta_{k_i,2}x_{it,2} + \varepsilon_{it}$, where $x_{it,1},x_{it,2},\varepsilon_{it}\sim \ \mathrm{i.i.d.} \ N(0,1)$, $(\theta_{1,1},\theta_{1,2})=(3,-3)$, $(\theta_{2,1},\theta_{2,2})=(1,-2)$, and $(\theta_{3,1},\theta_{3,2})=(4,-1)$.

DGP 2 (Dynamic panel). $y_{it}=\theta_{k_i,1}x_{it,1} + \theta_{k_i,2}y_{i,t-1} + \varepsilon_{it}$, where $(\theta_{1,1},\theta_{1,2})=(3,0.2)$, $(\theta_{2,1},\theta_{2,2})=(1,0.5)$, and $(\theta_{3,1},\theta_{3,2})=(4,0.8)$.

DGP 3 (GFE). $y_{it}=\theta_{k_i,1}x_{it,1} + \theta_{k_i,2}x_{it,2} + \mu_{k_it} + \varepsilon_{it}$, where $(\theta_{1,1},\theta_{1,2})=(3,-3)$, $(\theta_{2,1},\theta_{2,2})=(1,-2)$, $(\theta_{3,1},\theta_{3,2})=(4,-1)$, $\mu_{1t}=4t/T$, $\mu_{2t}=2t/T$, and $\mu_{3t}=4$ for all $t$.

For each DGP, we consider $N\in\{60,90,120\}$ and set $N_1^0=N/K^0$, $N_3=\lfloor c_\alpha N^{\alpha}\rfloor$, where $\alpha\in\{0.2,0.3,\ldots,1\}$, and
\begin{align}
    c_\alpha= \begin{cases}
        0.4 & \mathrm{if} \ \alpha=1 \\
        0.6 & \mathrm{if} \ \alpha=0.9 \\
        0.8 & \mathrm{if} \ \alpha=0.8 \\
        1 & \mathrm{otherwise}
    \end{cases},
\end{align}
and $N_2=N - (N_1+N_3)$. Here, scaling constant $c_\alpha$ is chosen so that $N_2$ and $N_3$ are not too small when $\alpha=1$ and $N_3$ is strictly increasing in $\alpha$. Group 3 is the smallest group unless $\alpha=1$. For DGPs 1 and 3, we apply within-transformation to the data before conducting K-means clustering to demonstrate that the theoretical results derived in the present work are robust to the presence of individual fixed-effects and the routinely applied transformation, as long as the regressors does not include a lagged dependent variable.\footnote{If the regressors include a lagged dependent variable as in DGP 2, and if the model includes individual fixed-effects, then one will need to apply bias correction or IV estimation; see, e.g., \citet{hahnAsymptoticallyUnbiasedInference2002}, \citet{phillipsBiasDynamicPanel2007}, and the supplementary material to \citet{bonhommeGroupedPatternsHeterogeneity2015}, Section S.4.1. Models that include both individual fixed-effects and a lagged dependent variable on the right-hand side are beyond our scope, so we do not apply within-transformation for DGP 2.} 
The number of replications is 100.

In the computational aspect, we follow Algorithm 1 of \citet{bonhommeGroupedPatternsHeterogeneity2015} to numerically solve \eqref{def:LS_estimator} and \eqref{def:LS_estimator_gfe}. Specifically, we employ the following iterative algorithm:
\begin{algo}\ 

    \noindent1. Randomly generate initial group memberships $\boldsymbol{\gamma}_N^{(K,0)}=(k_1^{(K,0)},\ldots,k_N^{(K,0)})\in[K]^N$ from the uniform distribution on $[K]$. Set $s=0$.

    \noindent2. Compute
    \begin{align}
        (\boldsymbol{\theta}^{(K,s+1)},\boldsymbol{\mu}^{(K,s+1)}) = \argmin_{(\boldsymbol{\theta},\boldsymbol{\mu})\in\Theta^K\times\mathcal{M}^{KT}} \sum_{i=1}^N\sum_{t=1}^T\left(y_{it}-x_{it}'\theta_{k_i^{(K,s)}} - \mu_{k_i^{(K,s)}t}\right)^2
    \end{align}

    \noindent3. Update $\boldsymbol{\gamma}_N^{(K,s)}$ by
    \begin{align}
        k_i^{(K,s+1)} = \argmin_{k\in[K]}\sum_{t=1}^T\left(y_{it}-x_{it}'\theta_{k}^{(K,s+1)} - \mu_{kt}^{(K,s+1)}\right)^2
    \end{align}
    for $i\in[N]$.

    \noindent4. If $\boldsymbol{\gamma}_N^{(K,s+1)} = \boldsymbol{\gamma}_N^{(K,s)}$ (after suitable relabeling), then terminate the iteration. Otherwise, set $s=s+1$ and go to Step 2.
    \label{algo:k-means}
\end{algo}
\noindent Since it is well-known that the above iterative procedure depends on the initial value and can converge to a local minimum, we repeat Algorithm \ref{algo:k-means} over 1000 initializations.

\subsection{Finite-sample performance of information criteria}

For DGPs 1-2, we analyze the performance of $\widehat{K}(h_{NT}^{\mathrm{BN}})$, $\widehat{K}(h_{NT}^{\mathrm{BIC}})$, and $\widehat{K}(h_{NT}^{\mathrm{MIC}1})$. For DGP 3, we consider $\widehat{K}(h_{NT}^{\mathrm{BN}})$, $\widehat{K}(h_{NT}^{\mathrm{BIC}})$, and $\widehat{K}(h_{NT}^{\mathrm{MIC}2})$. Information criteria are calculated over $2\leq K \leq K_{\max} = 5$. 

\subsubsection{Small $T$ case}\label{subsubsec:mc_smallT}

First, we study the performance of information criteria when $T$ is small relative to $N$. Specifically, we consider $T\in\{10,20\}$. Figure \ref{fig:dgp1_small_T} displays the mean of estimated number of groups for each penalty as a function of $\alpha$, calculated for DGP 1. When $h_{NT}^{\mathrm{BN}}$ is used, $\widehat{K}(h_{NT}^{\mathrm{BN}})$ correctly estimates the number of groups to be 3 on average for $\alpha\geq 0.7$ and $T=10$, but it tends to underestimate when $\alpha\leq 0.6$ and $T=10$. This tendency to underestimate is somewhat mitigated when $T=20$, but $\widehat{K}(h_{NT}^{\mathrm{BN}})$ still underestimates the number of groups for $\alpha\leq0.5$. This implies that $\widehat{K}(h_{NT}^{\mathrm{BN}})$ tends to overlook small groups. It is noteworthy that, if $T$ is fixed, the problem of underestimation is more severe for larger $N$, implying that only increasing the number of cross-sections does not improve (in fact, aggravates) the estimation of the number of groups. These observations corroborate the theoretical analysis given in Section \ref{sec:without_gfe}. In contrast, $\widehat{K}(h_{NT}^{\mathrm{BIC}})$ overestimates the number of groups on average for all $(N,T,\alpha)$. In particular, it estimates the number of groups to be 5 in \textit{all} replications for $N=90$ and $N=120$ (irrespective of the value of $T$). For the MIC, $\widehat{K}(h_{NT}^{\mathrm{MIC}1})$ shows reasonably good performance for all $(N,T,\alpha)$. Although it can overestimate or underestimate when $\alpha$ is small or $T$ is small relative to $N$, its mean stays within a close neighborhood of the true number 3, unlike $\widehat{K}(h_{NT}^{\mathrm{BN}})$ and $\widehat{K}(h_{NT}^{\mathrm{BIC}})$.

Figure \ref{fig:dgp2_small_T} shows the simulation results for DGP 2. The general tendency is the same as in the case of DGP 1. $\widehat{K}(h_{NT}^{\mathrm{BN}})$ performs well only when $\alpha$ is close to 1 and $T$ is not too small compared to $N$, and tends to underestimate the number of groups otherwise. $\widehat{K}(h_{NT}^{\mathrm{BIC}})$ overestimates for all $(N,T,\alpha)$. $\widehat{K}(h_{NT}^{\mathrm{MIC}1})$ can correctly estimate the number of groups on average.

In Figure \ref{fig:dgp3_small_T}, we show the performance of information criteria for DGP 3. Note that $h_{NT}^{\mathrm{BIC}}=h_{NT}^{\mathrm{MIC}2}$ because we consider the model with GFE and cases with $N>T$. As predicted in Section \ref{sec:with_gfe}, $\widehat{K}(h_{NT}^{\mathrm{BN}})$ underestimates the number of groups for all $(N,T,\alpha)$. In fact, it estimates the number of groups to be 2 in \textit{all} replications. For $h_{NT}^{\mathrm{BIC}}$, $\widehat{K}(h_{NT}^{\mathrm{BIC}})$ performs well for $(N,T)=(60,10)$ but tends to overestimate for $(N,T)=(90,10)$ and $(N,T)=(120,10)$, with this tendency stronger for larger $N$. Therefore, $\widehat{K}(h_{NT}^{\mathrm{BIC}})$ can overestimate the number of groups when $T$ is small relative to $N$, although this tendency is weaker than in models without GFE.\footnote{Overestimation by the BIC estimator in the GFE model with $N$ large relative to $T$ is also reported in \citet{janysMentalHealthAbortions2024}.} When $T$ increases to 20, $\widehat{K}(h_{NT}^{\mathrm{BIC}})$ shows good performance for all $(N,T,\alpha)$, indicating that the BIC-type penalty is a reasonable choice in models with GFE when $T$ is not too small compared to $N$.

\subsubsection{Large $T$ case}

In this subsection, we analyze the behavior of $\widehat{K}(h_{NT})$ when $T$ is larger than $N$. To do this, we consider cases with $T/N\in\{1.5,3\}$ for each DGP and $N$, and fix the size of the third group at $\alpha=0.3$. Because $h_{NT}^{\mathrm{MIC}1}=h_{NT}^{\mathrm{BN}}$ for DGPs 1-2, and $h_{NT}^{\mathrm{MIC}2}$ is asymptotically equivalent to $h_{NT}^{\mathrm{BIC}}$ for DGP 3, we only consider $\widehat{K}(h_{NT}^{\mathrm{BN}})$ and $\widehat{K}(h_{NT}^{\mathrm{BIC}})$ in this subsection. The results are summarized in Table \ref{tab:large_T}.

For DGP 1, we observe three facts. First, $\widehat{K}(h_{NT}^{\mathrm{BN}})=\widehat{K}(h_{NT}^{\mathrm{MIC}1})$ is close or equal to the true value 3 on average thanks to the simulation design where $T$ is larger than $N$. Second, increasing $N$ and $T$ \textit{simultaneously} and increasing $T$ only (with $N$ fixed) improve the performance of $\widehat{K}(h_{NT}^{\mathrm{BN}})=\widehat{K}(h_{NT}^{\mathrm{MIC}1})$. Third, the mean of $\widehat{K}(h_{NT}^{\mathrm{BIC}})$ approaches the true value $K^0=3$ as $T$ increases with $N$ fixed but approaches the upper bound, $K_{\max}=5$, as $N$ increases. This indicates that the overestimation by $\widehat{K}(h_{NT}^{\mathrm{BIC}})$ cannot be solved even in moderately large samples if GFE is not incorporated into the model. The same observations hold for DGP 2.

In DGP 3, which includes GFE, $\widehat{K}(h_{NT}^{\mathrm{BN}})$ underestimates the number of groups for all $(N,T)$ considered. For the BIC, $\widehat{K}(h_{NT}^{\mathrm{BIC}})$ precisely estimates the number of groups to be 3 on average. This, in conjunction with the results given in Section \ref{subsubsec:mc_smallT}, leads us to conclude that $h_{NT}=h_{NT}^{\mathrm{BIC}}$ is a reasonable choice in the GFE model \eqref{model:with_gfe} if $T$ is not too small compared to $N$. However, recall that $\widehat{K}(h_{NT}^{\mathrm{BIC}})$ can underestimate the number of groups if $T$ is too large relative to $N$, in which case $h_{NT}=h_{NT}^{\mathrm{MIC}2}$ is preferred; see Section \ref{subsec:ic_gfe}.

\subsection{Performance of the K-means estimator given $K=K^0$}

Next, we examine how the group size affects the performance of the K-means estimator calculated given the correct number of groups. Since (true) Groups 1 and 2 have sizes proportional to $N$ for all $\alpha\in\{0.2,\ldots,1\}$, we present simulation results for Group 3 only. For each DGP, we calculate the root mean squared error (RMSE). Specifically, for DGPs 1-2, the RMSE is defined by
\begin{align}
    \mathrm{RMSE} &= \left[\frac{1}{N_3}\sum_{i\in G_3^0}\left\{\left(\widehat{\theta}_{\widehat{k}_i,1} - \theta_{3,1}\right)^2 + \left(\widehat{\theta}_{\widehat{k}_i,2} - \theta_{3,2}\right)^2\right\}\right]^{1/2},
\end{align}
and for DGP 3,
\begin{align}
    \mathrm{RMSE} &= \left[\frac{1}{N_3}\sum_{i\in G_3^0}\left\{\left(\widehat{\theta}_{\widehat{k}_i,1} - \theta_{3,1}\right)^2 + \left(\widehat{\theta}_{\widehat{k}_i,2} - \theta_{3,2}\right)^2\right\} + \frac{1}{N_3T}\sum_{i\in G_3^0}\sum_{t=1}^T\left(\widehat{\mu}_{\widehat{k}_it} - \mu_{3t}\right)^2\right]^{1/2}.
\end{align}
Figure \ref{fig:rmse} reports the mean RMSE over 100 replications. For DGP 1 (Figure \ref{fig:rmse_dgp1}), there are three points to be noted. First, increasing $N$ with $T$ fixed improves the estimation accuracy when $\alpha\geq0.7$, but this improvement is only marginal. Second, increasing $N$ while fixing $T$ can lead to a larger estimation error for small $\alpha$. In particular, $N=120$ gives by far the largest RMSE on average when $T=10$ and $\alpha=0.2$. Third, this large estimation error can be mitigated by increasing $T$. These observations are consistent with our theoretical results. The same comment applies to the cases of DGPs 2-3.

We also evaluate the effect of the presence of a small group on classification performance. We calculate the proportion of perfect classification (PPC) defined as
\begin{align}
    \mathrm{PPC} = \frac{\text{The number of times when all $N$ units are classified correctly}}{\text{The number of replications}}.
\end{align}
See Figure \ref{fig:ppc}. We observe three facts. First, increasing $N$ while fixing $T=10$ decreases PPC irrespective of the value of $\alpha$. Second, for models without GFE, a smaller $\alpha$ can lead to a smaller PPC when $T=10$, but the effect of $\alpha$ is not remarkable. In contrast, PPC is highly sensitive to the group size when $T=10$ in the model with GFE. Third, PPC takes values close to 1 for all $\alpha$ when $T=20$, indicating that the group size has a minor effect on classification performance if $T$ is not too small relative to $N$.

\section{Empirical Application}\label{sec:empirics}

In this section, we apply K-means clustering to investigate the determinants of sales growth of Japanese firms. Our empirical application is inspired by \citet{lumsdaineEstimationPanelGroup2023}, who explain that sales growth is a key measure of corporate performance and hence it is important for economists, investors, and decision makers to understand how firms' variables affect sales growth. In particular, the authors are interested in the relationship between leverage and sales growth because there are competing views on the significance of leverage in relation to corporate performance \citep{myersDeterminantsCorporateBorrowing1977,millerLeverage1991}.

As \citet{lumsdaineEstimationPanelGroup2023} argue, it is documented in the literature that firms' responses to market-wide shocks exhibit a group pattern of heterogeneity. Motivated by this fact, we study the determinants of sales growth of Japanese firms using the following model:
\begin{align}
    y_{it} = x_{i,t-1}'\theta_{k_i} + \eta_i + \varepsilon_{it},
\end{align}
where $y_{it}$ is sales growth, $x_{it}$ is the set of regressors, and $\eta_i$ is the individual fixed effects. Following \citet{lumsdaineEstimationPanelGroup2023}, we let $x_{it}$ include leverage (LEV), the logarithm of total assets (TA), Tobin's q (TQ), cash flow (CF), property, plant and equipment (PPE), and return on assets (ROA). Lagged regressors $x_{i,t-1}$ are used in the regression to alleviate possible endogeneity \citep{lumsdaineEstimationPanelGroup2023}.

The variables used in this section are obtained from the eol database. After excluding firms that contain missing variables and outlying sales growth,\footnote{Outliers are defined as sales growth that fall outside $3$ standard deviations neighborhood of its mean.} 417 firms remain in our dataset. The sample period is 2005-2020.

We first determine the number of groups using the information criterion, $\mathrm{IC}(K,h_{NT})$. $\mathrm{IC}(K,h_{NT})$ is calculated over $2\leq K \leq K_{\max} = 8$ using different choices for $h_{NT}$. Table \ref{tab:empirics_group_size} shows the selected number of groups and the sizes of the estimated groups obtained from each $\widehat{K}(h_{NT})$-group clustering. When $h_{NT}=h_{NT}^{\mathrm{BN}}$, the estimated number of groups is 2, the minimum possible value, while $h_{NT}=h_{NT}^{\mathrm{BIC}}$ yields the upper bound value $\widehat{K}(h_{NT}^{\mathrm{BIC}})=8$, as expected from our numerical experiment. In contrast, our proposed $h_{NT}=h_{NT}^{\mathrm{MIC}1}$ leads to a 3-group model. Furthermore, the smallest group (labeled Group 3) contains only 3 firms. The fact that this small group is overlooked when we use $h_{NT}=h_{NT}^{\mathrm{BN}}$ is consistent with the analyses given in the preceding sections. Although the 8-group model suggested by the BIC leads to unbiased estimation of group heterogeneity (as long as $8\geq K^0$), the estimates may be unnecessarily inefficient (unless $K^0=8$) and interpretation of the group structure is difficult. Because $h_{NT}=h_{NT}^{\mathrm{MIC}1}$ leads to a reasonably parsimonious model and allows us to detect a small group, we will now proceed with the 3-group model.

Shown in Table \ref{tab:empirics_est_mic} are the slope coefficient estimates. For each regressor, the sign of the estimated coefficient is generally identical across groups, while the magnitude of the estimated slope differs remarkably across groups. For Group 1, all regressors but PPE have a significant effect on sales growth. While leverage, total assets and ROA have a negative effect, Tobin's q affects sales growth positively. The facts that total assets and ROA negatively affect sales growth and that Tobin's q exhibits a positive effect are largely consistent with the results of \citet{lumsdaineEstimationPanelGroup2023}, who investigate the determinants of sales growth using US firms data. For Group 2, all estimated slopes are significant and of larger magnitude than those of Group 1. For Group 3, only the coefficients on the log of total assets and ROA are significant, probably due to the small sample size for Group 3. Group 3 is notably differentiated from the other groups by the magnitudes of the coefficients on the total assets and ROA. For example, compared with Group 1, the coefficient on the total assets is more than 40 times larger, and that on ROA is more than 30 times larger, which reveals that there is remarkable heterogeneity across groups.

An interesting question arising from the above 3-group model is how we can characterize the small group (Group 3). Group 3 consists of (ferrous and nonferrous) metal producers, and so our clustering analysis reveals that there is a small group in the metal industry. Let us also characterize Group 3 based on observables. Figure \ref{fig::empirics_cluster_variables} presents the time path of the averages of sales growth and regressors, where the averages are taken within each group. Up to the sampling error (mainly due to the extremely small sample size of Group 3 with $N_3=3$), there are notable differences between the group averages of total assets and PPE; namely, Group 3 can be characterized by high total assets and low PPE, or a low PPE/total assets ratio.

\begin{rem}
    We also ran the regression model with GFE. To determine the number of groups, we used $h_{NT}=h_{NT}^{\mathrm{BN}}$ and $h_{NT}=h_{NT}^{\mathrm{BIC}}$ for the information criterion\footnote{Note that $h_{NT}^{\mathrm{MIC2}}=h_{NT}^{\mathrm{BIC}}$ since $N>T$.} and obtained $\widehat{K}(h_{NT}^{\mathrm{BN}})=2$ (the lower bound value) and $\widehat{K}(h_{NT}^{\mathrm{BIC}})=8$ (the upper bound value). Recalling the theoretical and numerical results given in the previous sections, $\widehat{K}(h_{NT}^{\mathrm{BN}})$ certainly underestimates the number of groups, and we suspect that $\widehat{K}(h_{NT}^{\mathrm{BIC}})$ overestimates the number of groups because $T$ is relatively small compared to $N$ in our dataset. One possible solution to the overestimation by the BIC would be to impose a larger penalty on the number of groups, as considered in \citet{janysMentalHealthAbortions2024}. Determining the number of groups in the GFE model with a large $N/T$ ratio is a nontrivial issue and will be an important topic for future research.
\end{rem}

\section{Conclusion}\label{sec:conclusion}

In this article, we analyzed the behavior of the K-means/LS estimator relaxing the conventional assumption of all groups having a size proportional to $N$. Our framework allows for smaller groups whose sizes are of order $N^{\alpha}$ with $\alpha<1$, and we derived conditions under which the LS estimators for the slope parameter and GFE are consistent and asymptotically normal, assuming that the true number of groups is known. In the case of an unknown number of groups, we derived conditions under which the information criterion for the number of groups is inconsistent. Our theoretical and numerical analyses showed that an information criterion proposed by \citet{baiDeterminingNumberFactors2002} can underestimate the number of groups, particularly when there are small groups or the model includes GFE. The BIC proposed by \citet{bonhommeGroupedPatternsHeterogeneity2015} was found to overestimate the number of groups in models without GFE, although it is a reasonable choice if the model includes GFE and $T$ is not too small or too large relative to $N$. We proposed information criteria that are designed not to satisfy the derived conditions for inconsistency and to improve the above precursors. A simulation study confirmed their good performance in finite samples. In the simulation study, we also evaluated the finite-sample performance of the LS estimator given the correct number of groups. We found that increasing $N$ with $T$ fixed does not necessarily improve, or can worsen, classification performance and the accuracy of parameter estimation for small groups, a problem that becomes more severe as the sizes of those groups are smaller. Our empirical application illustrates that our modified information criterion enables detecting and characterizing small groups in a reasonably parsimonious group structure.

\bibliographystyle{standard}

\bibliography{clustering}

\clearpage
\begin{table} \caption{Mean of $\widehat{K}(h_{NT})$, large $T$, $\alpha=0.3$}
	\centering
	\begin{threeparttable}
		\renewcommand{\arraystretch}{1.3}	\begin{tabular}{ccccccccccc}
			\hline & & \multicolumn{2}{c}{DGP 1} & & \multicolumn{2}{c}{DGP 2} & & \multicolumn{2}{c}{DGP 3} & \\
            \cline{2-11} & & $T/N$ & & & $T/N$ & & & $T/N$ & & \\
			\cline{2-11}  & $N$ & 1.5 & 3 & & 1.5 & 3 & & 1.5 & 3 &  \\
            \hline \multirow{3}{*}{$h_{NT}^{\mathrm{BN}}$} & \multicolumn{1}{l}{60} & 2.85 &3 & & 3 & 3 & & 2 & 2 & \\ 
             &\multicolumn{1}{l}{90} & 2.57 &3 & & 3 & 3 & & 2 & 2 & \\ 
             &\multicolumn{1}{l}{120} & 3 & 3 & & 3 & 3 & & 2 & 2 & \\ 
             \hline \multirow{3}{*}{$h_{NT}^{\mathrm{BIC}}$} & \multicolumn{1}{l}{60} & 4.58 &4.34 & &4.57 &4.46 & & 3 & 3 & \\ 
             &\multicolumn{1}{l}{90} & 4.88 &4.75 & &4.96 &4.93 & & 3 & 3 & \\ 
             &\multicolumn{1}{l}{120} & 4.99 &4.96 & & 5 & 5 & & 3 & 3 & \\ 
             \hline 
		\end{tabular}
        \begin{tablenotes}
			\footnotesize
			\item[]Note: Each entry is based on 100 replications.
		\end{tablenotes}
	\end{threeparttable} \label{tab:large_T}
\end{table} 

\clearpage
\begin{table} \caption{Estimated group sizes for different penalties for the information criterion}
	\centering
	\begin{threeparttable}
		\renewcommand{\arraystretch}{1.7}	\begin{tabular}{ccccccccc}
			\hline & 1 & 2 & 3 & 4 & 5 & 6 & 7 & 8 \\
            \hline $h_{NT}^{\mathrm{BN}}$ & 334 & 83 & - & - & - & - &	- & - \\ 
             $h_{NT}^{\mathrm{MIC}1}$ & 320 & 94 & 3 & - & - & - & - & - \\ 
             $h_{NT}^{\mathrm{BIC}}$ & 130 &	92 & 77 & 40 & 39 & 21 & 17 & 1 \\ 
             \hline 
		\end{tabular}
        \begin{tablenotes}
			\footnotesize
			\item[]Note: Group labels are ordered so that Group 1 is the largest and Group $\widehat{K}(h_{NT})$ is the smallest.
		\end{tablenotes}
	\end{threeparttable} \label{tab:empirics_group_size}
\end{table} 

\clearpage
\begin{table} \caption{Slope estimates obtained under $\widehat{K}(h_{NT}^{\mathrm{MIC}1})=3$ groups}
	\centering
	\begin{threeparttable}
		\renewcommand{\arraystretch}{1.6}	\begin{tabular}{ccccccc}
			\hline & LEV & TA & TQ & CF & PPE & ROA \\
            \hline Group 1 & $-0.476^{***}$ &	$-6.402^{***}$ &	$4.448^{***}$ &	$0.000^{***}$ &	$-0.000$ &	$-0.301^{***}$ \\ 
              & (0.119) &	(0.932) &	(0.407) &	(0.000) &	(0.000) &	(0.046) \\ 
             Group 2 & $-0.706^{***}$ &	$-9.793^{***}$ &	$46.018^{***}$ &	$0.000^{***}$ &	$-0.000^{***}$ &	$-2.729^{***}$ \\ 
             & (0.035) &	(2.930) &	(1.954) &	(0.000) &	(0.000) &	(0.145) \\
             Group 3 & $-12.359$ &	$-286.237^{***}$ &	$15.621$ &	$-0.000$ &	$0.003$ &	$-10.162^{**}$ \\
             & (18.183) &	(86.101) &	(47.307) &	(0.000) &	(0.002) &	(4.614) \\
             \hline 
		\end{tabular}
        \begin{tablenotes}
			\footnotesize
			\item[]Note: Group labels are ordered so that Group 1 is the largest and Group 3 is the smallest. Standard errors are in parentheses. $^{***}$, $^{**}$, and $^{*}$ denote significance at 1\%, 5\%, and 10\%, respectively.
		\end{tablenotes}
	\end{threeparttable} \label{tab:empirics_est_mic}
\end{table}

\clearpage
\begin{sidewaysfigure}
	\centering
	    \begin{subfigure}{0.48\columnwidth}
        \includegraphics[width=\columnwidth]{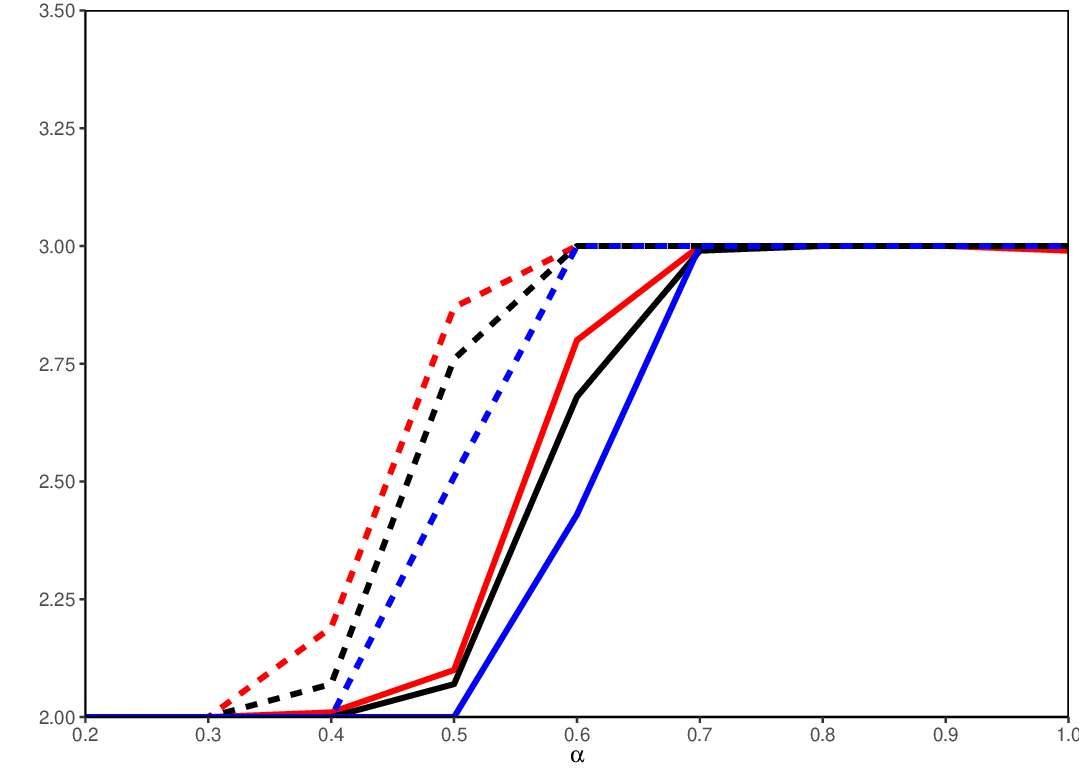}
    \caption{$h_{NT}=h_{NT}^{\mathrm{BN}}$}
	\label{fig:bn_dgp1}
    \end{subfigure} \quad
    \begin{subfigure}{0.48\columnwidth}
        \includegraphics[width=\columnwidth]{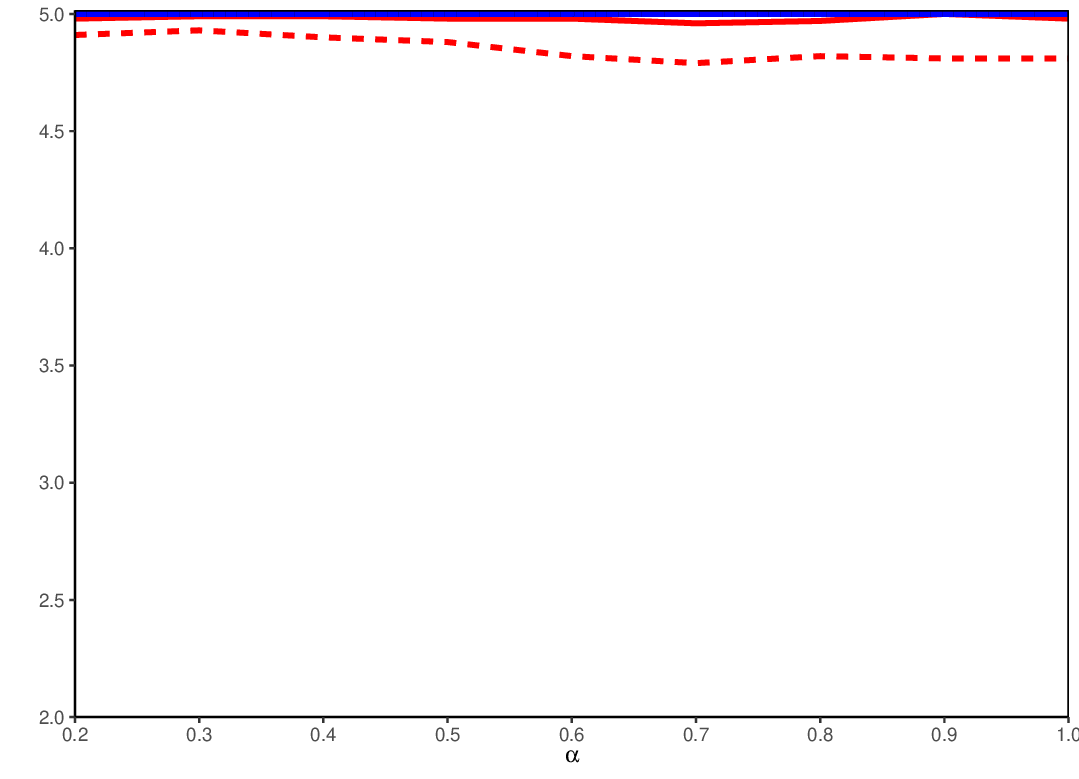}
    \caption{$h_{NT}=h_{NT}^{\mathrm{BIC}}$}
	\label{fig:bic_dgp1}
    \end{subfigure} \quad
    \hfill
    \begin{subfigure}{0.48\columnwidth}
        \includegraphics[width=\columnwidth]{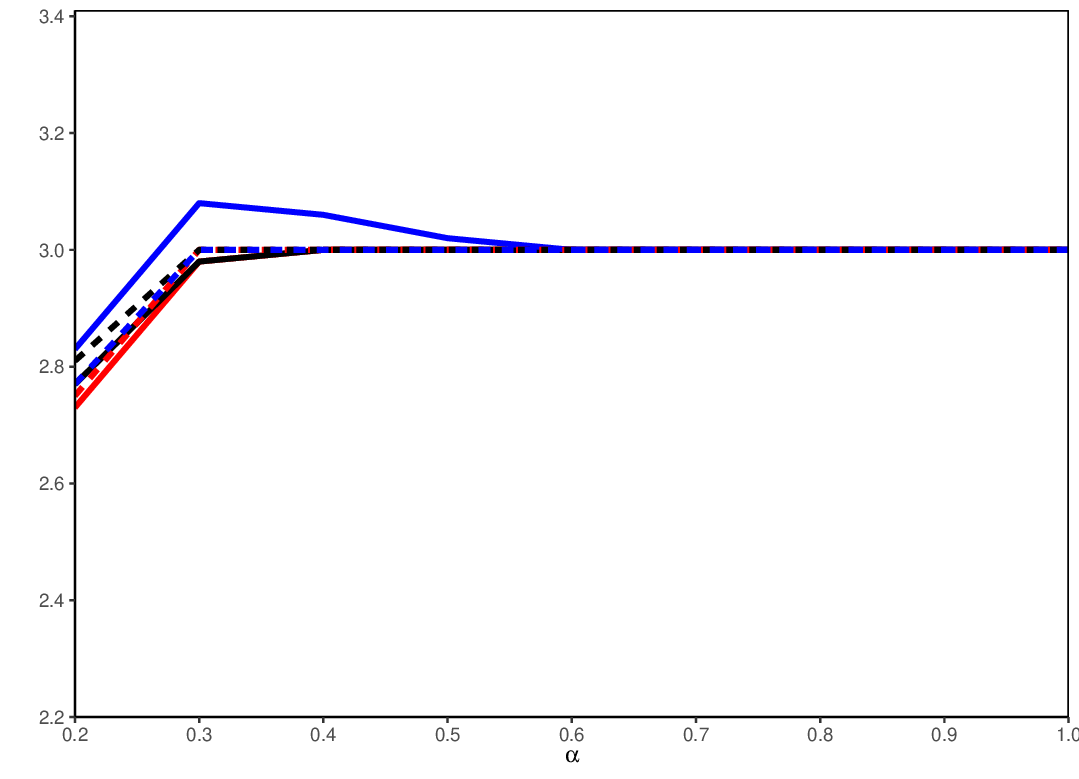}
    \caption{$h_{NT}=h_{NT}^{\mathrm{MIC}1}$}
	\label{fig:mic_dgp1}
    \end{subfigure} \quad
    \caption{Mean of $\widehat{K}(h_{NT})$ over 100 replications, DGP 1}
    \caption*{\color{red}\full\color{black}: $N=60$, \quad \color{black}\full\color{black}: $N=90$, \quad \color{blue}\full\color{black}: $N=120$, \quad \color{black}\full\color{black}: $T=10$, \color{black}\dashed\color{black}: $T=20$}
    \label{fig:dgp1_small_T}
\end{sidewaysfigure}

\clearpage
\begin{sidewaysfigure}
	\centering
	    \begin{subfigure}{0.48\columnwidth}
        \includegraphics[width=\columnwidth]{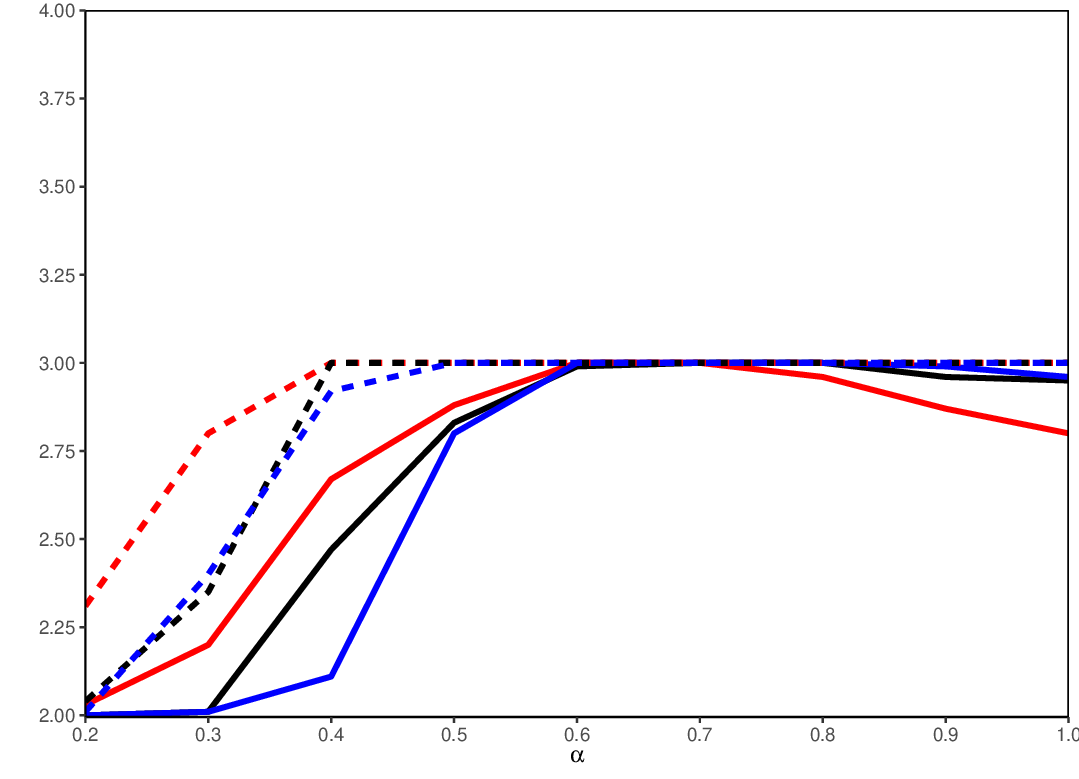}
    \caption{$h_{NT}=h_{NT}^{\mathrm{BN}}$}
	\label{fig:bn_dgp2}
    \end{subfigure} \quad
    \begin{subfigure}{0.48\columnwidth}
        \includegraphics[width=\columnwidth]{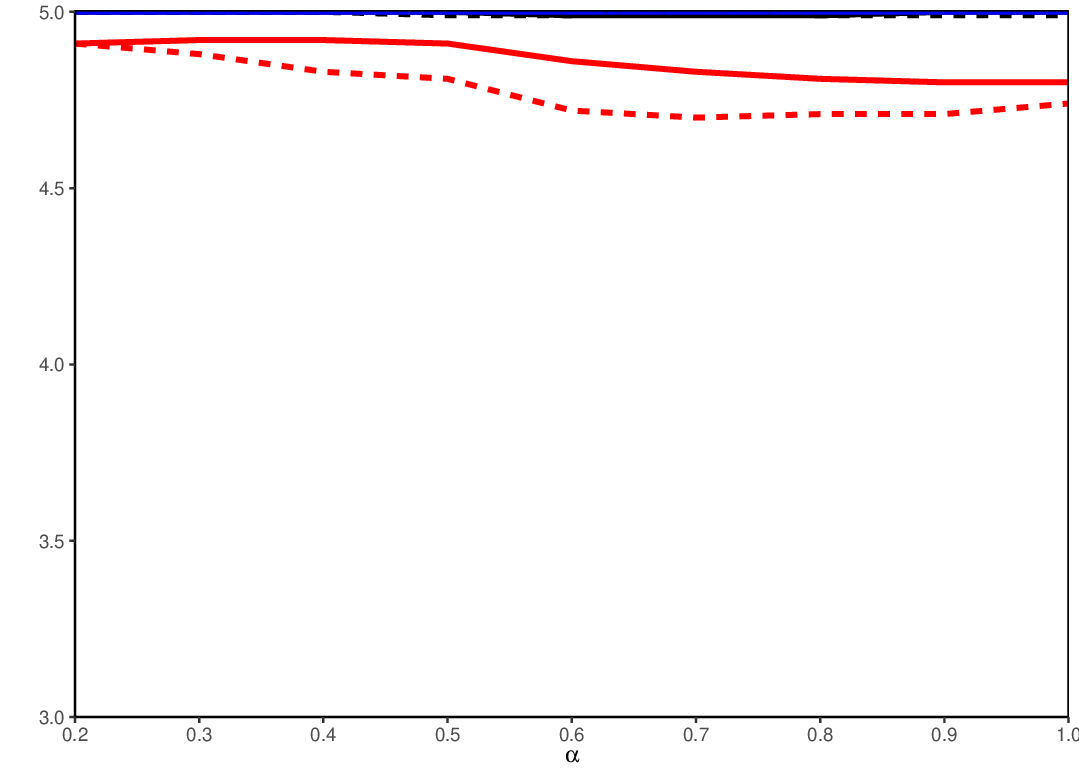}
    \caption{$h_{NT}=h_{NT}^{\mathrm{BIC}}$}
	\label{fig:bic_dgp2}
    \end{subfigure} \quad
    \hfill
    \begin{subfigure}{0.48\columnwidth}
        \includegraphics[width=\columnwidth]{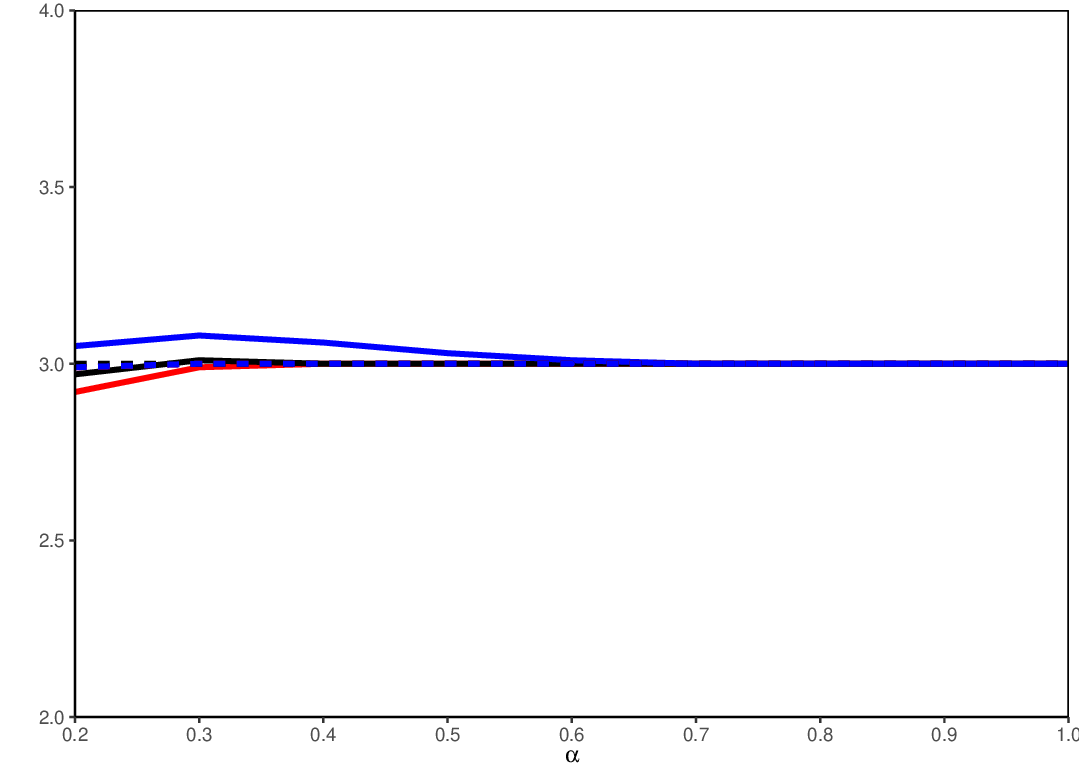}
    \caption{$h_{NT}=h_{NT}^{\mathrm{MIC}1}$}
	\label{fig:mic_dgp2}
    \end{subfigure} \quad
    \caption{Mean of $\widehat{K}(h_{NT})$ over 100 replications, DGP 2}
    \caption*{\color{red}\full\color{black}: $N=60$, \quad \color{black}\full\color{black}: $N=90$, \quad \color{blue}\full\color{black}: $N=120$, \quad \color{black}\full\color{black}: $T=10$, \color{black}\dashed\color{black}: $T=20$}
    \label{fig:dgp2_small_T}
\end{sidewaysfigure}

\clearpage
\begin{figure}
	\centering
	\begin{subfigure}{0.7\columnwidth}
    \includegraphics[width=\columnwidth]{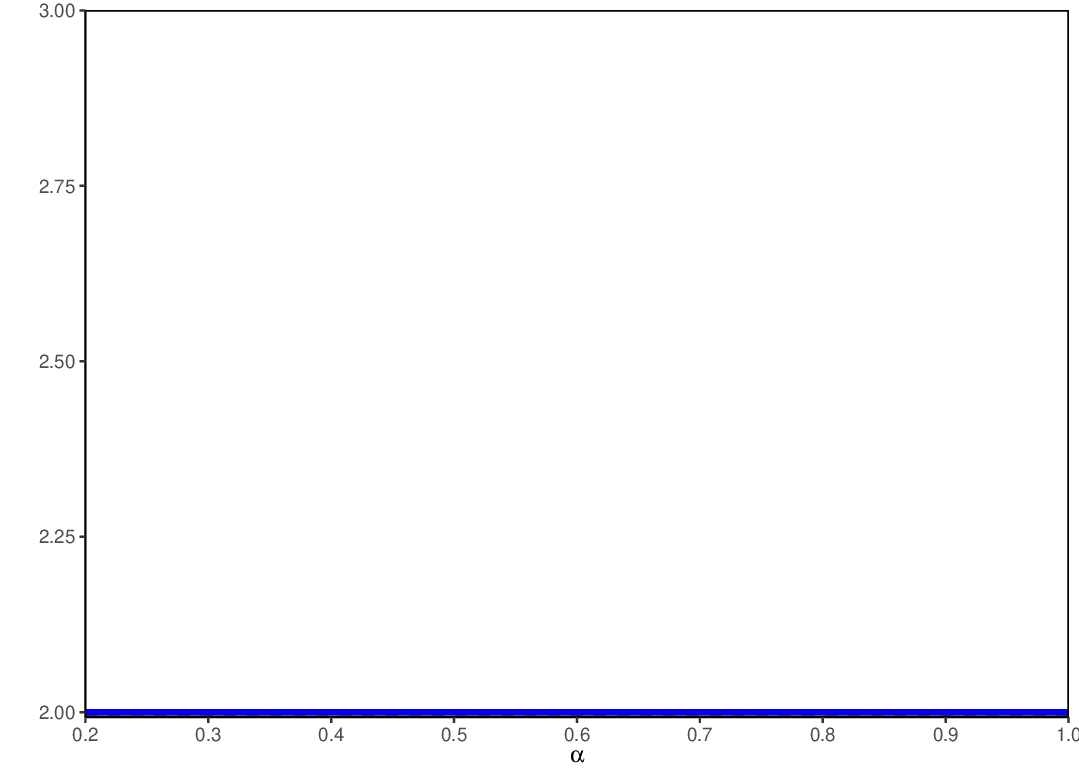}
    \caption{$h_{NT}=h_{NT}^{\mathrm{BN}}$}
	\label{fig:bn_dgp3}
    \end{subfigure}
    \hfill
    \begin{subfigure}{0.7\columnwidth}
        \includegraphics[width=\columnwidth]{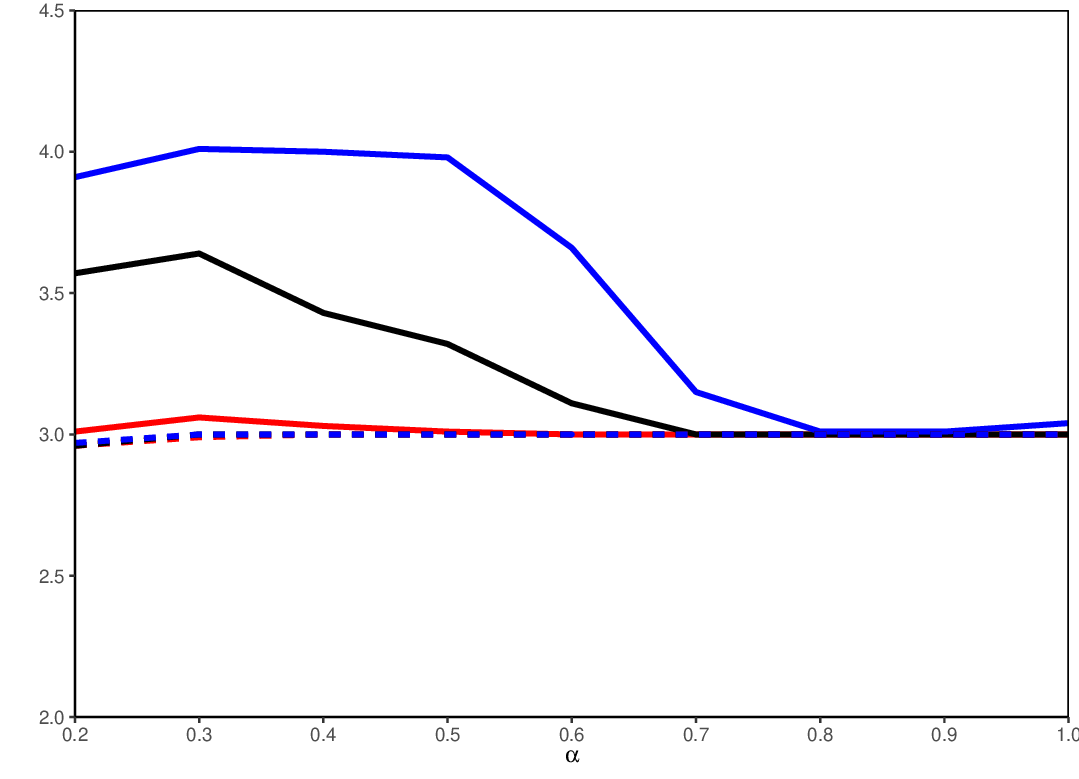}
    \caption{$h_{NT}=h_{NT}^{\mathrm{BIC}}=h_{NT}^{\mathrm{MIC}2}$}
	\label{fig:bic_dgp3}
    \end{subfigure}
    \caption{Mean of $\widehat{K}(h_{NT})$ over 100 replications, DGP 3}
    \caption*{\color{red}\full\color{black}: $N=60$, \quad \color{black}\full\color{black}: $N=90$, \quad \color{blue}\full\color{black}: $N=120$, \quad \color{black}\full\color{black}: $T=10$, \color{black}\dashed\color{black}: $T=20$}
    \label{fig:dgp3_small_T}
\end{figure}

\clearpage
\begin{sidewaysfigure}
	\centering
	    \begin{subfigure}{0.48\columnwidth}
        \includegraphics[width=\columnwidth]{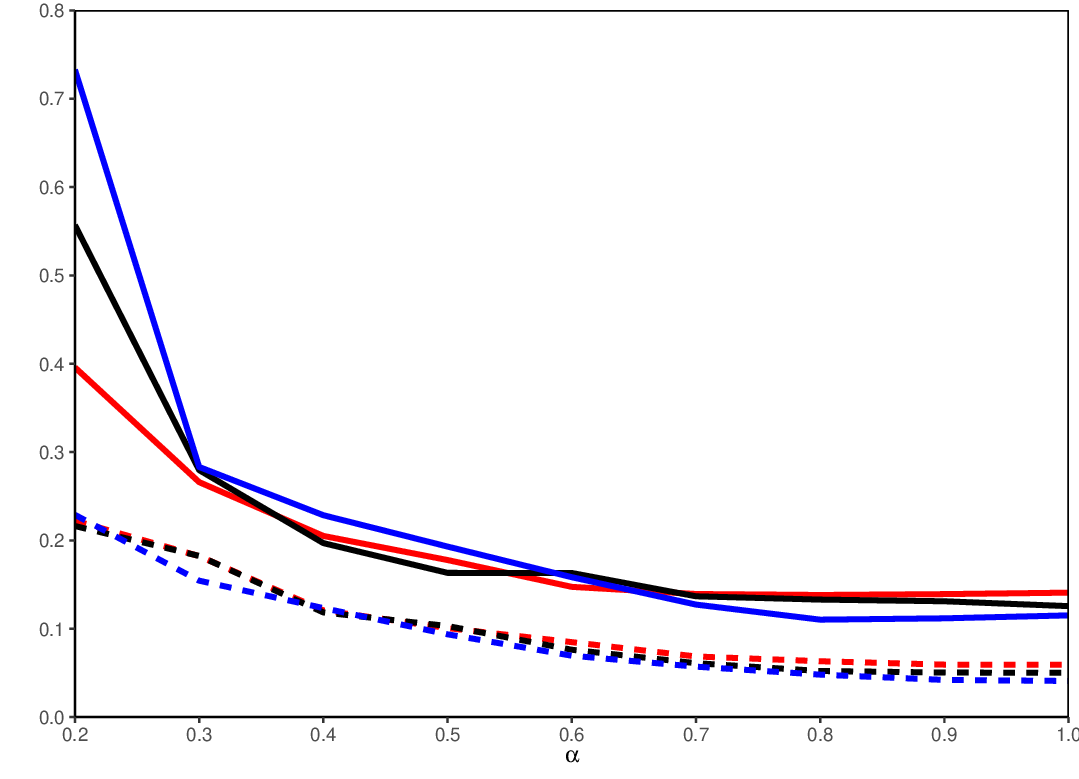}
    \caption{DGP 1}
	\label{fig:rmse_dgp1}
    \end{subfigure} \quad
    \begin{subfigure}{0.48\columnwidth}
        \includegraphics[width=\columnwidth]{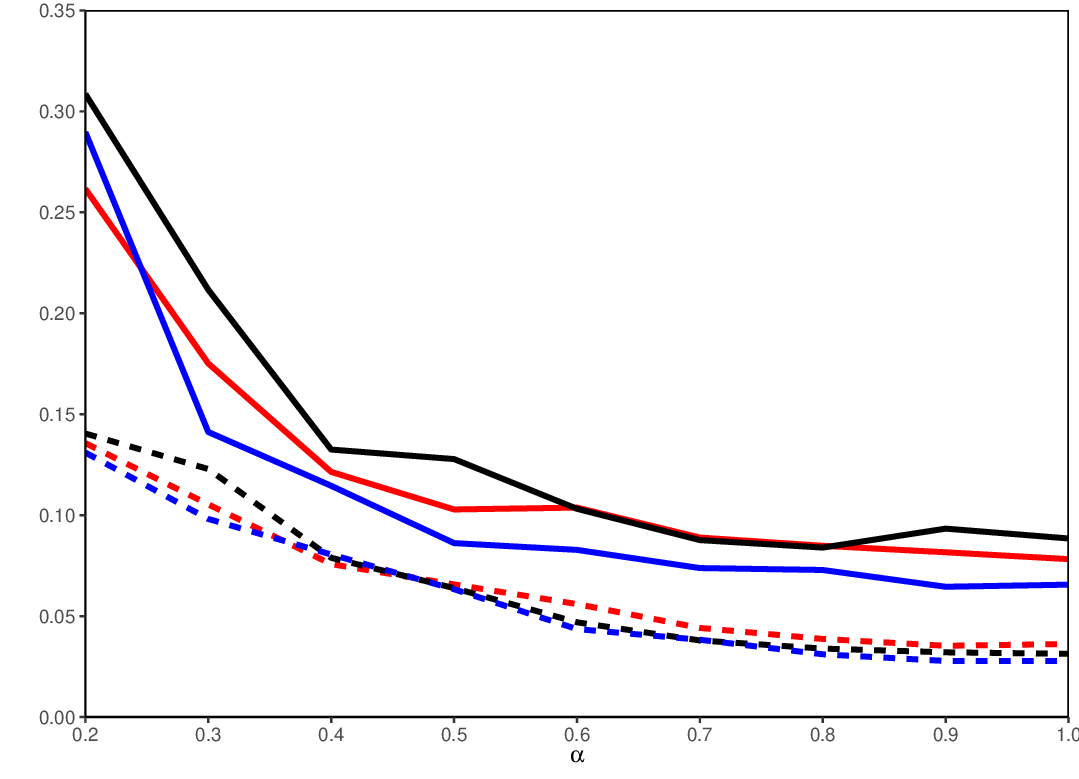}
    \caption{DGP 2}
	\label{fig:rmse_dgp2}
    \end{subfigure} \quad
    \hfill
    \begin{subfigure}{0.48\columnwidth}
        \includegraphics[width=\columnwidth]{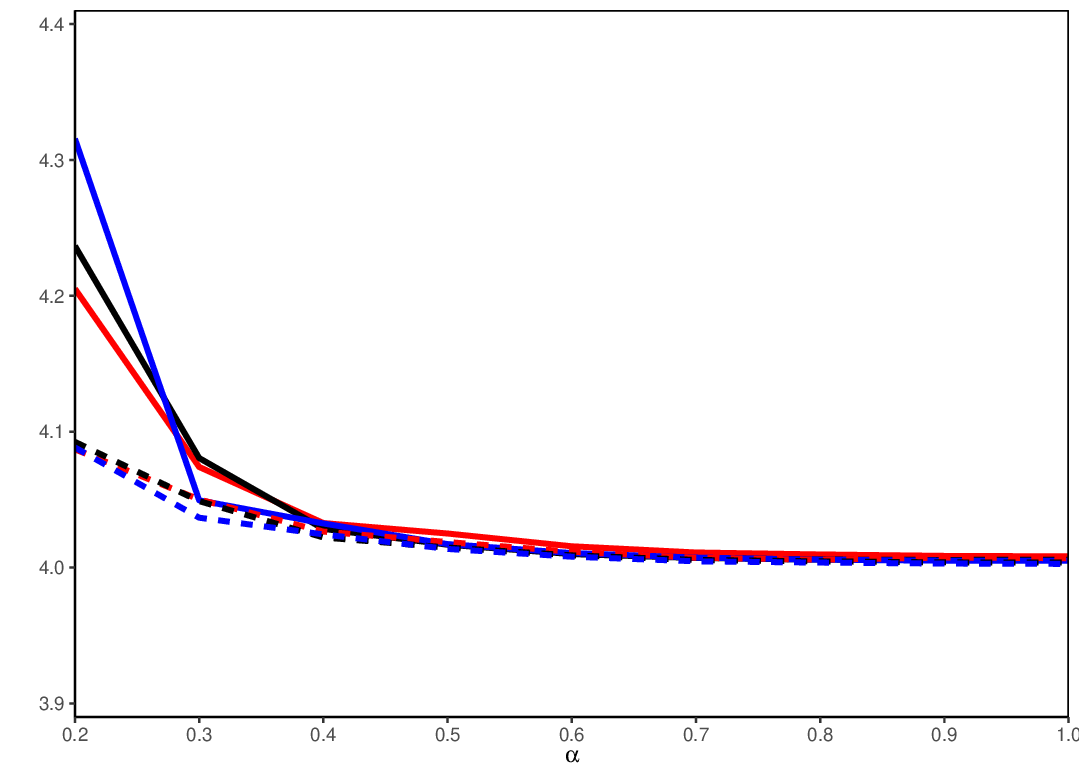}
    \caption{DGP 3}
	\label{fig:rmse_dgp3}
    \end{subfigure} \quad
    \caption{Mean of RMSE over 100 replications, Group 3}
    \caption*{\color{red}\full\color{black}: $N=60$, \quad \color{black}\full\color{black}: $N=90$, \quad \color{blue}\full\color{black}: $N=120$, \quad \color{black}\full\color{black}: $T=10$, \color{black}\dashed\color{black}: $T=20$}
    \label{fig:rmse}
\end{sidewaysfigure}

\clearpage
\begin{sidewaysfigure}
	\centering
	    \begin{subfigure}{0.48\columnwidth}
        \includegraphics[width=\columnwidth]{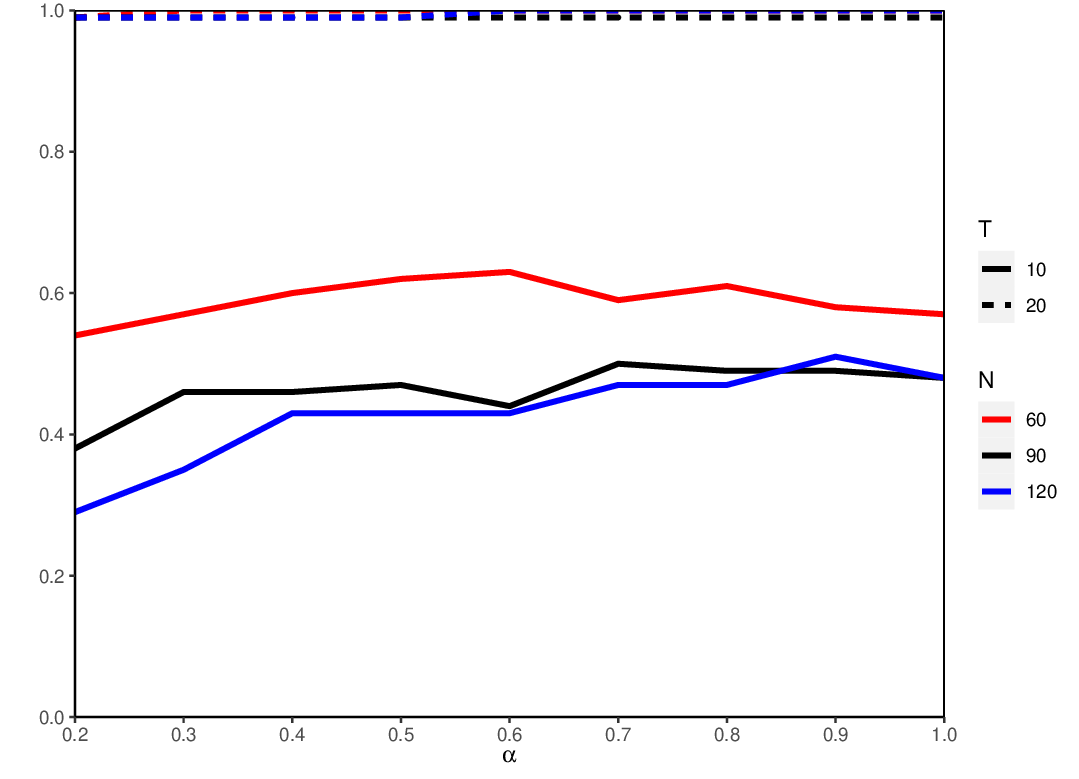}
    \caption{DGP 1}
	\label{fig:ppc_dgp1}
    \end{subfigure} \quad
    \begin{subfigure}{0.48\columnwidth}
        \includegraphics[width=\columnwidth]{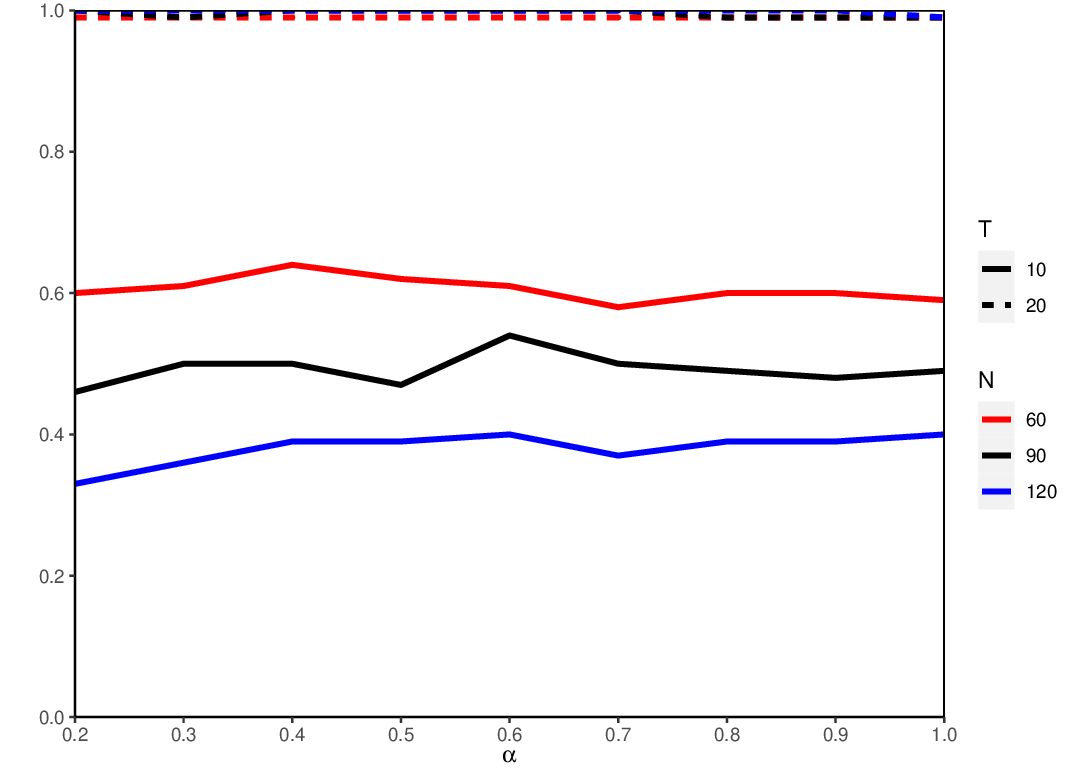}
    \caption{DGP 2}
	\label{fig:ppc_dgp2}
    \end{subfigure} \quad
    \hfill
    \begin{subfigure}{0.48\columnwidth}
        \includegraphics[width=\columnwidth]{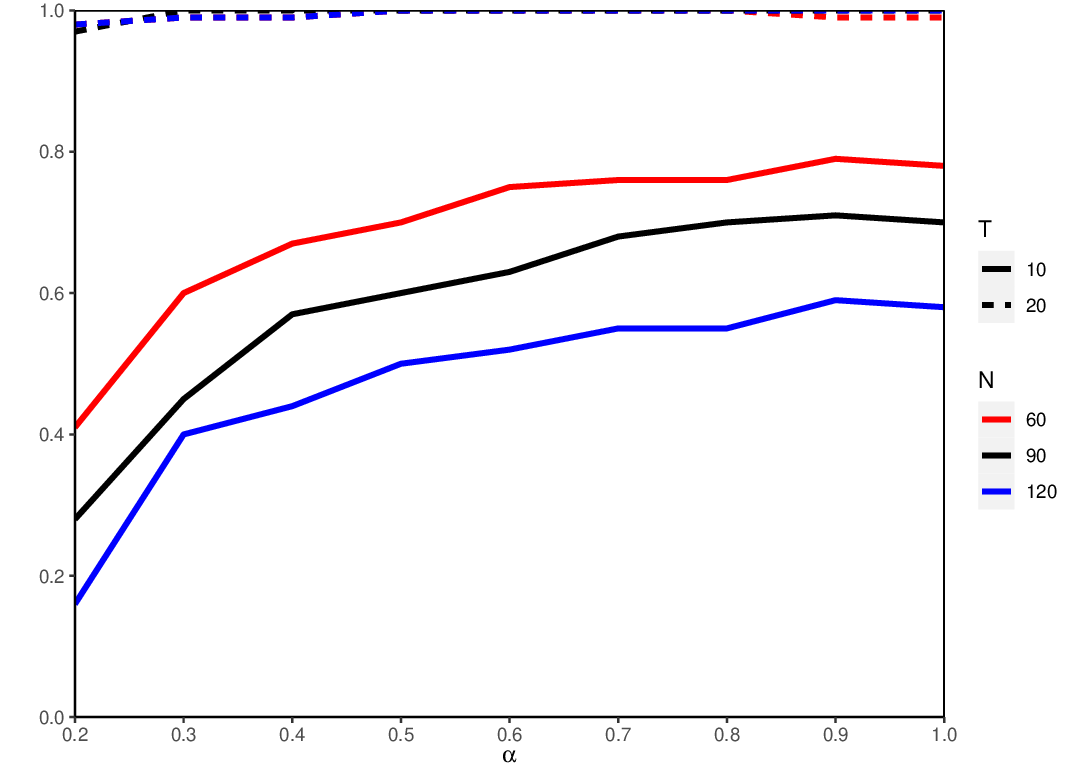}
    \caption{DGP 3}
	\label{fig:ppc_dgp3}
    \end{subfigure} \quad
    \caption{PPC over 100 replications}
    \caption*{\color{red}\full\color{black}: $N=60$, \quad \color{black}\full\color{black}: $N=90$, \quad \color{blue}\full\color{black}: $N=120$, \quad \color{black}\full\color{black}: $T=10$, \color{black}\dashed\color{black}: $T=20$}
    \label{fig:ppc}
\end{sidewaysfigure}

\clearpage
\begin{figure}
        \centering
	\begin{subfigure}{0.27\textheight}
		\includegraphics[width=\columnwidth]{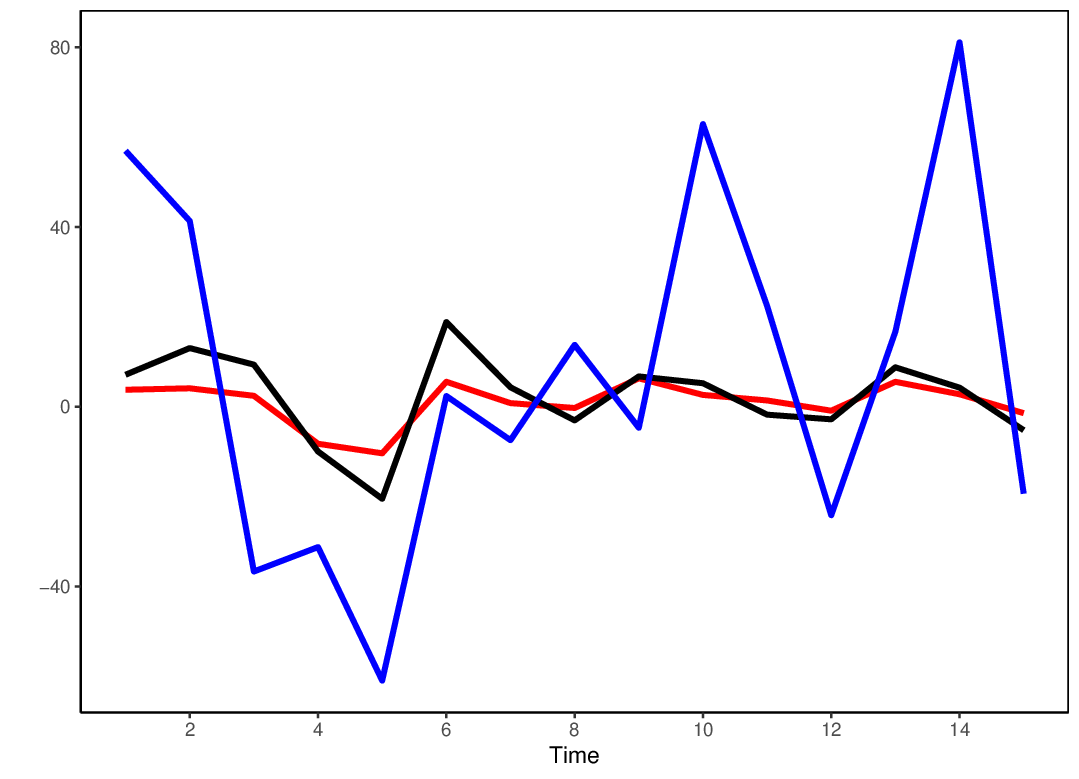}
		\caption{Sales Growth}
	\end{subfigure} \quad
	\begin{subfigure}{0.27\textheight}
		\includegraphics[width=\columnwidth]{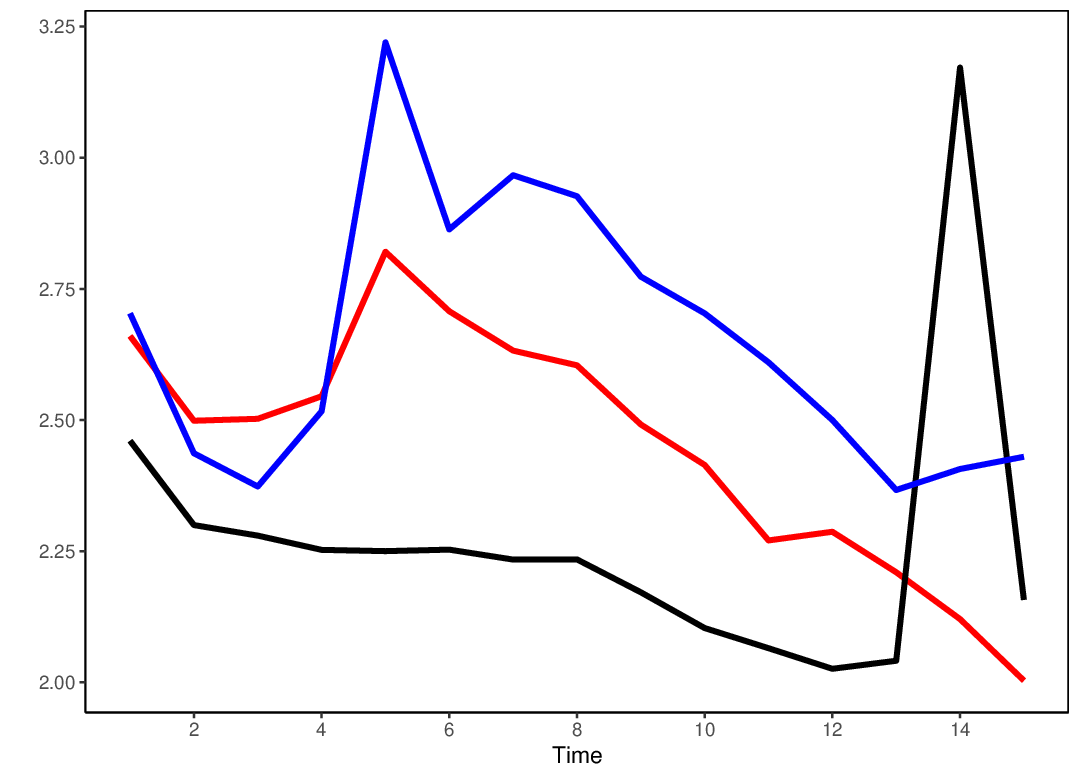}
		\caption{Leverage}
	\end{subfigure} \quad
	\begin{subfigure}{0.27\textheight}
		\includegraphics[width=\columnwidth]{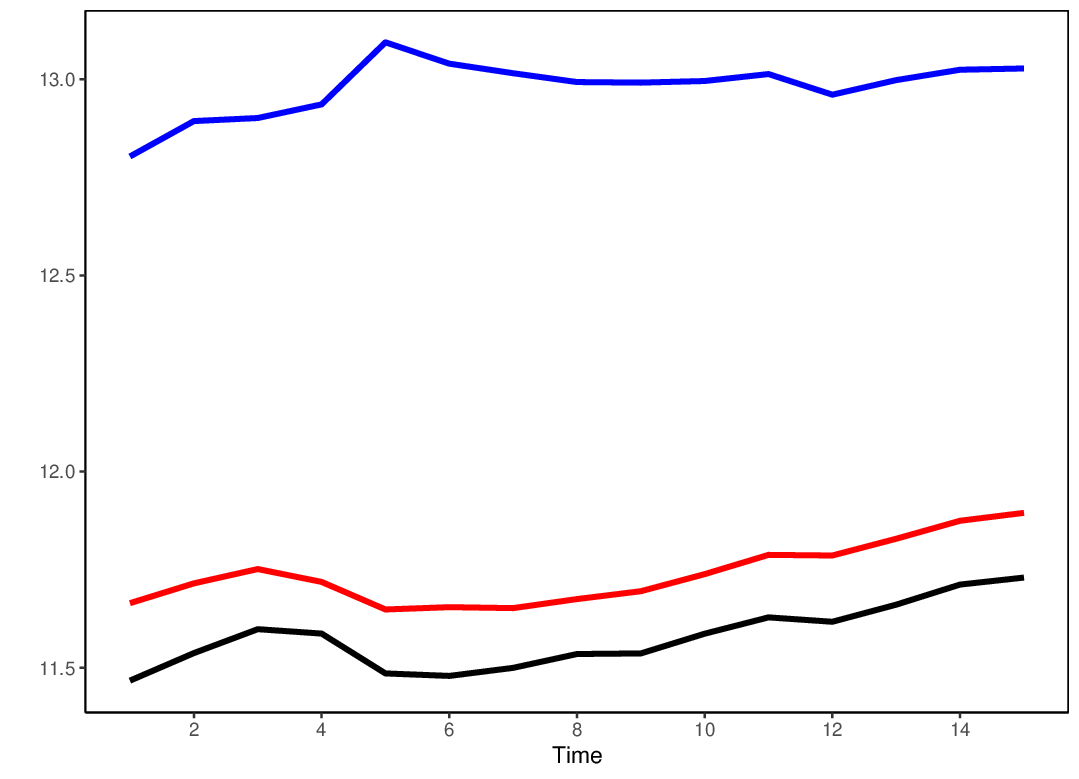}
		\caption{Log of total assets}
	\end{subfigure} \quad
	\begin{subfigure}{0.27\textheight}
		\includegraphics[width=\columnwidth]{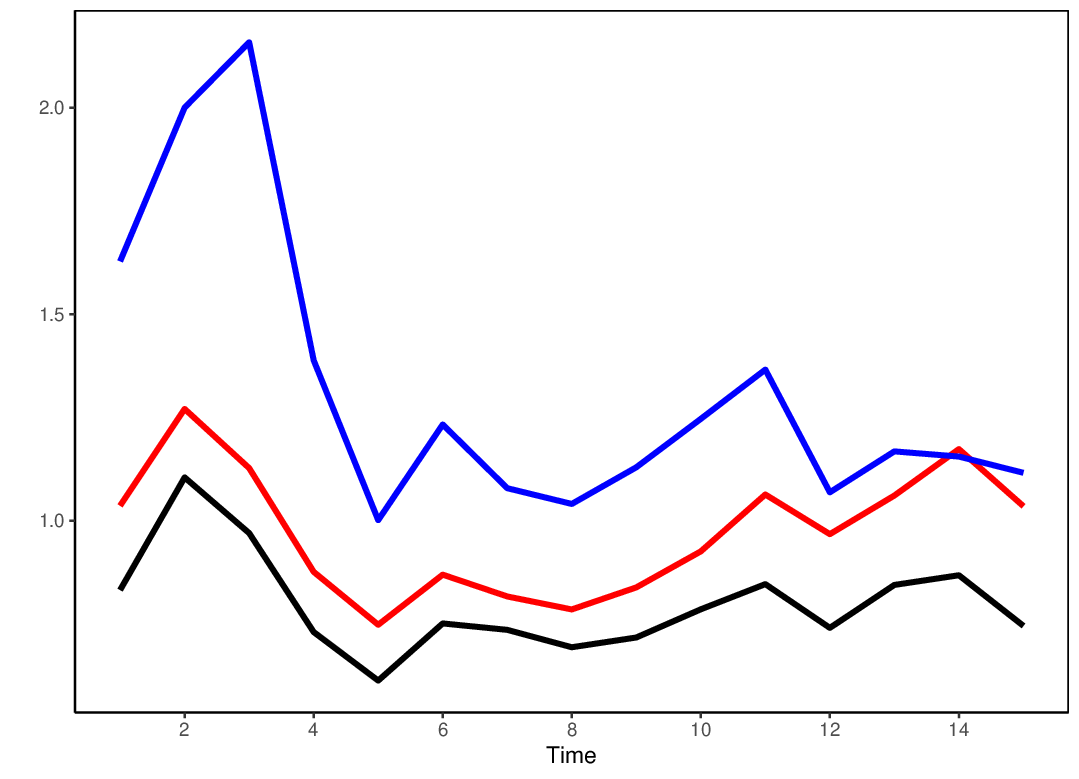}
		\caption{Tobin's q}
	\end{subfigure} \quad
        \centering
	\begin{subfigure}{0.27\textheight}
		\includegraphics[width=\columnwidth]{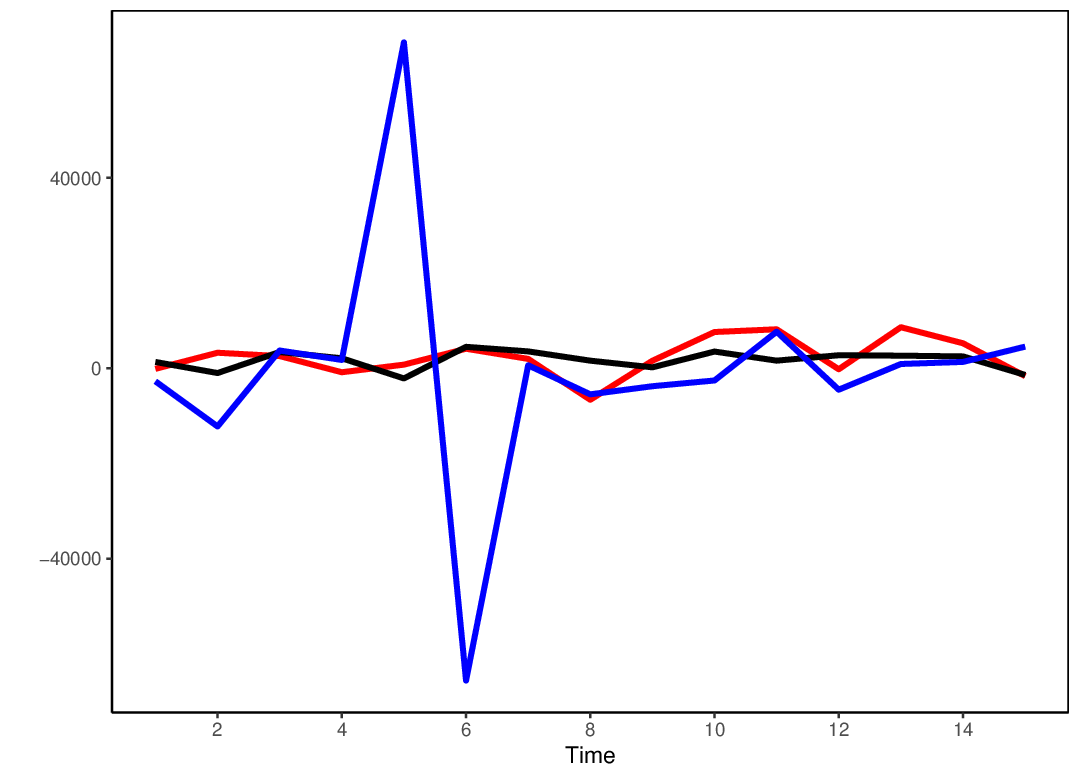}
		\caption{Cash flow}
	\end{subfigure} \quad
	\begin{subfigure}{0.27\textheight}
		\includegraphics[width=\columnwidth]{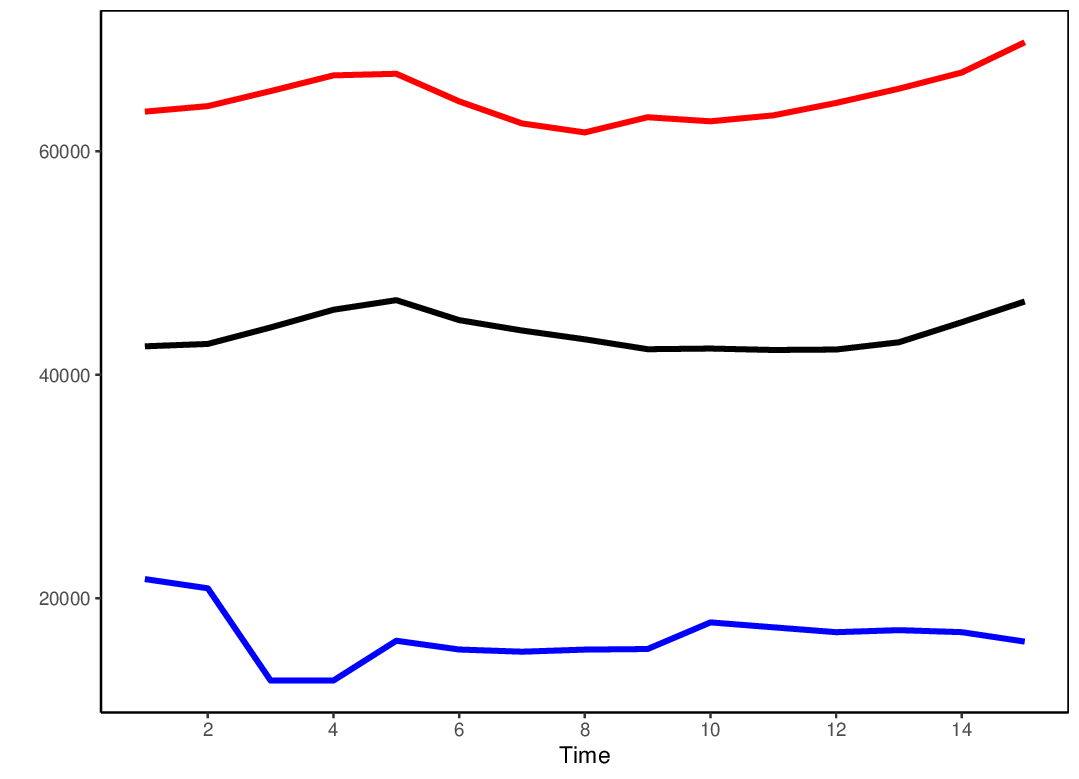}
		\caption{PPE}
	\end{subfigure} \quad
	\begin{subfigure}{0.27\textheight}
		\includegraphics[width=\columnwidth]{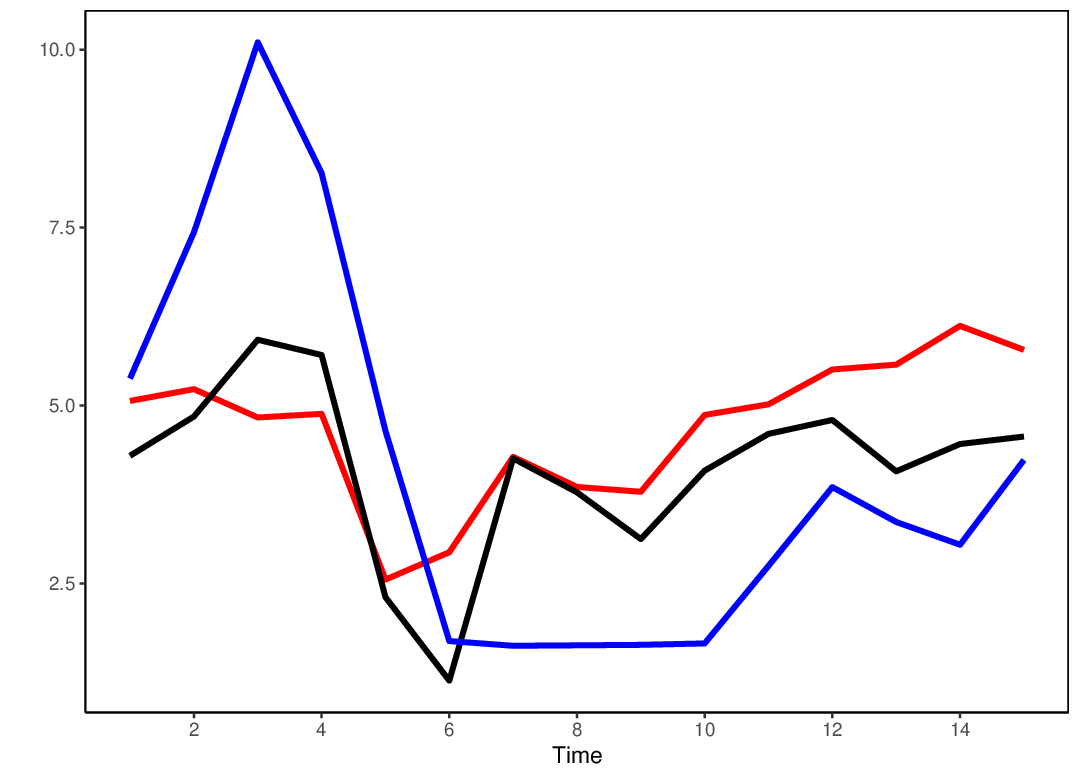}
		\caption{ROA}
	\end{subfigure} \quad
	\caption{Plots of the group averages of variables}
    \caption*{\color{red}\full\color{black}: Group 1, \quad \color{black}\full\color{black}: Group 2, \quad \color{blue}\full\color{black}: Group 3}
	\label{fig::empirics_cluster_variables}
\end{figure}

\clearpage
\appendix
\setcounter{secnumdepth}{0}

\section{Appendix A: Proofs of Results of Section \ref{sec:without_gfe}}

\setcounter{equation}{0}
\renewcommand{\theequation}{A.\arabic{equation}}

\setcounter{lem}{0}
\renewcommand{\thelem}{A.\arabic{lem}}

\setcounter{prop}{0}
\renewcommand{\theprop}{A.\arabic{prop}}

For $(\boldsymbol{\theta},\boldsymbol{\gamma}_N)\in\Theta^{K^0}\times[K^0]^N$, where $\boldsymbol{\gamma}_N = (k_1,\ldots,k_N)$, define 
\begin{align}
    \widehat{Q}(\boldsymbol{\theta},\boldsymbol{\gamma}_N)&\coloneqq (NT)^{-1}\sum_{i=1}^N\sum_{t=1}^T\left(y_{it}-x_{it}'\theta_{k_i}\right)^2,
    \intertext{and}
    \widetilde{Q}(\boldsymbol{\theta},\boldsymbol{\gamma}_N)&\coloneqq (NT)^{-1}\sum_{i=1}^N\sum_{t=1}^T\left\{x_{it}'\left(\theta_{k_i^0}^0-\theta_{k_i}\right)\right\}^2 + (NT)^{-1}\sum_{i=1}^N\sum_{t=1}^T\varepsilon_{it}^2.
\end{align}
Note that $(\boldsymbol{\theta}^0,\boldsymbol{\gamma}_N^0)$ is the minimizer of $\widetilde{Q}(\boldsymbol{\theta},\boldsymbol{\gamma}_N)$.

\begin{lem}
    Consider model \eqref{model:without_gfe}. Suppose that Assumptions \ref{asm:basic}(a) and (d) hold. Then we have
    \begin{align}
        \sup_{(\boldsymbol{\theta},\boldsymbol{\gamma}_N)\in\Theta^{K^0}\times[K^0]^N}\left|\widehat{Q}(\boldsymbol{\theta},\boldsymbol{\gamma}_N) - \widetilde{Q}(\boldsymbol{\theta},\boldsymbol{\gamma}_N)\right| = O_p(1/\sqrt{NT}).
    \end{align}
    \label{lem:Qhat_Qtilde}
\end{lem}

\noindent\begin{proof}
    Since $\widehat{Q}(\boldsymbol{\theta},\boldsymbol{\gamma}_N) - \widetilde{Q}(\boldsymbol{\theta},\boldsymbol{\gamma}_N) = 2(NT)^{-1}\sum_{i=1}^N\sum_{t=1}^T\varepsilon_{it}x_{it}'(\theta_{k_i^0}^0-\theta_{k_i})$ and $\sup\|\theta_k\|<\infty$ by assumption \ref{asm:basic}(a), the claimed result follows from Assumption \ref{asm:basic}(d).
\end{proof}

For $\boldsymbol{\theta},\overline{\boldsymbol{\theta}}\in\Theta^{K^0}$, define the Hausdorff distance between $\boldsymbol{\theta}$ and $\overline{\boldsymbol{\theta}}$:
\begin{align}
    d_H\left(\boldsymbol{\theta},\overline{\boldsymbol{\theta}}\right)\coloneqq\max\left\{\max_{k\in[K^0]}\min_{\widetilde{k}\in[K^0]}\left\|\theta_{\widetilde{k}}-\overline{\theta}_k\right\|,\max_{\widetilde{k}\in[K^0]}\min_{k\in[K^0]}\left\|\theta_{\widetilde{k}}-\overline{\theta}_k\right\|\right\}.
\end{align}

\begin{lem}
    Consider model \eqref{model:without_gfe}. Suppose that Assumptions \ref{asm:group_size} and \ref{asm:basic}(a)-(h) hold. Then we have $d_H\left(\widehat{\boldsymbol{\theta}}^{(K^0)},\boldsymbol{\theta}^0\right) \stackrel{p}{\to} 0$.
    \label{lem:theta_consistency}
\end{lem}

\noindent\begin{proof}
    The proof closely follows that of Lemma S.4 of \citet{bonhommeGroupedPatternsHeterogeneity2015}. Using Assumption \ref{asm:basic}(e), we have
    \begin{align}
        \widetilde{Q}\left(\boldsymbol{\theta},\boldsymbol{\gamma}_N\right) - \widetilde{Q}\left(\boldsymbol{\theta}^0,\boldsymbol{\gamma}_N^0\right) 
        &= \frac{1}{N}\sum_{i=1}^N\left(\theta_{k_i^0}^0-\theta_{k_i}\right)'\frac{1}{T}\sum_{t=1}^Tx_{it}x_{it}'\left(\theta_{k_i^0}^0-\theta_{k_i}\right) \\
        &= \sum_{k=1}^{K^0}\frac{N_k^0}{N}\sum_{\widetilde{k}=1}^{K^0}\left(\theta_k^0-\theta_{\widetilde{k}}\right)'\left\{\frac{1}{N_k^0}\sum_{i=1}^N1\{k_i^0=k\}1\{k_i=\widetilde{k}\}\frac{1}{T}\sum_{t=1}^Tx_{it}x_{it}'\right\}\left(\theta_k^0-\theta_{\widetilde{k}}\right) \\
        &\geq \sum_{k=1}^{K^0}\frac{N_k^0}{N}\sum_{\widetilde{k}=1}^{K^0}\rho_{NT}(\boldsymbol{\gamma}_N,k,\widetilde{k})\left\|\theta_k^0-\theta_{\widetilde{k}}\right\|^2 \\
        &\geq \widehat{\rho}_{NT}\sum_{k=1}^{K^0}\frac{N_k^0}{N}\min_{\widetilde{k}\in[K^0]}\left\|\theta_k^0-\theta_{\widetilde{k}}\right\|^2.
        \label{bound:Qtilde_diff_lower}
    \end{align}
    By Lemma \ref{lem:Qhat_Qtilde} and the definitions of $(\widehat{\boldsymbol{\theta}}^{(K^0)}, \widehat{\boldsymbol{\gamma}}_N^{(K^0)})$ and $\widetilde{Q}(\boldsymbol{\theta},\boldsymbol{\gamma}_N)$, we have
    \begin{align}
        0\leq \widetilde{Q}\left(\widehat{\boldsymbol{\theta}}^{(K^0)}, \widehat{\boldsymbol{\gamma}}_N^{(K^0)}\right) 
        &= \widehat{Q}\left(\widehat{\boldsymbol{\theta}}^{(K^0)}, \widehat{\boldsymbol{\gamma}}_N^{(K^0)}\right) + O_p(1/\sqrt{NT}) \\
        &\leq \widehat{Q}\left(\boldsymbol{\theta}^0,\boldsymbol{\gamma}_N^0\right) + O_p(1/\sqrt{NT}) \\
        &=\widetilde{Q}\left(\boldsymbol{\theta}^0,\boldsymbol{\gamma}_N^0\right) + O_p(1/\sqrt{NT}) \\
        &\leq \widetilde{Q}\left(\widehat{\boldsymbol{\theta}}^{(K^0)}, \widehat{\boldsymbol{\gamma}}_N^{(K^0)}\right) + O_p(1/\sqrt{NT}),
    \end{align}
    which implies
    \begin{align}
        \widetilde{Q}\left(\widehat{\boldsymbol{\theta}}^{(K^0)}, \widehat{\boldsymbol{\gamma}}_N^{(K^0)}\right) - \widetilde{Q}\left(\boldsymbol{\theta}^0,\boldsymbol{\gamma}_N^0\right) = O_p(1/\sqrt{NT}).
        \label{order:Qtilde_true_ols_diff}
    \end{align}

    Combining \eqref{bound:Qtilde_diff_lower}, \eqref{order:Qtilde_true_ols_diff}, and the fact that $\widehat{\rho}_{NT}\stackrel{p}{\to}\rho>0$ by Assumption \ref{asm:basic}(e), we have w.p.a.1
    \begin{align}
        O_p(1/\sqrt{NT}) 
        &= \widetilde{Q}\left(\widehat{\boldsymbol{\theta}}^{(K^0)},\widehat{\boldsymbol{\gamma}}_N^{(K^0)}\right) - \widetilde{Q}\left(\boldsymbol{\theta}^0,\boldsymbol{\gamma}^0\right) \\
        &\geq \frac{\rho}{2} \sum_{k=1}^{K^0}\frac{N_k^0}{N}\min_{\widetilde{k}\in[K^0]}\left\|\theta_k^0-\widehat{\theta}_{\widetilde{k}}^{(K^0)}\right\|^2 \\
        &\geq \frac{\rho}{2}\frac{N_{K^0}^0}{N}\max_{k\in[K^0]}\min_{\widetilde{k}\in[K^0]}\left\|\theta_k^0-\widehat{\theta}_{\widetilde{k}}^{(K^0)}\right\|^2,
    \end{align}
    and hence
    \begin{align}
        \max_{k\in[K^0]}\min_{\widetilde{k}\in[K^0]}\left\|\theta_k^0-\widehat{\theta}_{\widetilde{k}}^{(K^0)}\right\|^2 = O_p\left(N^{1/2-\alpha_{K^0}}/\sqrt{T}\right) = o_p(1), \label{order:hausdorff} 
    \end{align}
    by Assumption \ref{asm:basic}(g).

    Define $\sigma(k)\coloneqq \argmin_{\widetilde{k}\in[K^0]}\|\theta_k^0-\widehat{\theta}_{\widetilde{k}}^{(K^0)}\|^2$. Using the triangle inequality and Assumption \ref{asm:basic}(h), we have, for all $k\neq\widetilde{k}$,
    \begin{align}
        \left\|\widehat{\theta}_{\sigma(k)}^{(K^0)} - \widehat{\theta}_{\sigma(\widetilde{k})}^{(K^0)}\right\| \geq \left\|\theta_k^0-\theta_{\widetilde{k}}^0\right\| - \left\|\widehat{\theta}_{\sigma(k)}^{(K^0)} - \theta_k^0\right\| - \left\|\widehat{\theta}_{\sigma(\widetilde{k})}^{(K^0)} - \theta_{\widetilde{k}}^0\right\| \geq c/2,        
    \end{align}
    where the last inequality holds w.p.a.1 in view of \eqref{order:hausdorff}. This implies that $\sigma(k)\neq\sigma(\widetilde{k})$ for $k\neq\widetilde{k}$ w.p.a.1; that is, $\sigma$ is a permutation of $[K^0]$ and bijective w.p.a.1. Therefore, we can take a well-defined inverse $\sigma^{-1}$, and for all $\widetilde{k}\in[K^0]$,
    \begin{align}
        \min_{k\in[K^0]}\left\|\theta_k^0-\widehat{\theta}_{\widetilde{k}}^{(K^0)}\right\| 
        \leq \left\|\theta_{\sigma^{-1}(\widetilde{k})}^0-\widehat{\theta}_{\widetilde{k}}^{(K^0)}\right\|=\min_{h\in[K^0]}\left\|\theta_{\sigma^{-1}(\widetilde{k})}^0-\widehat{\theta}_h^{(K^0)}\right\| = o_p(1)
    \end{align}
    by \eqref{order:hausdorff}. This, together with \eqref{order:hausdorff}, implies $d_H(\widehat{\boldsymbol{\theta}}^{(K^0)},\boldsymbol{\theta}^0)\stackrel{p}{\to}0.$
\end{proof}

In what follows, we take $\sigma(k)=k$ by relabeling the groups, which yields $\|\widehat{\theta}_k^{(K^0)}-\theta_k^0\|\stackrel{p}{\to}0$ for all $k\in[K^0]$. For any $\eta>0$, let $\mathcal{N}_\eta$ denote the set of parameters $\boldsymbol{\theta}\in\Theta^{K^0}$ that satisfy $\|\theta_k-\theta_k^0\|^2<\eta$ for all $k\in[K^0]$.

\begin{lem}
    Consider model \eqref{model:without_gfe}. Suppose that Assumptions \ref{asm:group_size} and \ref{asm:basic} hold. Then for $\eta>0$ small enough, and for all $\delta>0$, we have
    \begin{align}
        \sup_{\boldsymbol{\theta}\in\mathcal{N}_\eta}\frac{1}{N}\sum_{i=1}^N1\left\{\widehat{k}_i^{(K^0)}\left(\boldsymbol{\theta}\right)\neq k_i^0\right\}=o_p\left(T^{-\delta}\right)
    \end{align}
    as $N,T\to\infty$, where $\widehat{k}_i^{(K^0)}(\boldsymbol{\theta})=\argmin_{k\in[K^0]}\sum_{t=1}^T(y_{it} - x_{it}'\theta_k)^2$ for given $\boldsymbol{\theta}=(\theta_1,\ldots,\theta_{K^0})$.
    \label{lem:prob_missclassify}
\end{lem}

\noindent\begin{proof}
    Noting that for all $k\in[K^0]$,
    \begin{align}
        1\left\{\widehat{k}_i^{(K^0)}(\boldsymbol{\theta})=k\right\}\leq 1\left\{\sum_{t=1}^T(y_{it}-x_{it}'\theta_k)^2 \leq \sum_{t=1}^T(y_{it}-x_{it}'\theta_{k_i^0})^2\right\},
    \end{align}
    we have
    \begin{align}
        \frac{1}{N}\sum_{i=1}^N1\left\{\widehat{k}_i^{(K^0)}(\boldsymbol{\theta})\neq k_i^0\right\} 
        &= \sum_{k=1}^{K^0}\frac{1}{N}\sum_{i=1}^N1\left\{k_i^0\neq k\right\}1\left\{\widehat{k}_i^{(K^0)}(\boldsymbol{\theta})=k\right\} \\
        &\leq\sum_{k=1}^{K^0}\frac{1}{N}\sum_{i=1}^N1\left\{k_i^0\neq k\right\}1\left\{\sum_{t=1}^T(y_{it}-x_{it}'\theta_k)^2 \leq \sum_{t=1}^T(y_{it}-x_{it}'\theta_{k_i^0})^2\right\} \\
        &=\sum_{k=1}^{K^0}\frac{1}{N}\sum_{i=1}^NZ_{ik}(\boldsymbol{\theta}),
        \label{bound:average_misclassification}
    \end{align}
    where
    \begin{align}
        Z_{ik}(\boldsymbol{\theta})
        &\coloneqq 1\{k_i^0\neq k\}1\left\{\sum_{t=1}^T(y_{it}-x_{it}'\theta_k)^2 \leq \sum_{t=1}^T(y_{it}-x_{it}'\theta_{k_i^0})^2\right\} \\
        &=1\left\{k_i^0\neq k\right\}1\left\{\sum_{t=1}^T\left(\theta_{k_i^0}-\theta_k\right)'\left\{x_{it}x_{it}'\left(2\theta_{k_i^0}^0-\theta_k-\theta_{k_i^0}\right)+2x_{it}\varepsilon_{it}\right\}\leq 0\right\} \\
        &\leq \max_{\widetilde{k}\neq k}1\left\{\sum_{t=1}^T(\theta_{\widetilde{k}}-\theta_k)'x_{it}\left\{x_{it}'\left(\theta_{\widetilde{k}}^0- \frac{\theta_k+\theta_{\widetilde{k}}}{2}\right)+\varepsilon_{it}\right\}\leq 0\right\}.
        \label{bound:z_ik}
    \end{align}

    Define $A_T\coloneqq \sum_{t=1}^T(\theta_{\widetilde{k}}-\theta_k)'x_{it}\{x_{it}'(\theta_{\widetilde{k}}^0- (\theta_k+\theta_{\widetilde{k}})/2+\varepsilon_{it}\}$. We decompose $A_T$ as follows:
    \begin{align}
        A_T
        &=\left(\theta_{\widetilde{k}}-\theta_{\widetilde{k}}^0\right)'\sum_{t=1}^Tx_{it}x_{it}'\left(\theta_{\widetilde{k}}^0-\frac{\theta_k+\theta_{\widetilde{k}}}{2}\right) +\left(\theta_{\widetilde{k}}^0-\theta_k^0\right)'\sum_{t=1}^Tx_{it}x_{it}'\left(\theta_{\widetilde{k}}^0-\frac{\theta_k+\theta_{\widetilde{k}}}{2}\right) \\
        &\quad -\left(\theta_k-\theta_k^0\right)'\sum_{t=1}^Tx_{it}x_{it}'\left(\theta_{\widetilde{k}}^0-\frac{\theta_k+\theta_{\widetilde{k}}}{2}\right) - \left(\theta_{\widetilde{k}}^0-\theta_k^0\right)'\sum_{t=1}^Tx_{it}x_{it}'\left(\theta_{\widetilde{k}}^0-\frac{\theta_k^0+\theta_{\widetilde{k}}^0}{2}\right) \\
        &\quad +\left(\theta_{\widetilde{k}}^0-\theta_k^0\right)'\sum_{t=1}^Tx_{it}x_{it}'\left(\theta_{\widetilde{k}}^0-\frac{\theta_k^0+\theta_{\widetilde{k}}^0}{2}\right) + \left(\theta_{\widetilde{k}}-\theta_k\right)'\sum_{t=1}^Tx_{it}\varepsilon_{it} \\
        &=\left(\theta_{\widetilde{k}}-\theta_{\widetilde{k}}^0\right)'\sum_{t=1}^Tx_{it}x_{it}'\theta_{\widetilde{k}}^0 - \frac{1}{2}\left(\theta_{\widetilde{k}}-\theta_{\widetilde{k}}^0\right)'\sum_{t=1}^Tx_{it}x_{it}'\left(\theta_k+\theta_{\widetilde{k}}\right) \\
        &\quad-\frac{1}{2}\left(\theta_{\widetilde{k}}^0-\theta_k^0\right)'\sum_{t=1}^Tx_{it}x_{it}'\left(\theta_k-\theta_k^0\right) - \frac{1}{2}\left(\theta_{\widetilde{k}}^0-\theta_k^0\right)'\sum_{t=1}^Tx_{it}x_{it}'\left(\theta_{\widetilde{k}}-\theta_{\widetilde{k}}^0\right) \\
        &\quad-\left(\theta_k-\theta_k^0\right)'\sum_{t=1}^Tx_{it}x_{it}'\theta_{\widetilde{k}}^0 + \frac{1}{2}\left(\theta_k-\theta_k^0\right)'\sum_{t=1}^Tx_{it}x_{it}'\left(\theta_k+\theta_{\widetilde{k}}\right) \\
        &\quad+\left(\theta_{\widetilde{k}}^0-\theta_k^0\right)'\sum_{t=1}^Tx_{it}x_{it}'\left(\theta_{\widetilde{k}}^0 - \frac{\theta_k^0+\theta_{\widetilde{k}}^0}{2}\right) + \left(\theta_{\widetilde{k}}-\theta_{\widetilde{k}}^0\right)'\sum_{t=1}^Tx_{it}\varepsilon_{it} \\
        &\quad+\left(\theta_{\widetilde{k}}^0 - \theta_k^0\right)'\sum_{t=1}^Tx_{it}\varepsilon_{it} - \left(\theta_k-\theta_k^0\right)'\sum_{t=1}^Tx_{it}\varepsilon_{it} \\
        &=A_{T,1} + A_{T,2},
    \end{align}
    where
    \begin{align}
        A_{T,1} &\coloneqq \left(\theta_{\widetilde{k}}-\theta_{\widetilde{k}}^0\right)'\sum_{t=1}^Tx_{it}x_{it}'\left(\theta_{\widetilde{k}}^0 - \frac{\theta_{\widetilde{k}} + \theta_{\widetilde{k}}^0}{2}\right) - \left(\theta_k-\theta_k^0\right)\sum_{t=1}^Tx_{it}x_{it}'\left(\theta_{\widetilde{k}}^0-\frac{\theta_k+\theta_k^0}{2}\right) \\
        &\quad +\left(\theta_{\widetilde{k}}-\theta_{\widetilde{k}}^0\right)'\sum_{t=1}^Tx_{it}\varepsilon_{it} - \left(\theta_k-\theta_k^0\right)'\sum_{t=1}^Tx_{it}\varepsilon_{it},
        \intertext{and}
        A_{T,2} &\coloneqq \left(\theta_{\widetilde{k}}^0-\theta_k^0\right)'\sum_{t=1}^Tx_{it}\left\{x_{it}'\left(\theta_{\widetilde{k}}^0 - \frac{\theta_k^0+\theta_{\widetilde{k}}^0}{2}\right)+\varepsilon_{it}\right\}.
    \end{align}
    Substituting this into \eqref{bound:z_ik} and using the Cauchy-Schwarz inequality yield
    \begin{align}
        Z_{ik}(\boldsymbol{\theta})
        &\leq \max_{\widetilde{k}\neq k}1\{A_T\leq 0\} \\
        &\leq \max_{\widetilde{k}\neq k}1\{A_{T,2} \leq |A_{T,1}|\} \\
        &\leq \max_{\widetilde{k}\neq k}1\left\{\left(\theta_{\widetilde{k}}^0-\theta_k^0\right)'\sum_{t=1}^Tx_{it}\left\{x_{it}'\left(\theta_{\widetilde{k}}^0 - \frac{\theta_k^0+\theta_{\widetilde{k}}^0}{2}\right)+\varepsilon_{it}\right\} \right. \\
        &\left. \hspace{3cm} \leq TC\sqrt{\eta}\left(T^{-1}\sum_{t=1}^T\|x_{it}\|^2 + \left\|T^{-1}\sum_{t=1}^Tx_{it}\varepsilon_{it}\right\|\right)\right\}
    \end{align}
    for $\boldsymbol{\theta}\in\mathcal{N}_{\eta}$ and some constant $C>0$ independent of $\eta$ and $T$.

    Noting that the most right hand side of the above inequality does not depend on $\boldsymbol{\theta}$, we deduce $\sup_{\boldsymbol{\theta}\in\mathcal{N}_{\eta}}Z_{ik}(\boldsymbol{\theta})\leq\widetilde{Z}_{ik}$, where
    \begin{align}
        \widetilde{Z}_{ik} \coloneqq \max_{\widetilde{k}\neq k}1\Biggl\{\left(\theta_{\widetilde{k}}^0-\theta_k^0\right)'\sum_{t=1}^Tx_{it}\varepsilon_{it} &\leq 
        - \frac{1}{2}\left(\theta_{\widetilde{k}}^0-\theta_k^0\right)'\sum_{t=1}^Tx_{it}x_{it}'\left(\theta_{\widetilde{k}}^0 - \theta_k^0\right)  \\
        &+ TC\sqrt{\eta}\left(T^{-1}\sum_{t=1}^T\|x_{it}\|^2 + \left\|T^{-1}\sum_{t=1}^Tx_{it}\varepsilon_{it}\right\|\right)\Biggr\}. \\
        \label{def:z_ik_tilde}
    \end{align}
    It follows from \eqref{bound:average_misclassification} and \eqref{def:z_ik_tilde} that
    \begin{align}
        \sup_{\boldsymbol{\theta}\in\mathcal{N}_{\eta}}\frac{1}{N}\sum_{i=1}^N1\left\{\widehat{k}_i^{(K^0)}(\boldsymbol{\theta})\neq k_i^0\right\} \leq \sum_{k=1}^{K^0}\frac{1}{N}\sum_{i=1}^N\widetilde{Z}_{ik}.
        \label{bound:uniform_average_misclass}
    \end{align}

    By Assumptions \ref{asm:basic}(f) and (h), we have for all $k\in[K^0]$,
    \begin{align}
        P\left(\widetilde{Z}_{ik}=1\right) 
        &\leq \sum_{k\neq \widetilde{k}}P\Biggl(\left(\theta_{\widetilde{k}}^0 - \theta_k^0\right)'\sum_{t=1}^Tx_{it}\varepsilon_{it} \leq - \frac{1}{2}\left(\theta_{\widetilde{k}}^0-\theta_k^0\right)'\sum_{t=1}^Tx_{it}x_{it}'\left(\theta_{\widetilde{k}}^0 - \theta_k^0\right) \\
        &\hspace{5.5cm} + TC\sqrt{\eta}\left(T^{-1}\sum_{t=1}^T\|x_{it}\|^2 + \left\|T^{-1}\sum_{t=1}^Tx_{it}\varepsilon_{it}\right\|\right)\Biggr) \\
        &\leq\sum_{k\neq\widetilde{k}}\Biggl[P\Biggl(\left(\theta_{\widetilde{k}}^0 - \theta_k^0\right)'\sum_{t=1}^Tx_{it}\varepsilon_{it} \leq - \frac{Tc^3}{4} + 2TC'\sqrt{\eta}\Biggr) \\
        &\hspace{2.5cm} + P\left(T^{-1}\sum_{t=1}^T\|x_{it}\|^2 \geq M\right) + P\left(\left\|T^{-1}\sum_{t=1}^Tx_{it}\varepsilon_{it}\right\| \geq M\right)\Biggr],
    \end{align}
    where $C'=MC$. Taking $\eta$ such that $0<\eta\leq (c^3/16C')^2$, we get
    \begin{align}
        P\Biggl(\left(\theta_{\widetilde{k}}^0 - \theta_k^0\right)'\sum_{t=1}^Tx_{it}\varepsilon_{it} \leq - \frac{Tc^3}{4} + 2TC'\sqrt{\eta}\Biggr) 
        &\leq P\left(\left(\theta_{\widetilde{k}}^0 - \theta_k^0\right)'T^{-1}\sum_{t=1}^Tx_{it}\varepsilon_{it} \leq - \frac{c^3}{8}\right) \\
        &\leq P\left(\left\|T^{-1}\sum_{t=1}^Tx_{it}\varepsilon_{it}\right\| \geq  \frac{c^3}{16\sup_{\theta\in\Theta}\|\theta\|}\right) \\
        &= o\left(T^{-\delta}\right)
    \end{align}
    for all $\delta>0$ uniformly in $i$ given Assumption \ref{asm:basic}(j). With this and Assumptions \ref{asm:basic}(i) and (j), we arrive at
    \begin{align}
        P\left(\widetilde{Z}_{ik}=1\right)\leq K^0\left(K^0-1\right)\times o\left(T^{-\delta}\right) = o\left(T^{-\delta}\right) \label{prob:tilde_z_ik}
    \end{align}
    for all $\delta>0$ uniformly in $i$. Using \eqref{bound:uniform_average_misclass}, \eqref{prob:tilde_z_ik}, and the Markov inequality, we have for any $\epsilon>0$ and any $\delta>0$,
    \begin{align}
        P\left(\sup_{\boldsymbol{\theta}\in\mathcal{N}_{\eta}}\frac{1}{N}\sum_{i=1}^N1\left\{\widehat{k}_i^{(K^0)}(\boldsymbol{\theta})\neq k_i^0\right\} > \epsilon T^{-\delta}\right)  &\leq P\left(\sum_{k=1}^{K^0}\frac{1}{N}\sum_{i=1}^N\widetilde{Z}_{ik} > \epsilon T^{-\delta}\right) \\
        &\leq \frac{N^{-1}\sum_{k=1}^{K^0}\sum_{i=1}^NE\left[\widetilde{Z}_{ik}\right]}{\epsilon T^{-\delta}} \\
        &= o(1)
    \end{align}
    This ends the proof of Lemma \ref{lem:prob_missclassify}.
\end{proof}

For $k\in[K^0]$, let $\widetilde{\theta}_k^{(K^0)}\coloneqq (\sum_{i\in G_k^0}\sum_{t=1}^Tx_{it}x_{it}')^{-1}\sum_{i\in G_k^0}\sum_{t=1}^Tx_{it}y_{it}$ denote the oracle estimator of $\theta_k^0$ calculated using cross-sections belonging to $G_k^0$.

\begin{prop}
    Consider model \eqref{model:without_gfe}. Suppose that Assumptions \ref{asm:group_size} and \ref{asm:basic} hold. Then as $N,T\to\infty$, we have, for all $\delta>0$,
    
    \noindent (i) $\widehat{\theta}_k^{(K^0)} = \widetilde{\theta}_k^{(K^0)} + o_p\left(N^{(1-\alpha_k)/2}T^{-\delta}\right)$ for all $k\in[K^0]$,
    
    \noindent (ii) $P\left(\bigcup_{i=1}^N\left\{\widehat{k}_i^{(K^0)}\neq k_i^0\right\}\right)=o(1)+o\left(NT^{-\delta}\right)$.
    \label{prop:asym_equivalence}
\end{prop}

\noindent\begin{proof}
    (i) Define 
    \begin{align}
        \widehat{Q}(\boldsymbol{\theta})&\coloneqq (NT)^{-1}\sum_{i=1}^N\sum_{t=1}^T\left(y_{it}-x_{it}'\theta_{\widehat{k}_i(\boldsymbol{\theta})}\right)^2
        \intertext{and}
        \widetilde{Q}(\boldsymbol{\theta})&\coloneqq (NT)^{-1}\sum_{i=1}^N\sum_{t=1}^T\left(y_{it}-x_{it}'\theta_{k_i^0}\right)^2.
    \end{align}
    Note that $\widehat{Q}(\boldsymbol{\theta})$ is minimized at $\boldsymbol{\theta}=\widehat{\boldsymbol{\theta}}^{(K^0)}$, and $\widetilde{Q}(\boldsymbol{\theta})$ is minimized at $\boldsymbol{\theta}=\widetilde{\boldsymbol{\theta}}^{(K^0)}$. Let $\eta>0$ be sufficiently small that Lemma \ref{lem:prob_missclassify} holds. Because $\{\widehat{Q}(\boldsymbol{\theta})-\widetilde{Q}(\boldsymbol{\theta})\}1\{\widehat{k}_i^{(K^0)}(\boldsymbol{\theta})=k_i^0\}=0$, $E[\varepsilon_{it}^4]\leq M$, $E[\|x_{it}\|^4]\leq M$, and $\sup_{\boldsymbol{\theta}\in\Theta^{K^0}}\|\boldsymbol{\theta}\|\leq M$ under Assumption \ref{asm:basic}, applying the Cauchy-Schwarz inequality and Lemma \ref{lem:prob_missclassify} gives
    \begin{align}
        \sup_{\boldsymbol{\theta}\in\mathcal{N}_{\eta}}\left|\widehat{Q}(\boldsymbol{\theta})-\widetilde{Q}(\boldsymbol{\theta})\right| 
        &\leq \sup_{\boldsymbol{\theta}\in\mathcal{N}_{\eta}}\frac{1}{NT}\sum_{i=1}^N\sum_{t=1}^T\left|\left(y_{it}-x_{it}'\theta_{\widehat{k}_i^{(K^0)}(\boldsymbol{\theta})}\right)^2 - \left(y_{it}-x_{it}'\theta_{k_i^0}\right)^2\right|1\left\{\widehat{k}_i^{(K^0)}(\boldsymbol{\theta})\neq k_i^0\right\} \\
        &\leq O_p(1)\times \sup_{\boldsymbol{\theta}\in\mathcal{N}_{\eta}}\left(\frac{1}{N}\sum_{i=1}^N 1\left\{\widehat{k}_i^{(K^0)}(\boldsymbol{\theta})\neq k_i^0\right\}\right)^{1/2} \\
        &=o_p\left(T^{-\delta}\right)
        \label{bound:qhat_qtilde_diff}
    \end{align}
    for all $\delta>0$.

    Because $\|\widehat{\theta}_k^{(K^0)} -\theta_k^0\|\stackrel{p}{\to}0$ for all $k\in[K^0]$, we have
    \begin{align}
        P\left(\widehat{\boldsymbol{\theta}}^{(K^0)}\notin\mathcal{N}_{\eta}\right)\to0 \ \mathrm{as} \ N,T\to\infty.
        \label{prob:theta_hat_consistency}
    \end{align}
    Furthermore, it is straightforward to prove
    \begin{align}
        P\left(\widetilde{\boldsymbol{\theta}}^{(K^0)}\notin\mathcal{N}_{\eta}\right)\to0 \ \mathrm{as} \ N,T\to\infty.
        \label{prob:theta_tilde_consistency}
    \end{align}
    \eqref{bound:qhat_qtilde_diff} and \eqref{prob:theta_hat_consistency} imply that for every $\epsilon>0$,
    \begin{align}
        &P\left(\left|\widehat{Q}\left(\widehat{\boldsymbol{\theta}}^{(K^0)}\right)-\widetilde{Q}\left(\widehat{\boldsymbol{\theta}}^{(K^0)}\right)\right| > \epsilon T^{-\delta}\right) \\
        &\leq P\left(\sup_{\boldsymbol{\theta}\in \mathcal{N}_{\eta}}\left|\widehat{Q}(\boldsymbol{\theta})-\widetilde{Q}(\boldsymbol{\theta})\right| + \left|\widehat{Q}\left(\widehat{\boldsymbol{\theta}}^{(K^0)}\right)-\widetilde{Q}\left(\widehat{\boldsymbol{\theta}}^{(K^0)}\right)\right|1\left\{\widehat{\boldsymbol{\theta}}^{(K^0)}\notin\mathcal{N}_{\eta}\right\}>\epsilon T^{-\delta}\right) \\
        &\leq P\left(\sup_{\boldsymbol{\theta}\in \mathcal{N}_{\eta}}\left|\widehat{Q}(\boldsymbol{\theta})-\widetilde{Q}(\boldsymbol{\theta})\right| > \frac{\epsilon}{2}T^{-\delta}\right) + P\left(\widehat{\boldsymbol{\theta}}^{(K^0)}\notin\mathcal{N}_{\eta}\right) \to 0,
    \end{align}
    and consequently
    \begin{align}
        \widehat{Q}\left(\widehat{\boldsymbol{\theta}}^{(K^0)}\right)-\widetilde{Q}\left(\widehat{\boldsymbol{\theta}}^{(K^0)}\right) = o_p\left(T^{-\delta}\right).
        \label{order:qhat_qtilde_diff_theta_hat}
    \end{align}
    Similarly, combining \eqref{bound:qhat_qtilde_diff} and \eqref{prob:theta_tilde_consistency} yields
    \begin{align}
        \widehat{Q}\left(\widetilde{\boldsymbol{\theta}}^{(K^0)}\right)-\widetilde{Q}\left(\widetilde{\boldsymbol{\theta}}^{(K^0)}\right) = o_p\left(T^{-\delta}\right).
        \label{order:qhat_qtilde_diff_theta_tilde}
    \end{align}

    Using \eqref{order:qhat_qtilde_diff_theta_hat}, \eqref{order:qhat_qtilde_diff_theta_tilde}, and the fact that $\widetilde{\boldsymbol{\theta}}^{(K^0)}$ and $\widehat{\boldsymbol{\theta}}^{(K^0)}$ minimize $\widetilde{Q}(\boldsymbol{\theta})$ and $\widehat{Q}(\boldsymbol{\theta})$, respectively, we deduce
    \begin{align}
        0\leq \widetilde{Q}\left(\widehat{\boldsymbol{\theta}}^{(K^0)}\right) - \widetilde{Q}\left(\widetilde{\boldsymbol{\theta}}^{(K^0)}\right) 
        = \widehat{Q}\left(\widehat{\boldsymbol{\theta}}^{(K^0)}\right) - \widehat{Q}\left(\widetilde{\boldsymbol{\theta}}^{(K^0)}\right) + o_p\left(T^{-\delta}\right)\leq o_p\left(T^{-\delta}\right).
    \end{align}
    Hence $\widetilde{Q}\left(\widehat{\boldsymbol{\theta}}^{(K^0)}\right) - \widetilde{Q}\left(\widetilde{\boldsymbol{\theta}}^{(K^0)}\right)=o_p(T^{-\delta})$.

    Writing $\widetilde{\varepsilon}_{it}\coloneqq y_{it} - x_{it}'\widetilde{\theta}_{k_i^0}^{(K^0)}$, and noting that $\widetilde{\boldsymbol{\theta}}^{(K^0)}$ is a least-squares estimator, we obtain
    \begin{align}
        \widetilde{Q}\left(\widehat{\boldsymbol{\theta}}^{(K^0)}\right) - \widetilde{Q}\left(\widetilde{\boldsymbol{\theta}}^{(K^0)}\right) 
        &= \frac{1}{NT}\sum_{i=1}^N\sum_{t=1}^T\left[\left\{x_{it}'\left(\widetilde{\theta}_{k_i^0}^{(K^0)} - \widehat{\theta}_{k_i^0}^{(K^0)}\right) + \widetilde{\varepsilon}_{it}\right\}^2 - \widetilde{\varepsilon}_{it}^2\right] \\
        &=\sum_{k=1}^{K^0}\frac{N_k^0}{N}\left(\widetilde{\theta}_k^{(K^0)} - \widehat{\theta}_k^{(K^0)}\right)'\frac{1}{N_k^0}\sum_{i=1}^N1\{k_i^0=k\}\frac{1}{T}\sum_{t=1}^Tx_{it}x_{it}'\left(\widetilde{\theta}_k^{(K^0)} - \widehat{\theta}_k^{(K^0)}\right) \\
        &\geq \widehat{\rho}_{NT}\sum_{k=1}^{K^0}\frac{N_k^0}{N}\left\|\widetilde{\theta}_k^{(K^0)} - \widehat{\theta}_k^{(K^0)}\right\|^2.
    \end{align}
    This implies that $\|\widetilde{\theta}_k^{(K^0)} - \widehat{\theta}_k^{(K^0)}\|^2=o_p(N^{1-\alpha_k}T^{-\delta})$ for all $\delta>0$ and $k\in[K^0]$.

    \noindent (ii) By the union bound, we obtain
    \begin{align}
        P\left(\bigcup_{i=1}^N\left\{\widehat{k}_i^{(K^0)}\left(\widehat{\boldsymbol{\theta}}^{(K^0)}\right)\neq k_i^0\right\}\right) \leq P\left(\widehat{\boldsymbol{\theta}}^{(K^0)}\notin\mathcal{N}_{\eta}\right) + N\sup_{i\in[N]}P\left(\widehat{\boldsymbol{\theta}}^{(K^0)}\in\mathcal{N}_{\eta}, \ \widehat{k}_i^{(K^0)}\left(\widehat{\boldsymbol{\theta}}^{(K^0)}\right)\neq k_i^0\right).
    \end{align}
    Letting $\eta>0$ be small enough, we have $P(\widehat{\boldsymbol{\theta}}^{(K^0)}\notin\mathcal{N}_\eta)=o(1)$. Moreover, by \eqref{bound:average_misclassification} and the fact that $\sup_{\boldsymbol{\theta}\in\mathcal{N}_\eta}Z_{ik}(\boldsymbol{\theta})\leq \widetilde{Z}_{ik}$, we have $\sup_{\boldsymbol{\theta}\in\mathcal{N}_\eta}1\{\widehat{k}_i^{(K^0)}(\boldsymbol{\theta})\neq k_i^0\}\leq\sum_{k=1}^{K^0}\widetilde{Z}_{ik}$. It follows from \eqref{prob:tilde_z_ik} that
    \begin{align}
    N\sup_{i\in[N]}P\left(\widehat{\boldsymbol{\theta}}^{(K^0)}\in\mathcal{N}_{\eta}, \ \widehat{k}_i^{(K^0)}\left(\widehat{\boldsymbol{\theta}}^{(K^0)}\right)\neq k_i^0\right)
    &=N\sup_{i\in[N]}E\left[1\left\{\widehat{\boldsymbol{\theta}}^{(K^0)}\in\mathcal{N}_{\eta}\right\}1\left\{\widehat{k}_i^{(K^0)}\left(\widehat{\boldsymbol{\theta}}^{(K^0)}\right)\neq k_i^0\right\}\right] \\
    &\leq N\sup_{i\in[N]}\sum_{k=1}^{K^0}E\left[\widetilde{Z}_{ik}\right] \\
    &=N\sup_{i\in[N]}\sum_{k=1}^{K^0}P\left(\widetilde{Z}_{ik}=1\right)=o\left(NT^{-\delta}\right)
    \end{align}
    for all $\delta>0$.
\end{proof}

\noindent\begin{proof}[Proof of Theorem \ref{thm:consistency_normality}]

\noindent(i) Part (i) immediately follows from Proposition \ref{prop:asym_equivalence}(ii) and Assumption \ref{asm:basic}(g).

\noindent(ii) From Proposition \ref{prop:asym_equivalence}(i), we have, for each $k\in[K^0]$, $\sqrt{N_{k}^0T}(\widehat{\theta}_k^{(K^0)} - \theta_k^0) = \sqrt{N_{k}^0T}(\widetilde{\theta}_k^{(K^0)} - \theta_k^0) + o_p(N^{1/2}T^{1/2-\delta})$ for any $\delta>0$. The $o_p(N^{1/2}T^{1/2-\delta})$ term is $o_p(1)$ by Assumption \ref{asm:basic}(g). Since we have
\begin{align}
    \sqrt{N_{k}^0T}\left(\widetilde{\theta}_k^{(K^0)} - \theta_k^0\right) = \left(\frac{1}{N_k^0T}\sum_{i\in G_k^0}\sum_{t=1}^Tx_{it}x_{it}'\right)^{-1}\frac{1}{\sqrt{N_k^0T}}\sum_{i\in G_k^0}\sum_{t=1}^Tx_{it}\varepsilon_{it},
\end{align}
the weak convergence result readily follows from Assumption \ref{asm:dist}.

\noindent(iii) For each $k\in[K^0]$, let $\widehat{G}_k^{(K^0)} \coloneqq \{i\in [N]:\widehat{k}_i^{(K^0)}=k\}$. Given the result in part (i), we may assume without loss of generality that $\widehat{G}_k^{(K^0)}$ is the estimate of the true $k$-th group, $G_k^0$. Noting that $1\{i\in\widehat{G}_k^{(K^0)}\} = 1\{i\in G_k^0\} + 1\{i\in \widehat{G}_k^{(K^0)} \setminus G_k^0\} - 1\{i\in G_k^0 \setminus \widehat{G}_k^{(K^0)}\}$, we get
\begin{align}
    \widehat{\sigma}^2\left(K^0, \widehat{\boldsymbol{\gamma}}_N^{(K^0)}\right)
    &=\frac{1}{NT}\sum_{k=1}^{K^0}\sum_{\quad i\in \widehat{G}_k^{(K^0)}}\sum_{t=1}^T \left(y_{it}-x_{it}'\widehat{\theta}_k^{(K^0)}\right)^2 \\
    &=\frac{1}{NT}\sum_{k=1}^{K^0}\sum_{i\in G_k^0}\sum_{t=1}^T \left(y_{it}-x_{it}'\widehat{\theta}_k^{(K^0)}\right)^2 + B_{NT,1} - B_{NT,2},
\end{align}
where $B_{NT,1} \coloneqq (NT)^{-1}\sum_{k=1}^{K^0}\sum_{i\in \widehat{G}_k^{(K^0)} \setminus G_k^0}\sum_{t=1}^T(y_{it}-x_{it}'\widehat{\theta}_k^{(K^0)})^2$, and $B_{NT,2} \coloneqq (NT)^{-1}\sum_{k=1}^{K^0}$

\noindent$\sum_{i\in G_k^0 \setminus \widehat{G}_k^{(K^0)}}\sum_{t=1}^T(y_{it}-x_{it}'\widehat{\theta}_k^{(K^0)})^2$. Applying part (i) yields, for any $\epsilon>0$,
\begin{align}
    &P\left(NT\times B_{NT,1}\geq\epsilon\right) \leq P\left(\bigcup_{k=1}^{K^0}\bigcup_{\quad i\in\widehat{G}_k^{(K^0)}}\left\{i\notin G_k^0\right\}\right) \to 0,
    \intertext{and}
    &P\left(NT\times B_{NT,2}\geq\epsilon\right) \leq P\left(\bigcup_{k=1}^{K^0}\bigcup_{i\in G_k^0}\left\{i\notin \widehat{G}_k^{(K^0)}\right\}\right) \to 0,
\end{align}
which implies
\begin{align}
    \widehat{\sigma}^2\left(K^0,\widehat{\boldsymbol{\gamma}}_N\right) = \frac{1}{NT}\sum_{k=1}^{K^0}\sum_{i\in G_k^0}\sum_{t=1}^T \left(y_{it}-x_{it}'\widehat{\theta}_k^{(K^0)}\right)^2 + o_p\left((NT)^{-1}\right),
\end{align}
Noting that $\widehat{\theta}_k^{(K^0)} = \widetilde{\theta}_k^{(K^0)} + a_{NT}$ with $a_{NT}=o_p((NT)^{-1/2})$ from Proposition \ref{prop:asym_equivalence},  and that $\widetilde{\theta}_k^{(K^0)} - \theta_k^0=O_p((N^{\alpha_k}T)^{-1/2})$ from the proof of part (ii), we observe
\begin{align}
    \widehat{\sigma}^2\left(K^0,\widehat{\boldsymbol{\gamma}}_N^{(K^0)}\right) 
    &= \frac{1}{NT}\sum_{k=1}^{K^0}\sum_{i\in G_k^0}\sum_{t=1}^T\varepsilon_{it}^2 - \frac{2}{NT}\sum_{k=1}^{K^0}\left(\widetilde{\theta}_k^{(K^0)} - \theta_k^0\right)'\sum_{i\in G_k^0}\sum_{t=1}^Tx_{it}\varepsilon_{it} - \frac{2a_{NT}'}{NT}\sum_{k=1}^{K^0}\sum_{i\in G_k^0}\sum_{t=1}^Tx_{it}\varepsilon_{it} \\
    &\quad + \frac{1}{NT}\sum_{k=1}^{K^0}\left(\widetilde{\theta}_k^{(K^0)}-\theta_k^0\right)'\sum_{i\in G_k^0}\sum_{t=1}^Tx_{it}x_{it}'\left(\widetilde{\theta}_k^{(K^0)} - \theta_k^0\right) \\
    &\quad + \frac{2}{NT}\sum_{k=1}^{K^0}\left(\widetilde{\theta}_k^{(K^0)} - \theta_k^0\right)'\sum_{i\in G_k^0}\sum_{t=1}^Tx_{it}x_{it}'a_{NT} + \frac{a_{NT}'}{NT}\sum_{k=1}^{K^0}\sum_{i\in G_k^0}\sum_{t=1}^Tx_{it}x_{it}'a_{NT} + o_p\left((NT)^{-1}\right) \\
    &=\frac{1}{NT}\sum_{k=1}^{K^0}\sum_{i\in G_k^0}\sum_{t=1}^T\varepsilon_{it}^2 + O_p\left((NT)^{-1}\right) \\
    &\stackrel{p}{\to} \sigma^2
\end{align}
in view of Assumptions \ref{asm:basic}(k) and \ref{asm:dist}(c). This completes the proof.
\end{proof}

Fix $K\in\{m,m+1,\ldots,K^0-1\}$. Let $\Gamma_N(K,m)$ be the set of groupings $\boldsymbol{\gamma}_N(K,m)\in[K]^N$ satisfying the following condition:

\begin{definition}
$\Gamma_N(K,m)$ is the set of all $\boldsymbol{\gamma}_N(K,m)=(k_1,\ldots,k_N)\in[K]^N$ such that:
\begin{adjustwidth}{1cm}{}
 For each $k=1,\ldots,m$, $G_k^0$ is partitioned as $G_k^0 = G_k^0(0,\boldsymbol{\gamma}_N(K,m)) \cup G_k^0(1,\boldsymbol{\gamma}_N(K,m))$ where $G_k^0(0,\boldsymbol{\gamma}_N(K,m)) \coloneqq \{i\in G_k^0:k_i=k\}$ and $G_k^0(1,\boldsymbol{\gamma}_N(K,m)) \coloneqq \{i\in G_k^0:k_i\neq k\}$ with $|G_k^0(1,\boldsymbol{\gamma}_N(K,m))| \leq N^{\alpha_{m+1}}$, while $k_i\in [K]$ for $i$ such that $i\in G_k^0$ for some $k\in\{m+1,\ldots,K^0\}$. Additionally, $\sum_{i=1}^N1\{k_i=k\}\geq 1$ for all $k\in\{m+1,\ldots,K\}$ if $K\geq m+1$.
\end{adjustwidth}
\label{def:gamma_N_K_m}
\end{definition}

Given any grouping $\boldsymbol{\gamma}_N(K,m)=(k_1,\ldots,k_N)\in\Gamma_N(K,m)$, the number of cross-sections $i$ such that $k_i^0=k$ for some $k\in\{1,\ldots,m\}$ but $k_i\neq k$ is $O(N^{\alpha_{m+1}})=o(N)$. Furthermore, the hypothetical group membership, $k_i$, assigned for $i$ such that $k_i^0=k$ for some $k\in\{m+1,\ldots,K^0\}$ is unrestricted. Note that $\Gamma_N(K,m)$ is nonempty for any $K\in\{m,m+1,\ldots,K^0-1\}$.

For a given $\boldsymbol{\gamma}_N=(k_1,\ldots,k_N)\in[K]^N$, let $G_k(\boldsymbol{\gamma}_N)\coloneqq\{i\in[N]:k_i=k\}$ be the $k$-th group formed under $\boldsymbol{\gamma}_N$. Let $\widehat{\boldsymbol{\theta}}^{(K)}(\boldsymbol{\gamma}_N)$ denote the OLS estimator of $\boldsymbol{\theta}\in\Theta^K$ conditional on $\boldsymbol{\gamma}_N$:
\begin{align}
    \widehat{\boldsymbol{\theta}}^{(K)}(\boldsymbol{\gamma}_N) \coloneqq \argmin_{\boldsymbol{\theta}\in\Theta^K}\sum_{i=1}^N\sum_{t=1}^T(y_{it}-x_{it}'\theta_{k_i})^2.
\end{align}
In particular, $\widehat{\theta}_k^{(K)}(\boldsymbol{\gamma}_N)=(\sum_{i\in G_k(\boldsymbol{\gamma}_N)}\sum_{t=1}^Tx_{it}x_{it}')^{-1}\sum_{i\in G_k(\boldsymbol{\gamma}_N)}\sum_{t=1}^Tx_{it}y_{it}$ for each $k\in [K]$.

\begin{lem}
    Consider model \eqref{model:without_gfe}. Suppose Assumptions \ref{asm:group_size}, \ref{asm:basic}(a)-(d) and (f) hold. Take $\underline{K}\in\{m,\ldots,K^0-1\}$. For each $N\in\mathbb{N}$, take any $\boldsymbol{\gamma}_N(\underline{K},m)=(k_1,\ldots,k_N)\in\Gamma_N(\underline{K},m)$. Then the following results hold.

    \noindent(i) $\widehat{\theta}_k^{(\underline{K})}(\boldsymbol{\gamma}_N(\underline{K},m))=\theta_k^0+b_{k,NT}$ for each $k\in[m]$, where $b_{k,NT}=O_p(\max\{(NT)^{-1/2},N^{\alpha_{m+1}-1}\})$. 

    \noindent(ii) If $\underline{K}\geq m+1$, then $\widehat{\theta}_k^{(\underline{K})}(\boldsymbol{\gamma}_N(\underline{K},m))=O_p(1)$ for all $k\in\{m+1,\ldots,\underline{K}\}$.
    \label{lem:underfit}
\end{lem}

\noindent\begin{proof}
    (i) Take any $k\in[m]$. Note that $\widehat{\theta}_k^{(\underline{K})}(\boldsymbol{\gamma}_N(\underline{K},m))$ is well-defined for sufficiently large $T$ in view of Assumption \ref{asm:basic}(f) and the definition of $\Gamma_N(\underline{K},m)$. Using $G_k^0=G_k^0(0,\boldsymbol{\gamma}_N(\underline{K},m)) \cup G_k^0(1,\boldsymbol{\gamma}_N(\underline{K},m))$ where $G_k^0(0,\boldsymbol{\gamma}_N(\underline{K},m))=\{i\in G_k^0:k_i=k\}$, and $G_k(\boldsymbol{\gamma}_N(\underline{K},m))=\{i\in [N]:k_i=k\}$, $\widehat{\theta}_k^{(\underline{K})}(\boldsymbol{\gamma}_N(\underline{K},m))$ is expressed as
    \begin{align}
        \widehat{\theta}_k^{(\underline{K})}(\boldsymbol{\gamma}_N(\underline{K},m)) 
        &=\left\{\left(\sum_{i\in G_k^0(0,\boldsymbol{\gamma}_N(\underline{K},m))} + \sum_{j=1,j\neq k}^{K^0}\sum_{i\in G_j^0\cap G_k(\boldsymbol{\gamma}_N(\underline{K},m))}\right)\sum_{t=1}^Tx_{it}x_{it}'\right\}^{-1} \\
        &\qquad\times \left(\sum_{i\in G_k^0(0,\boldsymbol{\gamma}_N(\underline{K},m))} + \sum_{j=1,j\neq k}^{K^0}\sum_{i\in G_j^0\cap G_k(\boldsymbol{\gamma}_N(\underline{K},m))}\right)\sum_{t=1}^T\left(x_{it}x_{it}'\theta_{g_i^0}^0 + x_{it}\varepsilon_{it}\right) \\
        &=\left(\frac{1}{N_k^0T}\sum_{i\in G_k^0(0,\boldsymbol{\gamma}_N(\underline{K},m))}\sum_{t=1}^Tx_{it}x_{it}'+O_p\left(N^{\alpha_{m+1}-1}\right)\right)^{-1} \\
        &\qquad\times \left[\frac{1}{N_k^0T}\sum_{i\in G_k^0(0,\boldsymbol{\gamma}_N(\underline{K},m))}\sum_{t=1}^Tx_{it}x_{it}'\theta_k^0+O_p\left(\max\left\{(NT)^{-1/2},N^{\alpha_{m+1}-1}\right\}\right)\right] \\
        &=\theta_k^0 + O_p\left(\max\left\{(NT)^{-1/2},N^{\alpha_{m+1}-1}\right\}\right)
    \end{align}
    by Assumptions \ref{asm:group_size}, \ref{asm:basic}(a), (b), (d) and (f), and the definition of $\Gamma_N(\underline{K},m)$. This verifies part (i).

    \noindent(ii) Next, consider $k\in\{m+1,\ldots,\underline{K}\}$ if $\underline{K}\geq m+1$. Under any $\boldsymbol{\gamma}_N(\underline{K},m)\in\Gamma_N(\underline{K},m)$, the cardinality of $G_k(\boldsymbol{\gamma}_N(\underline{K},m))$ satisfies $1\leq |G_k(\boldsymbol{\gamma}_N(\underline{K},m))|\leq CN^{\alpha_{m+1}}$ for some fixed $C<\infty$, for $N$ sufficiently large. Therefore, $\widehat{\theta}_k^{(\underline{K})}(\boldsymbol{\gamma}_N(\underline{K},m))$ is well-defined under Assumption \ref{asm:basic}(f), and we have
    \begin{align}
        \widehat{\theta}_k^{(\underline{K})}(\boldsymbol{\gamma}_N(\underline{K},m)) 
        &= \left(\frac{1}{|G_k(\boldsymbol{\gamma}_N(\underline{K},m))|T}\sum_{i\in G_k(\boldsymbol{\gamma}_N(\underline{K},m))}\sum_{t=1}^Tx_{it}x_{it}'\right)^{-1} \\
        &\qquad \times \frac{1}{|G_k(\boldsymbol{\gamma}_N(\underline{K},m))|T}\sum_{i\in G_k(\boldsymbol{\gamma}_N(\underline{K},m))}\sum_{t=1}^Tx_{it}\left(x_{it}'\theta_{k_i^0}^0 + \varepsilon_{it}\right) \\
        &=O_p(1)
    \end{align}
    by Assumptions \ref{asm:basic}(a), (b), and (c).
\end{proof}

For a given $\boldsymbol{\gamma}_N=(k_1,\ldots,k_N)\in[K]^N$, we let $\widehat{\sigma}^2(K,\boldsymbol{\gamma}_N) \coloneqq (NT)^{-1}\sum_{i=1}^N\sum_{t=1}^T(y_{it}-x_{it}'\widehat{\theta}_{k_i}^{(K)}(\boldsymbol{\gamma}_N))^2$ denote the averaged SSR conditional on $\boldsymbol{\gamma}_N$.

For each $K\in\{K^0+1,\ldots,K_{\max}\}$, define $\overline{\Gamma}_N(K)$ as the set of all $\boldsymbol{\gamma}_N(K)=(k_1,\ldots,k_N)\in[K]^N$ satisfying the following condition:
\begin{align}
    \forall i,j\in[N], k_i^0\neq k_j^0 \Rightarrow k_i\neq k_j. \label{def:gamma_N_bar}
\end{align}
That is, $\boldsymbol{\gamma}_N(K)\in\overline{\Gamma}_N(K)$ is any grouping of $N$ cross-sections into $K$ groups that partitions the set of true groups, $\{G_1^0,\ldots,G_{K^0}^0\}$.

\begin{lem}
    Consider model \eqref{model:without_gfe}. Suppose Assumptions \ref{asm:group_size}, \ref{asm:basic}(a)-(d) and (f) hold. Then for any fixed $K\in\{K^0+1,K^0+2,\ldots\}$ and for any $\boldsymbol{\gamma}_N\in\overline{\Gamma}_N(K)$, 
    \begin{align}
        \left|\widehat{\sigma}^2(K,\boldsymbol{\gamma}_N(K)) - \widehat{\sigma}^2\left(K^0,\widehat{\boldsymbol{\gamma}}_N^{(K^0)}\right)\right|=O_p\left((NT)^{-1}\right).
    \end{align}
    \label{lem:overfit}
\end{lem}
\noindent\begin{proof}
    Fix any $K\in\{K^0+1,K^0+2,\ldots\}$ and any $\boldsymbol{\gamma}_N(K)=(k_1,\ldots,k_N)\in\overline{\Gamma}_N(K)$. Recalling from Theorem \ref{thm:consistency_normality} that $\widehat{\sigma}^2\left(K^0,\widehat{\boldsymbol{\gamma}}_N^{(K^0)}\right)=(NT)^{-1}\sum_{i=1}^N\sum_{t=1}^T\varepsilon_{it}^2 + O_p((NT)^{-1})$, we have
    \begin{align}
        \left|\widehat{\sigma}^2(K,\boldsymbol{\gamma}_N(K)) - \widehat{\sigma}^2\left(K^0,\widehat{\boldsymbol{\gamma}}_N^{(K^0)}\right)\right| \leq KJ_{NT} + O_p\left((NT)^{-1}\right),
    \end{align}
    where $J_{NT} \coloneqq \max_{k\in[K]}|\inf_{\theta\in\Theta}S_k(\theta)|$ with $S_k(\theta)=(NT)^{-1}\sum_{i:k_i=k}\sum_{t=1}^T[(y_{it}-x_{it}'\theta)^2-\varepsilon_{it}^2]$. Letting $\overline{\theta}_k\coloneqq\argmin_{\theta\in\Theta}S_k(\theta)$, we have $\overline{\theta}_k=(\sum_{i:k_i=k}\sum_{t=1}^Tx_{it}x_{it}')^{-1}\sum_{k_i=k}\sum_{t=1}^Tx_{it}y_{it}$ under Assumption \ref{asm:basic}(f). We can readily show that $\overline{\theta}_k-\theta_{k^*}^0=O_p((|G_k(\boldsymbol{\gamma}_N)|T)^{-1/2})$, where $k^*$ is the common true group label shared by all $i$ such that $k_i=k$, and that $S_k(\overline{\theta}_k)=O_p((NT)^{-1})$. It follows that
    \begin{align}
        \left|\widehat{\sigma}^2(K,\boldsymbol{\gamma}_N(K)) - \widehat{\sigma}^2\left(K^0,\widehat{\boldsymbol{\gamma}}_N^{(K^0)}\right)\right|
        &\leq K O_p\left((NT)^{-1}\right) + O_p\left((NT)^{-1}\right) \\
        &=O_p\left((NT)^{-1}\right).
    \end{align}
\end{proof}

\noindent\begin{proof}[Proof of Proposition \ref{prop:ic_inconsistency}]
    \noindent(i) Pick any $\overline{K}\in\{K^0+1,\ldots,K_{\max}\}$. For the OLS estimate of group memberships obtained under $K=\overline{K}$, $\widehat{\boldsymbol{\gamma}}_N^{(\overline{K})} = (\widehat{k}_1^{(\overline{K})},\ldots,\widehat{k}_N^{(\overline{K})})\in[\overline{K}]^N$, let $\overline{\boldsymbol{\gamma}}_N(K^*)=(\overline{k}_1,\ldots,\overline{k}_N)\in[K^*]^N$ be a grouping such that $\overline{k}_i=\overline{k}_j$ if and only if $k_i^0=k_j^0$ and $\widehat{k}_i^{(\overline{K})} = \widehat{k}_j^{(\overline{K})}$. Note that $\overline{\boldsymbol{\gamma}}_N(K^*)=(\overline{k}_1,\ldots,\overline{k}_N)\in\overline{\Gamma}_N(K^*)$, where $\overline{\Gamma}_N(K)$ is defined by \eqref{def:gamma_N_bar}. Since $\overline{\boldsymbol{\gamma}}_N(K^*)$ partitions groups formed under $\widehat{\boldsymbol{\gamma}}_N^{(\overline{K})}$, the OLS fitting based on $\overline{\boldsymbol{\gamma}}_N(K^*)$ is not worse than that based on $\widehat{\boldsymbol{\gamma}}_N^{(\overline{K})}$, which implies $\widehat{\sigma}^2(K^*,\overline{\boldsymbol{\gamma}}_N(K^*)) \leq \widehat{\sigma}^2(\overline{K},\widehat{\boldsymbol{\gamma}}_N^{(\overline{K})})$. Consequently, we get
    \begin{align}
        &P\left(\mathrm{IC}(K^0,h_{NT}) < \min_{K^0+1\leq \overline{K}\leq K_{\max}}\mathrm{IC}(\overline{K},h_{NT})\right) \\
        &=P\left(\min_{K^0+1\leq \overline{K}\leq K_{\max}}\left[p\left(\overline{K}-K^0\right)\widetilde{\sigma}^2h_{NT} + \left(\widehat{\sigma}^2(\overline{K},\widehat{\boldsymbol{\gamma}}_N^{(\overline{K})}) - \widehat{\sigma}^2(K^0,\widehat{\boldsymbol{\gamma}}_N^{(K^0)})\right)\right] > 0\right) \\
        &\geq P\left(\min_{K^0+1\leq \overline{K}\leq K_{\max}}\left[p\left(\overline{K}-K^0\right)\widetilde{\sigma}^2h_{NT} + \left(\widehat{\sigma}^2(K^*,\overline{\boldsymbol{\gamma}}_N(K^*)) - \widehat{\sigma}^2(K^0,\widehat{\boldsymbol{\gamma}}_N^{(K^0)})\right)\right] > 0\right) \\
        &\to 1
        \label{convergence:overfit}
    \end{align}
    if $NTh_{NT}\to\infty$, in view of Lemma \ref{lem:overfit}. This proves part (i).
    
    \noindent(ii) From Theorem \ref{thm:consistency_normality}(iii), we have
    \begin{align}
        \mathrm{IC}(K^0,h_{NT}) 
        =(NT)^{-1}\sum_{i=1}^N\sum_{t=1}^T\varepsilon_{it}^2 + n(K^0)\widetilde{\sigma}^2h_{NT} + O_p\left((NT)^{-1}\right).
        \label{eqn:ic_K^0}
    \end{align}

    Take any $\underline{K}\in\{m,\ldots,K^0-1\}$ and $\boldsymbol{\gamma}_N(\underline{K},m)=(k_1,\ldots,k_N)\in\Gamma_N(\underline{K},m)$ for each $N$. We first decompose $\widehat{\sigma}^2(\underline{K},\boldsymbol{\gamma}_N(\underline{K},m))$ as
    \begin{align}
        \widehat{\sigma}^2(\underline{K},\boldsymbol{\gamma}_N(\underline{K},m)) 
        &= \frac{1}{NT}\sum_{k=1}^m\sum_{i\in G_k(\boldsymbol{\gamma}_N(\underline{K},m))}\sum_{t=1}^T\left[\varepsilon_{it} - x_{it}'\left\{\widehat{\theta}_k(\boldsymbol{\gamma}_N(\underline{K},m)) - \theta_{k_i^0}^0\right\}\right]^2 \\
        &\qquad +\frac{1}{NT}\sum_{k=m+1}^{\underline{K}}\sum_{i\in G_k(\boldsymbol{\gamma}_N(\underline{K},m))}\sum_{t=1}^T\left[\varepsilon_{it} - x_{it}'\left\{\widehat{\theta}_k(\boldsymbol{\gamma}_N(\underline{K},m)) - \theta_{k_i^0}^0\right\}\right]^2 \\
        &=C_{NT,1} + C_{NT,2} + C_{NT,3},
    \end{align}
    where
    \begin{align}
        C_{NT,1} &\coloneqq \frac{1}{NT}\sum_{k=1}^m\sum_{i\in G_k^0(0,\boldsymbol{\gamma}_N(\underline{K},m))}\sum_{t=1}^T\left\{\varepsilon_{it} - x_{it}'\left(\widehat{\theta}_k(\boldsymbol{\gamma}_N(\underline{K},m)) - \theta_k^0\right)\right\}^2, \\
        C_{NT,2} &\coloneqq \frac{1}{NT}\sum_{k=1}^m\sum_{j=1,j\neq k}^{K^0}\sum_{i\in G_j^0\cap G_k(\boldsymbol{\gamma}_N(\underline{K},m))}\sum_{t=1}^T\left\{\varepsilon_{it} - x_{it}'\left(\widehat{\theta}_k(\boldsymbol{\gamma}_N(\underline{K},m)) - \theta_{k_i^0}^0\right)\right\}^2, 
        \intertext{and}
        C_{NT,3} &\coloneqq \frac{1}{NT}\sum_{k=m+1}^{\underline{K}}\sum_{i\in G_k(\boldsymbol{\gamma}_N(\underline{K},m))}\sum_{t=1}^T\left\{\varepsilon_{it} - x_{it}'\left(\widehat{\theta}_k(\boldsymbol{\gamma}_N(\underline{K},m)) - \theta_{k_i^0}^0\right)\right\}^2,
    \end{align}
    where $\sum_{k=m+1}^{\underline{K}}a_k=0$ if $\underline{K}<m+1$. Applying Lemma \ref{lem:underfit} yields
    \begin{align}
        C_{NT,1} 
        &= \frac{1}{NT}\sum_{k=1}^m\sum_{i\in G_k^0(0,\boldsymbol{\gamma}_N(\underline{K},m))}\sum_{t=1}^T(\varepsilon_{it} - x_{it}'b_{k,NT})^2 \\
        &= \frac{1}{NT}\sum_{k=1}^m\sum_{i\in G_k^0(0,\boldsymbol{\gamma}_N(\underline{K},m))}\sum_{t=1}^T\varepsilon_{it}^2 + O_p\left((NT)^{-1/2}b_{k,NT}\right) + O_p\left(b_{k,NT}^2\right) \\
        &=\frac{1}{NT}\sum_{i=1}^N\sum_{t=1}^T\varepsilon_{it}^2 + O_p(N^{\alpha_{m+1}-1}),
    \end{align}
    where the second equality holds by Assumption \ref{asm:basic}(d), and the last equality follows from the fact that $|G_k^0(1,\boldsymbol{\gamma}_N(\underline{K},m))|\leq N^{\alpha_{m+1}}$ for $k\in[m]$ by the definition of $\Gamma_N(\underline{K},m)$, and that $|G_{m+1}^0\cup\cdots\cup G_{K^0}^0|=O(N^{\alpha_{m+1}})$ by Assumption \ref{asm:group_size}.

    For $C_{NT,2}$, using Assumptions \ref{asm:basic}(a)-(c) and the fact that $|G_j^0\cap G_k(\boldsymbol{\gamma}_N(\underline{K},m))|=O(N^{\alpha_{m+1}})$ for $j\neq k$, we readily obtain $C_{NT,2}=O_p(N^{\alpha_{m+1}-1})$. Similarly, noting that $|G_k(\boldsymbol{\gamma}_N(\underline{K},m))| = O(N^{\alpha_{m+1}})$ for $k\in\{m+1,\ldots,\underline{K}\}$ (if $\underline{K}\geq m+1$) by the definition of $\Gamma_N(\underline{K},m)$ and Assumption \ref{asm:group_size}, we deduce $C_{NT,3}=O_p(N^{\alpha_{m+1}-1})$. It follows that
    \begin{align}
        \widehat{\sigma}^2(\underline{K},\boldsymbol{\gamma}_N(\underline{K},m)) = \frac{1}{NT}\sum_{i=1}^N\sum_{t=1}^T\varepsilon_{it}^2 + O_p(N^{\alpha_{m+1}-1}).
    \end{align}
    With this and the fact that $\widehat{\sigma}^2(\underline{K},\widehat{\boldsymbol{\gamma}}_N^{(\underline{K})})\leq\widehat{\sigma}^2(\underline{K},\boldsymbol{\gamma}_N(\underline{K},m))$, we arrive at
    \begin{align}
        \mathrm{IC}(\underline{K},h_{NT}) 
        &\leq \widehat{\sigma}^2(\underline{K},\boldsymbol{\gamma}_N(\underline{K},m)) + n(\underline{K})\widetilde{\sigma}^2h_{NT} \\
        &= \frac{1}{NT}\sum_{i=1}^N\sum_{t=1}^T\varepsilon_{it}^2 +n(\underline{K})\widetilde{\sigma}^2h_{NT} + O_p(N^{\alpha_{m+1}-1}).
        \label{bound:ic_underfit}
    \end{align}

    Combining \eqref{eqn:ic_K^0} and \eqref{bound:ic_underfit} gives
    \begin{align}
        P\left(\mathrm{IC}(\underline{K},h_{NT})< \mathrm{IC}(K^0,h_{NT})\right)
        &\geq P\left(p(K^0-\underline{K})\widetilde{\sigma}^2N^{1-\alpha_{m+1}}h_{NT} + O_p(1)>0\right) \\
        &\to 1,
        \label{convergence:underestimate}
    \end{align}
    where the last convergence follows from the fact that $K^0-\underline{K}\geq 1$ and the condition that $N^{1-\alpha_{m+1}}h_{NT}\to\infty$. With this and part (i), we conclude $P(\widehat{K}(h_{NT})<K^0)\to1$ if $N^{1-\alpha_{m+1}}h_{NT}\to\infty$, completing the proof.
\end{proof}

\section{Appendix B: Proofs of Results of Section \ref{sec:with_gfe}}

\setcounter{equation}{0}
\renewcommand{\theequation}{B.\arabic{equation}}

\setcounter{lem}{0}
\renewcommand{\thelem}{B.\arabic{lem}}

\setcounter{prop}{0}
\renewcommand{\theprop}{B.\arabic{prop}}
For $(\boldsymbol{\theta},\boldsymbol{\mu},\boldsymbol{\gamma}_N)\in\Theta^{K^0}\times\mathcal{M}^{K^0T}\times[K^0]^N$, where $\boldsymbol{\gamma}_N = (k_1,\ldots,k_N)$, define
\begin{align}
    \widehat{Q}(\boldsymbol{\theta},\boldsymbol{\mu},\boldsymbol{\gamma}_N)&\coloneqq \frac{1}{NT}\sum_{i=1}^N\sum_{t=1}^T\left(y_{it}-x_{it}'\theta_{k_i} - \mu_{k_it}\right)^2,
    \intertext{and}
    \widetilde{Q}(\boldsymbol{\theta},\boldsymbol{\mu},\boldsymbol{\gamma}_N)&\coloneqq\frac{1}{NT}\sum_{i=1}^N\sum_{t=1}^T\left\{x_{it}'\left(\theta_{k_i^0}^0-\theta_{k_i}\right) + \left(\mu_{k_i^0t}^0-\mu_{k_it}\right)\right\}^2 + \frac{1}{NT}\sum_{i=1}^N\sum_{t=1}^T\varepsilon_{it}^2.
\end{align}
Note that $(\boldsymbol{\theta}^0,\boldsymbol{\mu}^0,\boldsymbol{\gamma}_N^0)$ minimizes $\widetilde{Q}(\boldsymbol{\theta},\boldsymbol{\mu},\boldsymbol{\gamma}_N)$.

\begin{lem}
    Consider model \eqref{model:with_gfe}. Suppose that Assumptions \ref{asm:basic_gfe}(a') and (d') hold. Then we have
    \begin{align}
        \sup_{(\boldsymbol{\theta},\boldsymbol{\mu},\boldsymbol{\gamma}_N)\in\Theta^{K^0}\times\mathcal{M}^{K^0T}\times[K^0]^N}\left|\widehat{Q}(\boldsymbol{\theta},\boldsymbol{\mu},\boldsymbol{\gamma}_N) - \widetilde{Q}(\boldsymbol{\theta},\boldsymbol{\mu},\boldsymbol{\gamma}_N)\right| = O_p(1/\sqrt{NT}).
    \end{align}
    \label{lem:Qhat_Qtilde_gfe}
\end{lem}

\noindent\begin{proof}
    We have
    \begin{align}
        \widehat{Q}(\boldsymbol{\theta},\boldsymbol{\mu},\boldsymbol{\gamma}_N) - \widetilde{Q}(\boldsymbol{\theta},\boldsymbol{\mu},\boldsymbol{\gamma}_N) = \frac{2}{NT}\sum_{i=1}^N\sum_{t=1}^T\varepsilon_{it}\left\{x_{it}'\left(\theta_{k_i^0}^0-\theta_{k_i}\right) + \mu_{k_i^0t}^0 - \mu_{k_it}\right\}.
    \end{align}
    As in the proof of Lemma \ref{lem:Qhat_Qtilde}, $(NT)^{-1}\sum_{i=1}^N\sum_{t=1}^T\varepsilon_{it}x_{it}'(\theta_{k_i^0}^0-\theta_{k_i})=O_p(1/\sqrt{NT})$ uniformly in $(\boldsymbol{\theta},\boldsymbol{\mu},\boldsymbol{\gamma}_N)$. Furthermore, $(NT)^{-1}\sum_{i=1}^N\sum_{t=1}^T\varepsilon_{it}(\mu_{k_i^0t}^0-\mu_{k_it})=O_p(1/\sqrt{NT})$ by Assumptions \ref{asm:basic_gfe}(a') and (d'). This completes the proof.
\end{proof}

For $\boldsymbol{\mu},\overline{\boldsymbol{\mu}}\in\mathcal{M}^{K^0T}$, define
\begin{align}
    d_H^2\left(\boldsymbol{\mu},\overline{\boldsymbol{\mu}}\right)\coloneqq\max\left\{\max_{k\in[K^0]}\min_{\widetilde{k}\in[K^0]}T^{-1}\sum_{t=1}^T(\mu_{\widetilde{k}t} - \overline{\mu}_{kt})^2,\max_{\widetilde{k}\in[K^0]}\min_{k\in[K^0]}T^{-1}\sum_{t=1}^T(\mu_{\widetilde{k}t} - \overline{\mu}_{kt})^2\right\}.
\end{align}

\begin{lem}
    Consider model \eqref{model:with_gfe}. Suppose that Assumptions \ref{asm:group_size_gfe} and \ref{asm:basic_gfe}(a')-(h') hold. Then we have $d_H\left(\widehat{\boldsymbol{\theta}}^{(K^0)},\boldsymbol{\theta}^0\right) \stackrel{p}{\to} 0$, and $d_H\left(\widehat{\boldsymbol{\mu}}^{(K^0)},\boldsymbol{\mu}^0\right)\stackrel{p}{\to}0$.
    \label{lem:theta_consistency_gfe}
\end{lem}

\noindent\begin{proof}
    The proof closely follows that of Lemma S.4 of \citet{bonhommeGroupedPatternsHeterogeneity2015}. Using Assumption \ref{asm:basic_gfe}(e'), we obtain
    \begin{align}
        &\widetilde{Q}\left(\boldsymbol{\theta},\boldsymbol{\mu},\boldsymbol{\gamma}_N\right) - \widetilde{Q}\left(\boldsymbol{\theta}^0,\boldsymbol{\mu}^0,\boldsymbol{\gamma}_N^0\right) \\ 
        &= \sum_{k=1}^{K^0}\sum_{\widetilde{k}=1}^{K^0}\frac{1}{N}\sum_{i=1}^N1\{k_i^0=k\}1\{k_i=\widetilde{k}\}\frac{1}{T}\sum_{t=1}^T\Bigl\{\left(\theta_{k}^0-\theta_{\widetilde{k}}\right)'x_{it}x_{it}'\left(\theta_{k}^0-\theta_{\widetilde{k}}\right) + 2\left(\theta_{k}^0-\theta_{\widetilde{k}}\right)'x_{it}(\mu_{kt}-\mu_{\widetilde{k}t}) \\
        & \hspace{8cm}+ (\mu_{kt}-\mu_{\widetilde{k}t})^2\Bigr\} \\
        &\geq \sum_{k=1}^{K^0}\frac{N_k^0}{N}\sum_{\widetilde{k}=1}^{K^0}\rho_{NT}^*(\boldsymbol{\gamma}_N,k,\widetilde{k})\left\{\left\|\theta_k^0-\theta_{\widetilde{k}}\right\|^2 + T^{-1}\sum_{t=1}^T\left(\mu_{kt}^0-\mu_{\widetilde{k}t}\right)^2\right\} \\
        &\geq \widehat{\rho}_{NT}^*\sum_{k=1}^{K^0}\frac{N_k^0}{N}\min_{\widetilde{k}\in[K^0]}\left\{\left\|\theta_k^0-\theta_{\widetilde{k}}\right\|^2 + T^{-1}\sum_{t=1}^T\left(\mu_{kt}^0-\mu_{\widetilde{k}t}\right)^2\right\}.
        \label{bound:Qtilde_diff_lower_gfe}
    \end{align}
    As in the proof of Lemma \ref{lem:theta_consistency}, an application of Lemma \ref{lem:Qhat_Qtilde_gfe} and the definition of $(\widehat{\boldsymbol{\theta}}^{(K^0)}, \widehat{\boldsymbol{\mu}}^{(K^0)}, \widehat{\boldsymbol{\gamma}}_N^{(K^0)})$ give
    \begin{align}
        \widetilde{Q}\left(\widehat{\boldsymbol{\theta}}^{(K^0)}, \widehat{\boldsymbol{\mu}}^{(K^0)}, \widehat{\boldsymbol{\gamma}}_N^{(K^0)}\right) - \widetilde{Q}\left(\boldsymbol{\theta}^0,\boldsymbol{\mu}^0,\boldsymbol{\gamma}_N^0\right) = O_p(1/\sqrt{NT}).
        \label{order:Qtilde_true_ols_diff_gfe}
    \end{align}

    Combining \eqref{bound:Qtilde_diff_lower_gfe}, \eqref{order:Qtilde_true_ols_diff_gfe}, and the fact that $\widehat{\rho}_{NT}^*\stackrel{p}{\to}\rho^*>0$, we get, w.p.a.1,
    \begin{align}
        &O_p(1/\sqrt{NT}) \geq \frac{\rho}{2}\frac{N_{K^0}^0}{N}\max_{k\in[K^0]}\min_{\widetilde{k}\in[K^0]}\left\{\left\|\theta_k^0-\widehat{\theta}_{\widetilde{k}}^{(K^0)}\right\|^2 + T^{-1}\sum_{t=1}^T\left(\mu_{kt}^0-\widehat{\mu}_{\widetilde{k}t}^{(K^0)}\right)^2\right\},
        \intertext{and hence}
        &\max_{k\in[K^0]}\min_{\widetilde{k}\in[K^0]}\left(\left\|\theta_k^0-\widehat{\theta}_{\widetilde{k}}^{(K^0)}\right\|^2 + T^{-1}\sum_{t=1}^T\left(\mu_{kt}^0-\widehat{\mu}_{\widetilde{k}t}^{(K^0)}\right)^2\right) = O_p\left(N^{1/2-\alpha_{K^0}}/\sqrt{T}\right) = o_p(1)
    \end{align}
    by Assumption \ref{asm:basic_gfe}(g).

    Define $\sigma^*(k)\coloneqq \argmin_{\widetilde{k}\in[K^0]}(\|\theta_k^0-\widehat{\theta}_{\widetilde{k}}^{(K^0)}\|^2 + T^{-1}\sum_{t=1}^T(\mu_{kt}^0 - \widehat{\mu}_{\widetilde{k}t}^{(K^0)})^2)$. Then following an argument closely similar to that starting after \eqref{order:hausdorff}, we can show that $\sigma^*(k)$ admits a well-defined inverse w.p.a.1, and that
    \begin{align}
        \max_{\widetilde{k}\in[K^0]}\min_{k\in[K^0]}\left(\left\|\theta_k^0-\widehat{\theta}_{\widetilde{k}}^{(K^0)}\right\|^2 + T^{-1}\sum_{t=1}^T\left(\mu_{kt}^0-\widehat{\mu}_{\widetilde{k}t}^{(K^0)}\right)^2\right) = o_p(1).
    \end{align}
    Consequently, we conclude $d_H\left(\widehat{\boldsymbol{\theta}}^{(K^0)},\boldsymbol{\theta}^0\right) \stackrel{p}{\to} 0$, and $d_H(\widehat{\boldsymbol{\mu}}^{(K^0)},\boldsymbol{\mu}^0)\stackrel{p}{\to}0$.
\end{proof}

Defining $\sigma^*(k)=k$ by relabeling the groups, we have $\|\widehat{\theta}_k^{(K^0)}-\theta_k^0\|\stackrel{p}{\to}0$ and $T^{-1}\sum_{t=1}^T(\widehat{\mu}_{kt}^{(K^0)}-\mu_{kt}^0)^2\stackrel{p}{\to}0$ for all $k\in[K^0]$. For any $\eta>0$, let $\mathcal{N}_\eta^*$ denote the set of parameters $(\boldsymbol{\theta},\boldsymbol{\mu})\in\Theta^{K^0}\times\mathcal{M}^{K^0T}$ that satisfy $\|\boldsymbol{\theta}-\boldsymbol{\theta}\|^2<\eta$ and $T^{-1}\sum_{t=1}^T(\mu_{kt}-\mu_{kt}^0)^2<\eta$ for all $k\in[K^0]$.

\begin{lem}
    Consider model \eqref{model:with_gfe}. Suppose that Assumptions \ref{asm:group_size_gfe} and \ref{asm:basic_gfe} hold. Then for $\eta>0$ small enough, and for all $\delta>0$, we have
    \begin{align}
        \sup_{(\boldsymbol{\theta},\boldsymbol{\mu})\in\mathcal{N}_\eta^*}\frac{1}{N}\sum_{i=1}^N1\left\{\widehat{k}_i^{(K^0)}\left(\boldsymbol{\theta},\boldsymbol{\mu}\right)\neq k_i^0\right\}=o_p\left(T^{-\delta}\right)
    \end{align}
    as $N,T\to\infty$, where $\widehat{k}_i^{(K^0)}(\boldsymbol{\theta},\boldsymbol{\mu})=\argmin_{k\in[K^0]}\sum_{t=1}^T(y_{it} - x_{it}'\theta_k - \mu_{kt})^2$ for given $(\boldsymbol{\theta},\boldsymbol{\mu})$.
    \label{lem:prob_missclassify_gfe}
\end{lem}

\noindent\begin{proof}
    As in the proof of Lemma \ref{lem:prob_missclassify}, we have
    \begin{align}
        &\frac{1}{N}\sum_{i=1}^N1\left\{\widehat{k}_i^{(K^0)}(\boldsymbol{\theta},\boldsymbol{\mu})\neq k_i^0\right\} \leq \sum_{k=1}^{K^0}\frac{1}{N}\sum_{i=1}^NZ_{ik}^*(\boldsymbol{\theta},\boldsymbol{\mu}),
        \label{bound:average_misclassification_gfe}
    \end{align}
        where $Z_{ik}^*(\boldsymbol{\theta},\boldsymbol{\mu})\coloneqq 1\{k_i^0\neq k\}1\{\sum_{t=1}^T(y_{it}-x_{it}'\theta_k-\mu_{kt})^2 \leq \sum_{t=1}^T(y_{it}-x_{it}'\theta_{k_i^0}-\mu_{k_i^0t})^2\}$. After some calculations, we obtain $Z_{ik}^*(\boldsymbol{\theta},\boldsymbol{\mu}) \leq \max_{\widetilde{k}\neq k}1\{A_T^*+B_T^*\leq 0\}$, where
    \begin{align} 
        &A_T^* \coloneqq \sum_{t=1}^T\left(\theta_{\widetilde{k}}-\theta_k\right)'x_{it}\left\{x_{it}'\left(\theta_{\widetilde{k}}^0 - \frac{\theta_k+\theta_{\widetilde{k}}}{2}\right) + \left(\mu_{\widetilde{k}t}^0 - \mu_{kt}\right) + \varepsilon_{it}\right\},
        \intertext{and}
        &B_T^* \coloneqq \sum_{t=1}^T\left(\mu_{\widetilde{k}t}-\mu_{kt}\right)\left\{\left(\mu_{\widetilde{k}t}^0 - \frac{\mu_{kt}+\mu_{\widetilde{k}t}}{2}\right) + x_{it}'\left(\theta_{\widetilde{k}}^0 - \theta_{\widetilde{k}}\right) + \varepsilon_{it}\right\}.
    \end{align}
    Furthermore, a bit lengthy calculation yields $A_T^*=A_{T,1}^*+A_{T,2}^*$ and $B_T^*=B_{T,1}^*+B_{T,2}^*$, where
    \begin{align}
        A_{T,1}^* &\coloneqq \left(\theta_{\widetilde{k}}-\theta_{\widetilde{k}}^0\right)'\sum_{t=1}^Tx_{it}x_{it}'\theta_{\widetilde{k}}^0 - \left(\theta_k-\theta_k^0\right)'\sum_{t=1}^Tx_{it}x_{it}'\theta_{\widetilde{k}}^0 \\
        &\quad -\frac{1}{2}\left(\theta_{\widetilde{k}}-\theta_{\widetilde{k}}^0\right)'\sum_{t=1}^Tx_{it}x_{it}'\left(\theta_{\widetilde{k}}+\theta_{\widetilde{k}}^0\right) + \frac{1}{2}\left(\theta_k-\theta_k^0\right)'\sum_{t=1}^Tx_{it}x_{it}'\left(\theta_k+\theta_k^0\right) \\
        &\quad + \left(\theta_{\widetilde{k}}-\theta_{\widetilde{k}}^0\right)\sum_{t=1}^Tx_{it}\left(\mu_{\widetilde{k}t}^0-\mu_{kt}^0\right) - \left(\theta_k-\theta_k^0\right)'\sum_{t=1}^Tx_{it}\left(\mu_{\widetilde{k}t}^0-\mu_{kt}^0\right) \\
        &\quad +\left(\theta_{\widetilde{k}}-\theta_k\right)'\sum_{t=1}^Tx_{it}\left(\mu_{kt}^0-\mu_{kt}\right) + \left(\theta_{\widetilde{k}}-\theta_{\widetilde{k}}^0\right)'\sum_{t=1}^Tx_{it}\varepsilon_{it} - \left(\theta_k-\theta_k^0\right)'\sum_{t=1}^Tx_{it}\varepsilon_{it}, \\
        A_{T,2}^* &\coloneqq \left(\theta_{\widetilde{k}}^0-\theta_k^0\right)'\sum_{t=1}^Tx_{it}\left\{x_{it}'\left(\theta_{\widetilde{k}}^0 - \frac{\theta_k^0+\theta_{\widetilde{k}}^0}{2}\right) + \left(\mu_{\widetilde{k}t}^0 - \mu_{kt}^0\right)+\varepsilon_{it}\right\}, \\
        B_{T,1}^* &\coloneqq \sum_{t=1}^T\left\{\left(\mu_{\widetilde{k}t}-\mu_{\widetilde{k}t}^0\right)\mu_{\widetilde{k}t}^0 - \left(\mu_{kt}-\mu_{kt}^0\right)\mu_{\widetilde{k}t}^0 - \frac{1}{2}\left(\mu_{\widetilde{k}t}-\mu_{\widetilde{k}t}^0\right)\left(\mu_{\widetilde{k}t}+\mu_{\widetilde{k}t}^0\right) \right. \\
        &\left. \qquad +\frac{1}{2}\left(\mu_{kt}-\mu_{kt}^0\right)\left(\mu_{kt}+\mu_{kt}^0\right) + \left(\mu_{\widetilde{k}t}-\mu_{kt}\right)x_{it}'\left(\theta_{\widetilde{k}}^0-\theta_{\widetilde{k}}\right) + \left(\mu_{\widetilde{k}t}-\mu_{\widetilde{k}t}^0\right)\varepsilon_{it} - \left(\mu_{kt}-\mu_{kt}^0\right)\varepsilon_{it}\right\},
        \intertext{and}
        B_{T,2}^* &\coloneqq \sum_{t=1}^T\left(\mu_{\widetilde{k}t}^0-\mu_{kt}^0\right)\left(\mu_{\widetilde{k}t}^0 - \frac{\mu_{\widetilde{k}t}^0+\mu_{kt}^0}{2} + \varepsilon_{it}\right).
    \end{align}
    Therefore
    \begin{align}
        Z_{ik}^*(\boldsymbol{\theta},\boldsymbol{\mu})
        &\leq \max_{\widetilde{k}\neq k}1\left\{A_{T,2}^* + B_{T,2}^*\leq\left|A_{T,1}^*\right| + \left|B_{T,1}^*\right|\right\} \\
        &\leq \max_{\widetilde{k}\neq k}1\Biggl\{\left(\theta_{\widetilde{k}}^0-\theta_k^0\right)'\sum_{t=1}^Tx_{it}\left\{x_{it}'\left(\theta_{\widetilde{k}}^0 - \frac{\theta_k^0+\theta_{\widetilde{k}}^0}{2}\right) + \left(\mu_{\widetilde{k}t}-\mu_{kt}^0\right) + \varepsilon_{it}\right\}\\
        & \hspace{2.5cm} +\sum_{t=1}^T\left(\mu_{\widetilde{k}t}^0-\mu_{kt}^0\right)\left(\mu_{\widetilde{k}t}^0 - \frac{\mu_{\widetilde{k}t}^0 + \mu_{kt}^0}{2}+\varepsilon_{it}\right) \\
        & \hspace{2.5cm} \leq TC\sqrt{\eta}\left(T^{-1}\sum_{t=1}^T\left\|x_{it}\right\|^2 + T^{-1}\sum_{t=1}^T\left\|x_{it}\right\| + \left\|T^{-1}\sum_{t=1}^Tx_{it}\varepsilon_{it}\right\| \right. \\
        &\left. \hspace{4.5cm} + \left(T^{-1}\sum_{t=1}^T\varepsilon_{it}^2\right)^{1/2} + 1\right)\Biggr\}
    \end{align}
    for $(\boldsymbol{\theta},\boldsymbol{\mu})\in\mathcal{N}_\eta^*$ and some constant $C>0$ independent of $\eta$ and $T$. It follows that $\sup_{(\boldsymbol{\theta},\boldsymbol{\mu})\in\mathcal{N}_\eta}Z^*_{i,k}(\boldsymbol{\theta},\boldsymbol{\mu})\leq\widetilde{Z}_{i,k}^*$, where
    \begin{align}
    \widetilde{Z}_{i,k}^*&\coloneqq \max_{\widetilde{k}\neq k}1\Biggl\{\left(\theta^0_{\widetilde{k}} - \theta_k^0\right)'\sum_{t=1}^Tx_{it}\varepsilon_{it} + \sum_{t=1}^T\left(\mu_{\widetilde{k}t}^0-\mu_{kt}^0\right)\varepsilon_{it} \\
    &\hspace{2.5cm} \leq -\frac{1}{2}\sum_{t=1}^T\left\{\left(\theta_{\widetilde{k}}^0-\theta_k^0\right)'x_{it} + \left(\mu_{\widetilde{k}t}-\mu_{kt}^0\right)\right\}^2 \\
    &\hspace{2.5cm} + TC\sqrt{\eta}\left(T^{-1}\sum_{t=1}^T\left\|x_{it}\right\|^2 + T^{-1}\sum_{t=1}^T\left\|x_{it}\right\| + \left\|T^{-1}\sum_{t=1}^Tx_{it}\varepsilon_{it}\right\| \right. \\
        &\left. \hspace{4.5cm} + \left(T^{-1}\sum_{t=1}^T\varepsilon_{it}^2\right)^{1/2} + 1\right)\Biggr\}.
    \end{align}
    Substituting this result into \eqref{bound:average_misclassification_gfe} gives
    \begin{align}
        \sup_{(\boldsymbol{\theta},\boldsymbol{\mu})\in\mathcal{N}_\eta}\frac{1}{N}\sum_{i=1}^N1\left\{\widehat{k}_i^{(K^0)}(\boldsymbol{\theta},\boldsymbol{\mu})\neq k_i^0\right\}\leq \sum_{k=1}^{K^0}\frac{1}{N}\sum_{i=1}^N\widetilde{Z}_{ik}^*. \label{bound:uniform_average_misclass_gfe}
    \end{align}

    As in the proof of Lemma \ref{lem:prob_missclassify}, we can show that
    \begin{align}
        P\left(\widetilde{Z}_{ik}^*=1\right)=o_p\left(T^{-\delta}\right)
        \label{prob:tilde_z_ik_gfe}
    \end{align}
    for all $\delta>0$ uniformly in $i$, although there are two major differences. First, using Assumptions \ref{asm:basic_gfe}(f') and (h'), we bound, for sufficiently large $T$,
    \begin{align}
        -\frac{1}{2}\sum_{t=1}^T\left\{\left(\theta_{\widetilde{k}}^0-\theta_k^0\right)'x_{it} + \left(\mu_{\widetilde{k}t}^0 - \mu_{kt}^0\right)\right\}^2
        &\leq -\frac{T}{2}\lambda_{\min}(D_{T,i})\left\{\left\|\theta_{\widetilde{k}}^0-\theta_k^0\right\|^2 + \frac{1}{T}\sum_{t=1}^T\left(\mu_{\widetilde{k}t}^0 - \mu_{kt}^0\right)^2\right\} \\
        &\leq -\frac{Tc^3}{8}.
    \end{align}
    Second, we need to show $P(T^{-1}\sum_{t=1}^T\varepsilon_{it}^2\geq\widetilde{M})=o(T^{-\delta})$ for a sufficiently large constant $\widetilde{M}$. Using Assumptions \ref{asm:basic_gfe}(c) and (j'), we get
    \begin{align}
        P\left(\frac{1}{T}\sum_{t=1}^T\varepsilon_{it}^2 \geq\widetilde{M}\right) = P\left(\frac{1}{T}\sum_{t=1}^T\left(\varepsilon_{it}^2 - E\left[\varepsilon_{it}^2\right]\right) \geq \widetilde{M} - \frac{1}{T}\sum_{t=1}^TE\left[\varepsilon_{it}^2\right] \right) = o\left(T^{-\delta}\right)
    \end{align}
    for $\widetilde{M} > \sqrt{M}$. With \eqref{bound:uniform_average_misclass_gfe} and \eqref{prob:tilde_z_ik_gfe}, we can derive the desired result as in the proof of Lemma \ref{lem:prob_missclassify}.
\end{proof}

Let $\widetilde{\mu}_{kt}^{(K^0)}$ denote the oracle estimator of $\mu_{kt}^0$ calculated using cross-sections belonging to $G_k^0$.

\begin{prop}
    Consider model \eqref{model:with_gfe}. Suppose that Assumptions \ref{asm:group_size_gfe} and \ref{asm:basic_gfe} hold. Then as $N,T\to\infty$, we have, for all $\delta>0$,
    
    \noindent (i) $\widehat{\theta}_k^{(K^0)} = \widetilde{\theta}_k^{(K^0)} + o_p\left(N^{(1-\alpha_k)/2}T^{-\delta}\right)$ for all $k\in[K^0]$,
    
    \noindent (ii) $\widehat{\mu}_k^{(K^0)} = \widetilde{\mu}_k^{(K^0)} + o_p\left(N^{(1-\alpha_k)/2}T^{-\delta}\right)$ for all $k\in[K^0]$ and $t\in[T]$,
    
    \noindent (iii) $P\left(\bigcup_{i=1}^N\left\{\widehat{k}_i^{(K^0)}\neq k_i^0\right\}\right)=o(1)+o\left(NT^{-\delta}\right)$.
    \label{prop:asym_equivalence_gfe}
\end{prop}

\noindent\begin{proof}
    The proof is essentially the same as that of Proposition \ref{prop:asym_equivalence} and hence is omitted.
\end{proof}

\noindent\begin{proof}[Proof of Theorem \ref{thm:consistency_normality_gfe}]

    \noindent(i) Part (i) immediately follows from Proposition \ref{prop:asym_equivalence_gfe}(iii) and Assumption \ref{asm:basic_gfe}(g).

    \noindent(ii) From Proposition \ref{prop:asym_equivalence_gfe}(i) and Assumption \ref{asm:basic_gfe}(g), we have $\sqrt{N_{k}^0T}(\widehat{\theta}_k^{(K^0)} - \theta_k^0) = \sqrt{N_{k}^0T}(\widetilde{\theta}_k^{(K^0)} - \theta_k^0) + o_p(1)$ for each $k\in[K^0]$. The result follows from Assumptions \ref{asm:dist_gfe}(a), (b') and (c') and the fact that
    \begin{align}
        \sqrt{N_{k}^0T}\left(\widetilde{\theta}_k^{(K^0)} - \theta_k^0\right) = \left(\frac{1}{N_{k}^0T}\sum_{i\in G_k^0}\sum_{t=1}^T(x_{it} - \overline{x}_{kt})(x_{it} - \overline{x}_{kt})'\right)^{-1}\frac{1}{\sqrt{N_k^0T}}\sum_{i\in G_k^0}\sum_{t=1}^T(x_{it} - \overline{x}_{kt})\varepsilon_{it}.
    \end{align}

    \noindent(iii) Noting that $\sqrt{N_k^0}(\widehat{\mu}_{kt}^{(K^0)} - \mu_{kt}^0) = \sqrt{N_k^0}(\widetilde{\mu}_{kt}^{(K^0)} - \mu_{kt}^0) + o_p(1)$ in view of Proposition \ref{prop:asym_equivalence_gfe}(ii) and Assumption \ref{asm:basic_gfe}(g), and that
    \begin{align}
        \widetilde{\mu}_{kt} 
        = \overline{y}_{kt} - \overline{x}_{kt}'\widetilde{\theta}_{k} = \mu_{kt}^0 + (N_k^0)^{-1}\sum_{i\in G_k^0}\varepsilon_{it} - \overline{x}_{kt}'\left(\widetilde{\theta}_k-\theta_k^0\right),
    \end{align}
    where $\overline{y}_{kt}\coloneqq (N_k^0)^{-1}\sum_{i\in G_k^0}y_{it}$, the result follows from part (ii) and Assumptions \ref{asm:basic_gfe}(b) and \ref{asm:dist_gfe}(d) and (e).

    \noindent(iv) As in the proof of Theorem \ref{thm:consistency_normality}(iii), we can show that
    \begin{align}
        \widehat{\sigma}^2\left(K^0,\widehat{\boldsymbol{\gamma}}_N^{(K^0)}\right) = \frac{1}{NT}\sum_{k=1}^{K^0}\sum_{i\in G_k^0}\sum_{t=1}^T\left(y_{it}-x_{it}'\widehat{\theta}_k^{(K^0)} - \widehat{\mu}_{kt}^{(K^0)}\right)^2 + o_p(1/N).
    \end{align}
    Noting that $\widehat{\theta}_k^{(K^0)} - \theta_k^0 = \widetilde{\theta}_k^{(K^0)} - \theta_k^0+o_p(1/\sqrt{N})$ and $\widehat{\mu}_{kt}^{(K^0)} - \mu_{kt}^0=\widetilde{\mu}_{kt}^{(K^0)} - \mu_{kt}^0 + o_p(1/\sqrt{N})$, and defining $a_{NT}^*=o_p(1/\sqrt{N})$, we get
    \begin{align}
        \widehat{\sigma}^2\left(K^0,\widehat{\boldsymbol{\gamma}}_N^{(K^0)}\right) 
        &=\frac{1}{NT}\sum_{k=1}^{K^0}\sum_{i\in G_k^0}\sum_{t=1}^T\left\{\varepsilon_{it} - x_{it}'\left(\widetilde{\theta}_k^{(K^0)} - \theta_k^0\right) - \left(\widetilde{\mu}_{kt}^{(K^0)} - \mu_{kt}\right) + a_{NT}^*\right\}^2 + o_p(1/N) \\
        &= \frac{1}{NT}\sum_{i=1}^N\sum_{t=1}^T\varepsilon_{it}^2 - \frac{2}{NT}\sum_{k=1}^{K^0}\sqrt{N_k^0T}\left(\widetilde{\theta}_k^{(K^0)} - \theta_k^0\right)'\frac{1}{\sqrt{N_k^0T}}\sum_{i\in G_k^0}\sum_{t=1}^Tx_{it}\varepsilon_{it} \\
        &\qquad -\frac{2}{N}\sum_{k=1}^{K^0}\frac{1}{T}\sum_{t=1}^T\sqrt{N_k^0}\left(\widetilde{\mu}_{kt}^{(K^0)} - \mu_{kt}^0\right)\frac{1}{\sqrt{N_k^0}}\sum_{i\in G_k^0}\varepsilon_{it} \\
        &\qquad + \frac{1}{NT}\sum_{k=1}^{K^0}\sqrt{N_k^0T}\left(\widetilde{\theta}_k^{(K^0)} - \theta_k^0\right)'\frac{1}{N_k^0T}\sum_{i\in G_k^0}\sum_{t=1}^Tx_{it}x_{it}'\sqrt{N_k^0T}\left(\widetilde{\theta}_k^{(K^0)} - \theta_k^0\right) \\
        &\qquad + \frac{2}{N}\sum_{k=1}^{K^0}\sqrt{N_k^0}\left(\widetilde{\theta}_k^{(K^0)} - \theta_k^0\right)\frac{1}{T}\sum_{t=1}^T\sqrt{N_k^0}\left(\widetilde{\mu}_{kt}^{(K^0)} - \mu_{kt}^0\right)\frac{1}{N_k^0}\sum_{i\in G_k^0}x_{it} \\
        &\qquad + \frac{1}{N}\sum_{k=1}^{K^0}\frac{1}{T}\sum_{t=1}^T(N_k^0)^2\left(\widetilde{\mu}_{kt}^{(K^0)} - \mu_{kt}^0\right)^2 + O_p(1/N) \\
        &=\frac{1}{NT}\sum_{i=1}^N\sum_{t=1}^T\varepsilon_{it}^2 + O_p(1/N) \\
        &\stackrel{p}{\to} \sigma^2.
    \end{align}
\end{proof}

Recall $\Gamma_N(K,m)$, which is defined in Definition \ref{def:gamma_N_K_m}. Also recall that we let $G_k(\boldsymbol{\gamma}_N)=\{i\in[N]:k_i=k\}$ for any $\boldsymbol{\gamma}_N=(k_1,\ldots,k_N)\in[K]^N$. For any $\boldsymbol{\gamma}_N(K,m)\in\Gamma_N(K,m)$, we define $\widehat{\boldsymbol{\theta}}^{(K)}(\boldsymbol{\gamma}_N(K,m))$ and $\widehat{\boldsymbol{\mu}}^{(K)}(\boldsymbol{\gamma}_N(K,m))$ as
\begin{align}
    &\widehat{\theta}^{(K)}_k(\boldsymbol{\gamma}_N(K,m)) \\
    &\coloneqq \left\{\sum_{i\in G_k(\boldsymbol{\gamma}_N(K,m))}\sum_{t=1}^T\left(x_{it} - \bar{x}_{k,t}\right)\left(x_{it}-\bar{x}_{k,t}\right)'\right\}^{-1}\sum_{i\in G_k(\boldsymbol{\gamma}_N(K,m))}\sum_{t=1}^T\left(x_{it} - \bar{x}_{k,t}\right)\left(y_{it}-\bar{y}_{k,t}\right),
\end{align}
and $\widehat{\mu}_{kt}^{(K)}(\boldsymbol{\gamma}_N(K,m)) \coloneqq \bar{y}_{k,t} - \bar{x}_{k,t}'\widehat{\theta}_k^{(K)}(\boldsymbol{\gamma}_N(K,m))$, where $\bar{x}_{k,t} \coloneqq |G_k(\boldsymbol{\gamma}_N(K,m))|^{-1}\sum_{i\in G_k(\boldsymbol{\gamma}_N(K,m))}x_{it}$ and $\bar{y}_{k,t} \coloneqq |G_k(\boldsymbol{\gamma}_N(K,m))|^{-1}\sum_{i\in G_k(\boldsymbol{\gamma}_N(K,m))}y_{it}$, for $k\in[K]$.

\begin{lem}
    Consider model \eqref{model:with_gfe}. Suppose Assumptions \ref{asm:group_size_gfe} and \ref{asm:basic_gfe}(a')-(d') hold. Suppose $\underline{K}\in\{m,m+1,\ldots,K^0-1\}$. For each $N\in\mathbb{N}$, take any $\boldsymbol{\gamma}_N(\underline{K},m)=(k_1,\ldots,k_N)\in\Gamma_N(\underline{K},m)$. Then the following results hold.
    
    \noindent(i) $\widehat{\theta}_k^{(\underline{K})}(\boldsymbol{\gamma}_N(\underline{K},m))=\theta_k^0+b_{1k,NT}^*$, and $\widehat{\mu}_{kt}^{(\underline{K})}(\boldsymbol{\gamma}_N(\underline{K},m))=\mu_{kt}^0+b_{2k,N}$ for each $k\in[m]$, where $b_{1k,NT}^*=O_p(\max\{(NT)^{-1/2}, N^{\alpha_{m+1}-1}\})$, and $b_{2k,N}=O_p(\max\{N^{-1/2},N^{\alpha_{m+1}-1}\})$. 

    \noindent(ii) If $\underline{K}\geq m+1$, then $\widehat{\theta}_k^{(\underline{K})}(\boldsymbol{\gamma}_N(\underline{K},m))=O_p(1)$ and $\widehat{\mu}_{kt}^{(\underline{K})}(\boldsymbol{\gamma}_N(\underline{K},m))=O_p(1)$ for all $k\in\{m+1,\ldots,\underline{K}\}$ and $t$.
    \label{lem:underfit_gfe}
\end{lem}

\noindent\begin{proof}
    \noindent(i) Take any $k\in[m]$. Because $G_k^0=G_k^0(0,\boldsymbol{\gamma}_N(\underline{K},m))\cup G_k^0(1,\boldsymbol{\gamma}_N(\underline{K},m))$, $N_k^0\geq |G_k^0(0,\boldsymbol{\gamma}_N(\underline{K},m))| \geq N_k^0-N^{\alpha_{m+1}}$, and $|G_j^0\cap G_k(\boldsymbol{\gamma}_N(\underline{K},m))|=O(N^{\alpha_{m+1}})$ for $j\neq k$, $\bar{x}_{k,t}$ satisfies
    \begin{align}
        \bar{x}_{k,t} 
        &= \frac{|G_k^0(0,\boldsymbol{\gamma}_N(\underline{K},m))|}{|G_k(\boldsymbol{\gamma}_N(\underline{K},m))|} \frac{1}{|G_k^0(0,\boldsymbol{\gamma}_N(\underline{K},m))|}\left(\sum_{i\in G_k^0(0,\boldsymbol{\gamma}_N(\underline{K},m))} + \sum_{j\neq k}\sum_{i\in G_j^0\cap G_k(\boldsymbol{\gamma}_N(\underline{K},m))}\right)x_{it} \\
        &=\bar{x}_{k\wedge k(0),t} + O_p(N^{\alpha_{m+1}-1}),
    \end{align}
    where $\bar{x}_{k\wedge k(0),t}\coloneqq |G_k^0(0,\boldsymbol{\gamma}_N(\underline{K},m))|^{-1}\sum_{i\in G_k^0(0,\boldsymbol{\gamma}_N(\underline{K},m))}x_{it}$. Similarly, we have $\bar{y}_{k,t} = \bar{x}_{k\wedge k(0),t}'\theta_k^0 + \bar{\varepsilon}_{k\wedge k(0),t} + \mu_{kt}^0 + O_p(N^{\alpha_{m+1}-1})$. Hence, the denominator and numerator of $\widehat{\theta}_k^{(\underline{K})}(\boldsymbol{\gamma}_N(\underline{K},m))$ satisfy
    \begin{align}
        &\frac{1}{|G_k^0(0,\boldsymbol{\gamma}_N(\underline{K},m))|T}\sum_{i\in G_k(\boldsymbol{\gamma}_N(\underline{K},m))}\sum_{t=1}^T\left(x_{it} - \bar{x}_{k,t}\right)\left(x_{it}- \bar{x}_{k,t}\right)' \\
        &=\frac{1}{|G_k^0(0,\boldsymbol{\gamma}_N(\underline{K},m))|T} \\ 
        &\qquad\times\left(\sum_{i\in G_k^0(0,\boldsymbol{\gamma}_N(\underline{K},m))} + \sum_{j\neq k}\sum_{i\in G_j^0\cap G_k(\boldsymbol{\gamma}_N(\underline{K},m))}\right)\sum_{t=1}^T\left(x_{it} - \bar{x}_{k\wedge k(0),t} + O_p\left(N^{\alpha_{m+1}-1}\right)\right) \\
        &\hspace{8.5cm}\times\left(x_{it} - \bar{x}_{k\wedge k(0),t} + O_p\left(N^{\alpha_{m+1}-1}\right)\right)' \\
        &=\frac{1}{|G_k^0(0,\boldsymbol{\gamma}_N(\underline{K},m))|T}\sum_{i\in G_k^0(0,\boldsymbol{\gamma}_N(\underline{K},m))}\sum_{t=1}^T\left(x_{it} - \bar{x}_{k\wedge k(0),t}\right)\left(x_{it} - \bar{x}_{k\wedge k(0),t}\right)' + O_p\left(N^{\alpha_{m+1}-1}\right),
        \label{eqn:denom_theta_hat}
    \end{align}
    and 
    \begin{align}
        &\frac{1}{|G_k^0(0,\boldsymbol{\gamma}_N(\underline{K},m))|T}\sum_{i\in G_k(\boldsymbol{\gamma}_N(\underline{K},m))}\sum_{t=1}^T\left(x_{it} - \bar{x}_{k,t}\right)\left(y_{it}- \bar{y}_{k,t}\right) \\
        &=\frac{1}{|G_k^0(0,\boldsymbol{\gamma}_N(\underline{K},m))|T}\sum_{i\in G_k^0(0,\boldsymbol{\gamma}_N(\underline{K},m))}\sum_{t=1}^T\left(x_{it} - \bar{x}_{k\wedge k(0),t}\right)\left(x_{it} - \bar{x}_{k\wedge k(0),t}\right)'\theta_k^0 \\
        &\quad + \frac{1}{|G_k^0(0,\boldsymbol{\gamma}_N(\underline{K},m))|T}\sum_{i\in G_k^0(0,\boldsymbol{\gamma}_N(\underline{K},m))}\sum_{t=1}^T\left(x_{it} - \bar{x}_{k\wedge k(0),t}\right)\left(\varepsilon_{it} - \bar{\varepsilon}_{k\wedge k(0),t}\right) + O_p\left(N^{\alpha_{m+1}-1}\right).
        \label{eqn:nume_theta_hat}
    \end{align}
    Combining \eqref{eqn:denom_theta_hat} and \eqref{eqn:nume_theta_hat}, and using Assumption \ref{asm:basic_gfe}(d'), we obtain $\widehat{\theta}_k^{(\underline{K})}(\boldsymbol{\gamma}_N(\underline{K},m))=\theta_k^0 + b_{1k,NT}^*$, where $b_{1k,NT}^* = O_p(\max\{(NT)^{-1/2},N^{\alpha_{m+1}-1}\})$.

    For $\widehat{\mu}_{kt}^{(\underline{K})}(\boldsymbol{\gamma}_N(\underline{K},m))$, we have
    \begin{align}
        \widehat{\mu}_{kt}^{(\underline{K})}(\boldsymbol{\gamma}_N(\underline{K},m)) 
        &= \bar{x}_{k\wedge k(0),t}'\theta_k^0 + \bar{\varepsilon}_{k\wedge k(0),t} + \mu_{kt}^0 - \bar{x}_{k\wedge k(0),t}'\left(\theta_k^0 + b_{ik,NT}^*\right) + O_p\left(N^{\alpha_{m+1}-1}\right) \\
        &=\mu_{kt}^0 + b_{2k,NT}^*,
    \end{align}
    where $b_{2k,NT}^* = O_p(\max\{N^{-1/2}, N^{\alpha_{m+1}-1}\})$. This completes the proof of part (i).

    \noindent(ii) Part (ii) is immediate in view of Assumptions \ref{asm:basic_gfe}(a'), (b) and (c).
\end{proof}

Following exactly the same argument as in the proof of Lemma \ref{lem:overfit}, we obtain the following result.

\begin{lem}
    Consider model \eqref{model:with_gfe}. Suppose Assumptions \ref{asm:group_size_gfe} and \ref{asm:basic_gfe}(a')-(d') and (f') hold. Then for any fixed $K\in\{K^0+1,K^0+2,\ldots\}$ and for any $\boldsymbol{\gamma}_N\in\overline{\Gamma}_N(K)$, 
    \begin{align}
        \left|\widehat{\sigma}^2(K,\boldsymbol{\gamma}_N(K)) - \widehat{\sigma}^2\left(K^0,\widehat{\boldsymbol{\gamma}}_N^{(K^0)}\right)\right|=O_p\left(N^{-1}\right).
    \end{align}
    \label{lem:overfit_gfe}
\end{lem}

\noindent\begin{proof}[Proof of Proposition \ref{prop:ic_inconsistency_gfe}]
    \noindent(i) As in the proof of Proposition \ref{prop:ic_inconsistency}(i), applying Lemma \ref{lem:overfit_gfe} yields
    \begin{align}
        &P\left(\mathrm{IC}\left(K^0,h_{NT}\right) < \min_{K^0+1\leq \overline{K}\leq K_{\max}}\mathrm{IC}\left(\overline{K},h_{NT}\right)\right) \\
        &\geq P\left(\min_{K^0+1\leq \overline{K}\leq K_{\max}}\left[\left(\overline{K}-K^0\right)(T+p)\widetilde{\sigma}^2h_{NT} +\left\{\widehat{\sigma}^2\left(K^*,\overline{\boldsymbol{\gamma}}_N\left(K^*\right)\right) - \widehat{\sigma}^2\left(K^0,\widehat{\boldsymbol{\gamma}}_N^{(K^0)}\right)\right\}\right]>0\right) \\
        &\stackrel{p}{\to}1
    \end{align}
    if $TNh_{NT}\to\infty$. This establishes part (i).

    \noindent(ii) From Theorem \ref{thm:consistency_normality_gfe}(iv), we have
    \begin{align}
        \mathrm{IC}\left(K^0,h_{NT}\right) 
        =\frac{1}{NT}\sum_{i=1}^N\sum_{t=1}^T\varepsilon_{it}^2 + n\left(K^0\right)\widetilde{\sigma}^2h_{NT} + O_p\left(N^{-1}\right).
        \label{eqn:ic_K^0_gfe}
    \end{align}
    Take any $\underline{K}\in\{m,\ldots,K^0-1\}$ and any $\boldsymbol{\gamma}_N(\underline{K},m)=(k_1,\ldots,k_N)\in\Gamma_N(\underline{K},m)$ for each $N$. Using Lemma \ref{lem:underfit_gfe}, we decompose $\widehat{\sigma}^2(\underline{K},\boldsymbol{\gamma}_N(\underline{K},m))$ as
    \begin{align}
        &\widehat{\sigma}^2(\underline{K},\boldsymbol{\gamma}_N(\underline{K},m)) \\
        &=\frac{1}{NT}\sum_{k=1}^m\sum_{i\in G_k(\boldsymbol{\gamma}_N(\underline{K},m))}\sum_{t=1}^T\left\{\varepsilon_{it} - x_{it}'\left(\widehat{\theta}_{k_i}^{(\underline{K})}(\boldsymbol{\gamma}_N(\underline{K},m)) - \theta_{k_i^0}^0\right) - \left(\widehat{\mu}_{k_it}^{(\underline{K})}(\boldsymbol{\gamma}_N(\underline{K},m)) - \mu_{k_i^0t}^0\right) \right\}^2 \\
        &\qquad + \frac{1}{NT}\sum_{k=m+1}^{\underline{K}}\sum_{i\in G_k(\boldsymbol{\gamma}_N(\underline{K},m))}\sum_{t=1}^T\left\{\varepsilon_{it} - x_{it}'\left(\widehat{\theta}_{k_i}^{(\underline{K})}(\boldsymbol{\gamma}_N(\underline{K},m)) - \theta_{k_i^0}^0\right) - \left(\widehat{\mu}_{k_it}^{(\underline{K})}(\boldsymbol{\gamma}_N(\underline{K},m)) - \mu_{k_i^0t}^0\right) \right\}^2 \\
        &=\frac{1}{NT}\sum_{k=1}^{m}\sum_{i\in G_k^0(0,\boldsymbol{\gamma}_N(\underline{K},m))}\sum_{t=1}^T\left(\varepsilon_{it}-x_{it}'b_{1k,NT}^* - b_{2k,NT}^*\right)^2 + O_p\left(N^{\alpha_{m+1}-1}\right) \\
        & = \frac{1}{NT}\sum_{i=1}^N\sum_{t=1}^T\varepsilon_{it}^2 + O_p\left(N^{\alpha_{m+1}-1}\right).
    \end{align}
    Consequently, noting that $\widehat{\sigma}^2(\underline{K},\widehat{\boldsymbol{\gamma}}_N^{(\underline{K})})\leq \widehat{\sigma}^2(\underline{K},\boldsymbol{\gamma}_N(\underline{K},m))$, we deduce
    \begin{align}
        \mathrm{IC}\left(\underline{K},h_{NT}\right) 
        &\leq \widehat{\sigma}^2(\underline{K},\boldsymbol{\gamma}_N(\underline{K},m)) + n(\underline{K})\widetilde{\sigma}^2h_{NT} \\
        &=\frac{1}{NT}\sum_{i=1}^N\sum_{t=1}^T\varepsilon_{it}^2 + n(\underline{K})\widetilde{\sigma}^2h_{NT} + O_p\left(N^{\alpha_{m+1}-1}\right).
        \label{bound:ic_underfit_gfe}
    \end{align}
    Using \eqref{eqn:ic_K^0_gfe} and \eqref{bound:ic_underfit_gfe}, we get
    \begin{align}
        P\left(\mathrm{IC}(\underline{K},h_{NT}) < \mathrm{IC}(K^0,h_{NT})\right)
        &\geq P\left(\frac{1}{NT}\sum_{i=1}^N\sum_{t=1}^T\varepsilon_{it}^2 + n(\underline{K})\widetilde{\sigma}^2h_{NT} + O_p\left(N^{\alpha_{m+1}-1}\right) \right.\\
        &\left. \hspace{1.5cm} < \frac{1}{NT}\sum_{i=1}^N\sum_{t=1}^T\varepsilon_{it}^2 + n\left(K^0\right)\widetilde{\sigma}^2h_{NT} + O_p\left(N^{-1}\right)\right) \\
        &= P\left((T+p)\left(K^0-\underline{K}\right)\widetilde{\sigma}^2h_{NT} + O_p\left(N^{\alpha_{m+1}-1}\right)>0\right) \\
        &\stackrel{p}{\to} 1
        \label{convergence:underest_gfe}
    \end{align}
    if $TN^{1-\alpha_{m+1}}h_{NT}\to\infty$. With this and part (i), it follows that $P(\widehat{K}(h_{NT})<K^0)\to 1$ if $TN^{1-\alpha_{m+1}}h_{NT}\to\infty$.
\end{proof}

\end{document}